\documentclass[preprint,12pt]{elsarticle}
\journal{Information and Computation} 
\pdfoutput=1
\usepackage{etex,etoolbox}
\usepackage[utf8]{inputenc}
\usepackage[english]{babel}
\usepackage{amsmath}
\usepackage{amssymb}
\usepackage{xspace}  % allows for indexgeneration
\usepackage{ifdraft}
\usepackage{environ}
\usepackage[colorlinks]{hyperref}
\usepackage{mathtools}
\usepackage{amsthm}
\usepackage{thmtools}
\usepackage{booktabs}

\usepackage{tikz-cd}
\usepackage{ifthen}

\usepackage{nicefrac}
\DeclareMathAlphabet{\mathpzc}{OT1}{pzc}{m}{normal} % for \mathpzc, especially \H

\usepackage{ifthen}
\usepackage{microtype}

\usepackage{enumitem}
\setlist[enumerate,1]{label=(\arabic*),font=\normalfont,align=left,leftmargin=0pt,labelindent=0pt,listparindent=\parindent,labelwidth=0pt,itemindent=!,topsep=3pt,parsep=0pt,itemsep=3pt,start=1}
\setlist[enumerate,2]{label=(\alph*),font=\normalfont,labelindent=*,leftmargin=*,start=1}
\setlist[itemize]{labelindent=*,leftmargin=*}
\setlist[description]{labelindent=*,leftmargin=*,itemindent=-1 em}

\usepackage{dsfont}

\ifnum 0=0
\newdimen\tightpaperwidth
\newdimen\tightpaperheight
\tightpaperwidth=\textwidth
\tightpaperheight=\textheight
\usepackage[marginparwidth=2cm,nohead,nofoot,
            headsep = 8mm, % gap between heading and text
            footskip = 1cm,
            text={\textwidth,\textheight},
            paperwidth=\tightpaperwidth,
            paperheight=\tightpaperheight,
            marginparsep=1mm,
            marginparwidth=1.4cm,
            footskip=1cm,
            centering
            ]{geometry}
\else\fi

\usepackage[final,notref,notcite]{showkeys}
\usepackage{seqsplit}
\usepackage{xstring}

\makeatletter
\renewcommand*\showkeyslabelformat[1]{%
\@ifundefined{hideNextShowKeysLabel}{%
\noexpandarg%
\StrSubstitute{#1}{ }{\textvisiblespace}[\TEMP]%
\parbox[t]{\marginparwidth}{\raggedright\normalfont\small\ttfamily\(\{\){\color{red!50!black}\expandafter\seqsplit\expandafter{\TEMP}}\(\}\)}%
}{}% end of \@ifundefined
}
\makeatother

\hypersetup{hidelinks}

\addto\extrasenglish{% only necessary if babel is used
}

\newenvironment{nicearray}[1]
    {\begin{array}{@{}#1@{}}\toprule\noalign{\smallskip}}
    {\\\bottomrule\end{array}}

\declaretheorem[name=Definition,style=definition,numberwithin=section]{definition}
\declaretheorem[name=Example,style=definition,sibling=definition]{example}

\declaretheorem[name=Remark,style=definition,sibling=definition]{remark}
\declaretheorem[name=Assumption,style=definition,sibling=definition]{assumption}

\declaretheorem[name=Notation,style=definition,sibling=definition]{notation}
\declaretheorem[name=Theorem,sibling=definition]{theorem}
\declaretheorem[sibling=definition]{corollary}

\declaretheorem[sibling=definition]{lemma}
\declaretheorem[sibling=lemma]{proposition}

\usepackage[nomargin,marginclue,footnote,status=final]{fixme}
\fxusetargetlayout{color}
\FXRegisterAuthor{sm}{asm}{SM}%Stefan
\FXRegisterAuthor{tw}{atw}{TW}%Thorsten

\newcommand{\twnotei}[1]{\twnote[nofootnote,nomargin,inline]{#1}}

\newcommand{\FinalCoalgDagger}{{\ddagger}}
\newcommand{\LFFDagger}{{\dagger}}
\newcommand{\EquationDagger}{{\dagger}}

\newcommand{\textqt}[1]{``#1''}
\newcommand{\op}[1]{\operatorname{\textsf{\upshape {#1}}}}
\newcommand{\A}{\ensuremath{\mathcal{A}}\xspace}
\newcommand{\B}{\ensuremath{\mathcal{B}}\xspace}
\newcommand{\C}{\ensuremath{\mathcal{C}}\xspace}
\newcommand{\Cfp}{\ensuremath{\mathcal{C}_\mathsf{fp}}\xspace}
\newcommand{\D}{\ensuremath{\mathcal{D}}\xspace}
\newcommand{\Fun}{\ensuremath{\mathsf{Fun}}\xspace}
\newcommand{\Coalg}{\ensuremath{\mathsf{Coalg}}\xspace}
\newcommand{\Coalgfg}{\ensuremath{\mathsf{Coalg}_{\mathsf{fg}}}\xspace}
\newcommand{\Coalglfg}{\ensuremath{\mathsf{Coalg}_{\mathsf{lfg}}}\xspace}
\newcommand{\Coalgfp}{\ensuremath{\mathsf{Coalg}_{\mathsf{fp}}}\xspace}

\newcommand{\Set}{\ensuremath{\mathsf{Set}}\xspace}

\newcommand{\Mon}{\ensuremath{\mathsf{Mon}}\xspace}
\renewcommand{\H}{\ensuremath{\mathpzc{H}}\xspace}
\newcommand{\Hf}{\ensuremath{\mathpzc{H}_\textnormal{\,\sffamily f}}\xspace}
\newcommand{\Funf}{\ensuremath{\mathsf{Fun}_\textnormal{\sffamily f}}\xspace}
\newcommand{\Mndc}{\ensuremath{\mathsf{Mnd}_\textnormal{\sffamily c}}\xspace}
\newcommand{\Mndf}{\ensuremath{\mathsf{Mnd}_\textnormal{\sffamily f}}\xspace}
\newcommand{\EQ}[1][]{\ensuremath{\mathsf{EQ}_\textnormal{#1}}\xspace}
\newcommand{\id}{\ensuremath{\textnormal{id}}\xspace}
\newcommand{\inj}{\ensuremath{\mathsf{in}}}
\newcommand{\inl}{\ensuremath{\mathsf{inl}}}
\newcommand{\inr}{\ensuremath{\mathsf{inr}}}
\newcommand{\Id}{\ensuremath{\textnormal{Id}}\xspace}
\newcommand{\Pot}{\ensuremath{{\mathcal{P}}}\xspace}
\newcommand{\Potf}{\ensuremath{{\mathcal{P}_\textnormal{\sffamily f}}}\xspace}
\newcommand{\lff}{\ensuremath{\ell}}

\newcommand{\FPS}[2][S]{\ensuremath{#1\llangle#2\rrangle}}
\newcommand{\Poly}[2][S]{\ensuremath{#1\langle#2\rangle}}
\newcommand{\colim}{\ensuremath{\operatorname{colim}}\xspace}
\renewcommand{\o}{\ensuremath{\cdot}}
\newcommand{\fgiterative}{fg-iterative\xspace}
\newcommand{\LFF}{\vartheta}
\newcommand{\fpair}[1]{\ensuremath{\langle #1 \rangle}}
\newcommand{\fuse}[1]{\ensuremath{[#1]}}
\newcommand{\N}{\ensuremath{\mathds{N}}}
\newcommand{\Z}{\ensuremath{\mathds{Z}}}
\newcommand{\Q}{\ensuremath{\mathds{Q}}}
\newcommand{\fp}{{\textsf{\upshape{fp}}}}
\newcommand{\lfp}{{\textsf{\upshape{lfp}}}}
\renewcommand{\Im}{\op{Im}}

\makeatletter
\newsavebox{\@brx}
\newcommand{\llangle}[1][]{\savebox{\@brx}{\(\m@th{#1\langle}\)}%
  \mathopen{\copy\@brx\kern-0.5\wd\@brx\usebox{\@brx}}}
\newcommand{\rrangle}[1][]{\savebox{\@brx}{\(\m@th{#1\rangle}\)}%
  \mathclose{\copy\@brx\kern-0.5\wd\@brx\usebox{\@brx}}}
\makeatother

\usepackage{newunicodechar}
\newunicodechar{×}{\ensuremath{\times}}
\newunicodechar{≤}{\ensuremath{\le}}
\newunicodechar{→}{\ensuremath{\to}}
\newunicodechar{⊗}{\ensuremath{\otimes}}
\newunicodechar{〈}{\ensuremath{\langle}}
\newunicodechar{〉}{\ensuremath{\rangle}}

\tikzset{
    every diagram/.style={
        row sep=1cm,
        column sep=1cm,
    }
}
\newcommand{\descto}[3][]{
    \arrow[draw=none]{#2}[anchor=center,#1]{#3}
}
\tikzset{oldequal/.style={
    double equal sign distance,
    -,
}}
\tikzset{shiftarr/.style={
        rounded corners,%
        to path={--([#1]\tikztostart.center)
                     -- ([#1]\tikztotarget.center) \tikztonodes
                     -- (\tikztotarget)},
}}
\tikzset{commutative diagrams/diagrams={
    rounded corners,
}}
\tikzstyle{mathnodes}=[
  nodes={%
   execute at begin node=\(,%
   execute at end node=\)%
  }%
]
\tikzstyle{anchorcenter}=[
    baseline={([yshift=-.5ex]current bounding box.center)}
]
\tikzstyle{lambdatree}=[
    level distance = 4mm,
    anchorcenter,
    level 1/.style={sibling distance=13mm},
    every node/.style={
        inner sep=1pt,
    },
    edge from parent/.style= {
        draw,
        edge from parent path={(\tikzparentnode) -- (\tikzchildnode.north)},
    },
    noedge/.style={
        edge from parent/.style={}
    },
]
\tikzstyle{level}=[
    sibling distance=20mm/#1,
]

\makeatletter
\newbox\xrat@below
\newbox\xrat@above
\newcommand{\xrightarrowtail}[2][]{%
  \setbox\xrat@below=\hbox{\ensuremath{\scriptstyle #1}}%
  \setbox\xrat@above=\hbox{\ensuremath{\scriptstyle #2}}%
  \pgfmathsetlengthmacro{\xrat@len}{max(\wd\xrat@below,\wd\xrat@above)+.6em}%
  \mathrel{\tikz [>->,baseline=-.75ex]
                 \draw (0,0) -- node[below=-2pt] {\box\xrat@below}
                                node[above=-2pt] {\box\xrat@above}
                       (\xrat@len,0) ;}}
\makeatother

\newcommand{\takeout}[1]{\empty}
\newcommand{\Sd}{S}
\def\epito{\twoheadrightarrow}
\def\monoto{\rightarrowtail}
\def\subto{\hookrightarrow}
\newcommand{\Rbeheq}{\mathbin{\mathord{\sim}R\mathord{\sim}}}

\title{A New Foundation for Finitary Corecursion and Iterative Algebras}
\author[fau]{Stefan Milius\fnref{dfg}}
\ead{stefan.milius@fau.de}
\author[anu]{Dirk Pattinson}
\ead{dirk.pattinson@anu.edu.au}
\author[fau]{Thorsten Wißmann\fnref{dfg}}
\ead{thorsten.wissmann@fau.de}
\fntext[dfg]{Supported by Deutsche Forschungsgemeinschaft (DFG) under project MI~717/5-2}
\address[fau]{Friedrich-Alexander-Universität Erlangen-Nürnberg}
\address[anu]{The Australian National University}

\begin{document}
\begin{frontmatter}
\begin{abstract}%
  This paper contributes to a theory of the behaviour of ``finite-state'' systems that is generic in the system type. We propose that such systems are modelled as coalgebras with a finitely generated carrier for an endofunctor on a locally finitely presentable category. Their behaviour gives rise to a new fixpoint of the coalgebraic type functor called \emph{locally finite fixpoint} (LFF). We prove that if the given endofunctor is finitary and preserves monomorphisms then the LFF always exists and is a subcoalgebra of the final coalgebra (unlike the rational fixpoint previously studied by \citeauthor*{iterativealgebras}). Moreover, we show that the LFF is characterized by two universal properties: (1)~as the final locally finitely generated coalgebra, and (2)~as the initial fg-iterative algebra. As instances of the LFF we first obtain the known instances of the rational fixpoint, e.g. regular languages, rational streams and formal power-series, regular trees etc. Moreover, we obtain a number of new examples, e.g.~(realtime deterministic resp.~non-deterministic) context-free languages, constructively $S$-algebraic formal power-series (in general, the behaviour of finite coalgebras under the coalgebraic language semantics arising from the generalized powerset construction by \citeauthor*{sbbr13}), and the monad of Courcelle's algebraic trees.%
\end{abstract}
\end{frontmatter}

\section{Introduction}

Coalgebras capture many types of state based system within a uniform and
mathematically rich framework \cite{Rutten:2000:UCT:abbrev}. One outstanding
feature of the general theory is \emph{final semantics} which gives
% (in the sense of denotational semantics)
a fully abstract account of system behaviour, i.e.~it provides
precisely all the behavioural equivalence classes.  For example, the
coalgebraic modelling of deterministic automata (without a finiteness
restriction on state sets) yields the set of all formal languages as a
final model, and restricting to \emph{finite} automata one precisely
obtains the regular languages \cite{Rutten:1998:ACE}. This
correspondence has been generalized to locally finitely presentable
categories~\cite{adamek1994locally,gu71}, where \emph{finitely
  presentable} objects play the role of finite sets, leading to the
notion of \emph{rational fixpoint} that provides final semantics to
all models with finitely presentable carrier~\cite{streamcircuits}. It
is known that the rational fixpoint is fully abstract for these models
as long as finitely presentable objects agree with finitely generated
objects in the base category~\cite[Proposition~3.12]{bms13}. While
this is the case in some categories (e.g. sets, posets, graphs, vector
spaces, commutative monoids), it is currently unknown in other base
categories that are used in the construction of system models, for
example in idempotent semirings (used in the treatment of context-free
grammars \cite{coalgcontextfree}), in algebras for the stack monad
(used for modelling configurations of stack
machines~\cite{coalgchomsky}); or it even fails, for example in the
category of finitary monads on sets (used in the categorical study of
algebraic trees \cite{secondordermonad}), or Eilenberg-Moore
categories for a monad in general (the target category of generalized
determinization \cite{sbbr13}, in which the above examples
live). Coalgebras over a category of Eilenberg-Moore algebras over
\Set in particular provide a paradigmatic setting: automata that
describe languages beyond the class of regular languages consist of a
finite state set, but their transitions produce side effects such as
the manipulation of a stack.  These can be described by a monad, so
that the (infinite) set of system configurations (machine states plus
stack content) is described by a free algebra (for that monad) that is
generated by the finite set of machine states.  This is formalized by
the generalized powerset construction \cite{sbbr13} and interacts
nicely with the coalgebraic framework we present.

Technically, the shortcoming of the rational fixpoint is due to the fact that
finitely presentable objects are not closed under quotients, so that the
rational fixpoint itself may fail to be a subcoalgebra of the final coalgebra
and so does not identify all behaviourally equivalent states. The main
conceptual contribution of this paper is the insight that also in cases where
finitely presentable and finitely generated do not agree, we have a canonical
domain for finitely generated behaviour. We introduce the \emph{locally finite fixpoint}
which provides a fully abstract model for such behaviour. We
support this claim both by general results and concrete examples: we show that
under mild assumptions, the locally finite fixpoint always exists, and we give a
coalgebraic construction of it (\autoref{thm:final}); we also prove that it is
indeed a subcoalgebra of the final coalgebra (\autoref{lffSubcoalg}). Moreover,
we give a characterization of the locally finite fixpoint as the initial
fg-iterative algebra (\autoref{lffInitialIter}). We then instantiate our results
to several scenarios studied in the literature.

First, we show that the locally finite fixpoint is universal (and fully
abstract) for the class of systems produced by the generalized powerset
construction over \Set: every determinized finite-state system
induces a
unique homomorphism to the locally finite fixpoint, and the latter contains precisely
the finite-state behaviours (\autoref{prop:LFFunion}).

Applied to the coalgebraic treatment of context-free languages, we
show that the locally finite fixpoint yields precisely the
context-free languages (\autoref{stackContextFree}), and real-time
deterministic context-free languages (\autoref{stackRealTime}),
respectively, when their accepting machines are modelled as coalgebras
over the category of algebras for the stack monad of
\cite{coalgchomsky}. For context-free languages weighted in a semiring
$S$, or equivalently for constructively $S$-algebraic power series
\cite{Petre:2009:ASP}, the locally finite fixpoint comprises precisely
those (\autoref{lffPowerSeries}), by phrasing the results of
\citet{jcssContextFree} in terms of the generalized powerset
construction.

Our last example shows the applicability of our results to  
Eilenberg-Moore algebras over categories beyond \Set, and we characterize the monad of Courcelle's
algebraic trees over a signature \cite{courcelle,secondordermonad} as the
locally finite fixpoint of an associated functor (on a category of monads)
(\autoref{lffAlgTree}), solving an open problem in~\cite{secondordermonad}.

The work extends the conference paper~\cite{mpw16}. The present paper
is a completely reworked version containing detailed proofs of all our
results. In addition, Section~\ref{sec:iterative} on fg-iterative
algebras is new.

\paragraph{Related Work} The characterization of languages in
terms of (co-)algebraic constructions has been carried out for
various examples, such as (weighted) context-free languages
\cite{Winter:2013:CCC,coalgchomsky} as well as regular languages
\cite{Rutten:1998:ACE} where characterization theorems were
established on a case-by-case basis. We show that the locally finite
fixpoint provides a more general, and conceptual account. 
We have already mentioned the
rational fixpoint~\cite{iterativealgebras,streamcircuits} that
serves a similar purpose and shares many technical similarities with
the locally finite fixpoint, introduced here. Many of the properties
of the rational fixpoint in fact hold, \emph{mutatis
mutandis}, also for locally finite fixpoint, cf.~\emph{op.cit}.

\paragraph{Outline of the paper} The rest of this paper is structured
as follows. In \autoref{sec:prelim} we recall a few basic facts about
the central notions of this paper: locally finitely presentable
categories, coalgebras, and the rational fixpoint of an
endofunctor. Next, in \autoref{sec:locfp} we introduce locally
finitely generated (lfg) coalgebras, and we prove that a final lfg
coalgebra exists, is a fixpoint (called locally finite fixpoint) and a
subcoalgebra of the final coalgebra. The new \autoref{sec:iterative}
provides a characterization of the locally finite fixpoint as an
algebra: it is the initial fg-iterative algebra. Then in
\autoref{sec:relrat} we investigate the relationship of the locally
finite fixpoint to the rational fixpoint. Under slightly stronger
assumptions than before we prove that the locally finite fixpoint is
the image of the rational fixpoint in the final coalgebra. Finally, in
\autoref{sec:app} we consider several examples of the locally finite
fixpoint, and \autoref{sec:con} discusses future work and concludes
the paper.

\paragraph{Acknowledgments} We are grateful to Henning Urbat for
pointing out a mistake in a previous version of the proof of
\autoref{LFFisIterative}, and to him and Ji\v{r}\'i Ad\'amek for many
helpful discussions.\smnote{Do we need to mention Joost Winter?}

\section{Preliminaries and Notation}
\label{sec:prelim}

In this section we briefly recall a number of technical preliminaries
needed throughout the paper. We assume that readers are familiar with
basic category theory and with algebras and coalgebras for
endofunctors. 

\subsection{Eilenberg-Moore-categories}

Given a monad $T\colon \C\to \C$, its Eilenberg-Moore category $\C^T$
is the category whose objects are the algebras for the monad $T$,
i.e.~pairs $(A,a)$ where $A$ is an object of $\C$ (the \emph{carrier}
of the algebra) and $a\colon TA\to A$ a morphism (the \emph{structure}
of the algebra) such that $a \cdot \eta_A = \id_A$ and
$a \cdot Ta = a\cdot \mu_A$, where $\eta: \Id \to T$ and
$\mu: TT \to T$ are the unit and multiplication of the monad $T$.
Morphisms of $T$-algebras are morphisms of $\C$ commuting with algebra structures. More precisely, a $T$-algebra morphism from $(A,a)$ to $(B,b)$ is a morphism $f: A \to B$ of $\C$ such that $f \cdot a = b \cdot Tf$. 
See \citeauthor{awodey2010category}~\cite[Chapter 10]{awodey2010category}
for a more detailed introduction.

\emph{Liftings} are a common way to define endofunctors on $\C^T$. Given a
functor on the base category $H\colon \C\to \C$, a lifting of $H$ is a functor
$H^T\colon \C^T\to \C^T$ such that the square below commutes, where $U: \C^T \to
\C$ denotes the forgetful functor:
\[
    \begin{tikzcd}
      \C^T
      \arrow{d}[swap]{U}
      \arrow{r}{H^T}
      & \C^T
      \arrow{d}{U}
      \\
      \C
      \arrow{r}{H}
      & \C
  \end{tikzcd}
\]
Recall that the forgetful functor $U$ has a left adjoint given by
assigning to an object $X$ of $\C$ the free Eilenberg-Moore algebra
$(TX, \mu_X)$.

The examples of Eilenberg-Moore-categories over \Set include groups,
monoids, (idempotent) semirings, \Set itself, and moreover any variety
of (finitary) algebras, i.e.~a class of algebras specified by 
(finitary) operations and equations.

Note that monos is $\Set^T$ are precisely the injective $T$-algebra
morphisms. However, epis need not be surjective in $\Set^T$; for
example the embedding $\Z \subto \Q$ from the integers to the rationals,
each considered as a monoid w.r.t.~multiplication, is an epi in the
category of monoids. The surjective $T$-algebra morphisms are
precisely the strong epis.\footnote{In $\Set^T$ the classes of
  strong and regular epimorphisms coincide.}

Recall that, in general, an epi $e\colon X\twoheadrightarrow Y$ is called
\emph{strong}, if for every mono $m\colon A \rightarrowtail B$ and morphisms
$f\colon X\to A$, $g\colon Y\to B$ with $g\cdot e=m\cdot f$, there exists a unique
\emph{diagonal fill-in}, i.e.~a unique $d\colon Y\to A$ such that:
\[
  \begin{tikzcd}
    X
    \arrow{d}[swap]{f}
    \arrow[->>]{r}{e}
    & Y \arrow{d}{g}
    \arrow[dashed]{dl}[description]{\exists ! d}
    \\
    A \arrow[>->]{r}{m}
    & B
  \end{tikzcd}
\]
Observe that for strong epis we have the same cancellation law as for
ordinary epis: 
\[
  \text{if $e' \cdot e$ is a strong epi, then $e'$ is a strong epi.}
\]
We will continue to denote monos and strong epis in a category by
$\monoto$ and $\epito$, respectively.

The coproduct of a family $(X_i)_{i \in I}$ of objects is denoted by
$\big(\inj_i\colon X_i \to \coprod_{i \in I} X_i\big)_{i\in I}$, and we call the
morphism $\inj_i$ the \emph{coproduct injections}. Furthermore, for a
family of morphisms $(f_i\colon X_i \to Y)_{i\in I}$, we denote by
\[
  [f_i]_{i\in I}\colon \coprod_{i\in I} X_i \to Y
\]
the unique morphism with $[f_i]_{i \in I} \cdot \inj_i = f_i$ for
every $i$. In the case of binary coproducts we write
\[
  X \xrightarrow{\inl} X_1 + X_2 \xleftarrow{\inr} X_2
  \qquad
  \text{and}
  \qquad
  X_1+X_2 \xrightarrow{[f_1,f_2]} Y.
\]
Similarly, for every colimit cocone $(c_i\colon C_i \to C)_{i \in I}$ we call the
morphisms $c_i$ the \emph{injections} of the colimit (even though they
are not injective maps, in general).

\begin{example}[label = colimitStrongEpi]
  Suppose that $\C$ is a cocomplete category. 
  For every diagram $D\colon \D\to \C$, the injections of the colimit
  cocone $(d_i\colon Di \to \colim D)_{i\in \D}$ yield the strong epi
  \[
    [d_i]_{i\in \D}\colon \coprod_{i\in \D} Di \epito \colim D.
  \]
\end{example}

\subsection{Locally finitely presentable categories} 

A \emph{filtered colimit} is the colimit of a diagram $\D \to \C$
where $\D$ is a filtered category (i.e.~every finite subcategory
$\D_0 \subto \D$ has a cocone in $\D$), and a \emph{directed colimit}
is a colimit of a diagram having a directed poset as its diagram
scheme $\D$. \emph{Finitary functors} preserve filtered (equivalently
directed) colimits. An object $C \in \C$ is called \emph{finitely
  presentable} (fp) if its hom-functor $\C(C, -)$ is finitary and
\emph{finitely generated} (fg) if $\C(C, -)$ preserves directed
colimits of monos (i.e.~all connecting morphisms in $\C$ are
monic). Clearly every fp object is fg, but not conversely in
general. Moreover, fg objects are closed under strong quotients; here,
a strong quotient of an object $X$ is represented by a strong epi
$X \twoheadrightarrow Y$, and closure under strong quotients means
that $Y$ is fg whenever $X$ is. For fp objects this fails in general.

A cocomplete category is called \emph{locally finitely presentable} (lfp) if
the full subcategory $\C_\fp$ of finitely presentable objects is
essentially small, i.e.~is up to isomorphism only a set, and every
object $C \in \C$ is a filtered colimit of a diagram in $\C_\fp$. We
refer to~\cite{gu71,adamek1994locally} for further details.

It is well known that the categories of sets, posets and graphs are
lfp with finitely presentable objects precisely the finite sets,
posets, graphs, respectively.  The category of vector spaces is lfp
with finite-dimensional spaces being the fp objects. Every finitary variety is
lfp.  The finitely generated objects are precisely the finitely
generated algebras, i.e.~those algebras having a finite set of
generators, and finitely presentable objects are precisely those
algebras specified by finitely many generators and relations. This
includes the categories of groups, monoids, (idempotent) semirings,
semi-modules, etc. More generally, for every finitary monad $T$,
i.e.~the underlying functor of $T$ is finitary, on the lfp category
$\C$, the Eilenberg-Moore category $\C^T$ is lfp
again~\cite[Remark~2.78]{adamek1994locally}.

Every lfp category has (strong epi,mono)-factorizations of
morphisms~\cite[Proposition 1.16]{adamek1994locally}, i.e.~every
morphism $f\colon A \to B$ factorizes as $f = m \o e$ for some mono
$m\colon \Im(f) \rightarrowtail B$ and strong epi
$e\colon A \twoheadrightarrow \Im(f)$. We call the subobject of $B$
represented by $m$ the \emph{image} of $f$.

\takeout{
  \twnotei{What's the point of this item? In my opinion, we are not
  interested in a sufficient criterion for fp=fg. Furthermore, we look at
  modules, but not necessarily of notherian rings. So I vote for dropping
  without substitution.} Modules for a Noetherian semiring. Recall that a (semi-)module for a
  semiring $\Sd$ is a commutative monoid $(M,+,0)$ together with an action of
  the semiring $\Sd$ on $M$ satisfying the usual distributive laws $r(m+n) =
  rm+rn$ and $r0 = 0$. Hence, modules for $\Sd$ form a finitary variety. In
  general the classes of fg modules and fp modules do not coincide. The semiring
  $\Sd$ is called \emph{Noetherian} if any submodule of a finitely generated
  module is itself finitely generated. For Noetherian semirings the classes of
  finitely generated and finitely presentable modules coincide (see
  e.g.~\cite[Prop.~2.6]{bms13} for a proof). There are also non-Noetherian
  semirings for which fp and fg modules coincide; e.g.~the module for the
  semiring of natural numbers for which modules are precisely the commutative
  monoids. 
}

We will subsequently make use of the following technical lemma. Recall
that a union of subobjects of some object $B$ is their join in the
poset of all subobjects of $B$. In an lfp category, a \emph{directed
  union} is, equivalently, a directed colimit of monos (see
  e.g.~\cite[Lemma~2.3]{amsw19functor}). 
\begin{lemma}[Ad\'amek, Milius, Sousa, and Wi\ss\/mann~{\cite[Lemma
2.9]{amsw19functor}}]\label{unionsofimages}%\hfill
  \leavevmode\newline
  Images of filtered colimits in the lfp category $\C$ are directed
  unions of images.
\end{lemma}
More precisely, suppose we have a filtered diagram $D\colon \D\to \C$
with a colimit cocone $(c_i\colon Di \to C)_{i\in \D}$ and a morphism
$f\colon C\to B$. Then the image of $f$ together with the induced
monomorphisms $d_i$ forms the directed union of the images of the
$f\cdot c_i$: 
\begin{equation}\label{equnionofimages}
  \begin{tikzcd}
    Di \arrow[->>]{rr}{e_i} \arrow{d}[left]{c_i} & {} &
    \Im(f \cdot c_i) \arrow[dashed]{dl}[above left]{d_i} 
    \arrow[>->]{d}[right]{m_i} \\
    C \arrow[->>]{r}{e}
    \arrow[shiftarr={yshift=-3ex}]{rr}[below]{f}
    \arrow[->>]{r}{e}
    & \Im(f) \arrow[>->]{r}{m} & B
  \end{tikzcd}
\end{equation}
% Furthermore, if the diagram $D$ is filtered, then the $A_i$ form a
% directed diagram.
%
\takeout{% this is not needed as we do not consider any diagrams only
         % closed under finite coproducts
\begin{remark}\label{rem:colimepi}
  Note that for the previous proof we do not need the full strength of
  filteredness of the given diagram $D$. In fact, in order to
  establish that the $A_i$ form a directed diagram it suffices that
  for each pair of object $i, j$ in $\D$ there exists an object
  $k$ in $\D$ and morphisms $i \to k$ and $j \to k$. 
\end{remark}
}%end takeout

\takeout{% this is proven in the paper with Jirka and Lurdes
We now establish that the well known characterization of fg algebras
being precisely the quotients of a free algebra on a finite set hold
more generally for every Eilenberg-Moore category over an lfp category
$\C$.
\begin{proposition}\label{prop:fgalgs}
  Let $\C$ be an lfp category and $T\colon\C \to \C$ a finitary monad. Then
  every fg $T$-algebra $(A,a)$ is a strong quotient of a free algebra 
  $(TX,\mu_X)$ where $X$ is fp in $\C$.
\end{proposition}
\begin{proof}
  Let $a\colon TA \to A$ be an algebra for the finitary monad $T$ on the lfp
  category $\C$. First observe that its algebra structure $a$ is a strong
  epi in $\C^T$; in fact, it follows from the
  axioms of Eilenberg-Moore algebras that $a$ is a $T$-algebra
  morphism and a coequalizer of $Ta$ and $\mu_A$.

  Recall that we can write $A$ as the filtered colimit of the
  canonical diagram $\Cfp/A \to \C$ with colimit injections all
  morphisms $f_i\colon X_i \to A$ in $\C$ with $X_i$ fp. The left-adjoint
  $\C\to \C^T$, which maps $X$ to the free $T$-algebra $(TX,\mu_X)$,
  preserves colimits. Thus, the $Tf_i\colon TX_i \to TA$ form a filtered
  colimit in $\C^T$. Now we take the (strong epi, mono) factorizations
  of the $a \cdot Tm_i\colon TX_i \to A$ in $\C^T$ to obtain a directed
  diagram of subalgebras of $A$:
  \[
    \begin{tikzcd}
      TX_i
      \arrow{d}[left]{Tm_i}\arrow[->>]{r}
      & A_i \arrow[>->]{d}{a_i}
      \\
      TA \arrow[->>]{r}{a}
      & 
      A
    \end{tikzcd}
  \]
  and we can apply \autoref{unionsofimages} in $\C^T$ to the morphism
  $a$ to see that $A$ is the colimit of this diagram. Since $A$ is fg,
  there exists some $i$ such $\id_A$ factors through one of the
  colimit injections $a_i\colon A_i \to A$, i.e.~there is some $f\colon A \to
  A_i$ with $a_i \cdot f = \id_A$. Thus $a_i$ is a split epi as well
  as a mono and therefore an iso, and we see that $A \cong A_i$ is
  strong quotient of $TX_i$.   
\end{proof}
}% end takeout

\takeout{% this has been moved to the paper with Jirka and Lurdes
We conclude this section with a brief discussion of finitary
adjunctions. Let $\A$ and $\B$ be lfp categories. An adjunction $L
\dashv R\colon \B \to \A$ is called \emph{finitary} if the right adjoint $R$ is
finitary (equivalently, the ensuing monad $T = RL$ on $\A$ is finitary). The
left-adjoint is, of course, always finitary.\smnote{It's surprising
  that the following shouldn't be somewhere in the literature. I can't
  see it in~\cite{adamek1994locally}; maybe Gabriel/Ulmer or Makkai/Par\'e?}
\begin{lemma}\label{lem:fppres}
  Let $L \dashv R\colon \B \to \A$ be a finitary adjunction between the lfp
  categories $\B$ and $\A$. Then we have:
  \begin{enumerate}
  \item $L$ preserves fp and fg objects, 
  \item if $L$ is fully faithful (i.e.~$\A$ may be regarded as a
    coreflective subcategory of $\B$) then $A$ is fp in $\A$ if and
    only if $LA$ is fp in $\B$,
  \item if, moreover $L$ preserves monos, then $A$ is fg in $\A$ if and
    only if $LA$ is fg in $\B$.
  \end{enumerate}
\end{lemma}
\begin{proof}
  \begin{enumerate}
  \item Let $X$ be an fp object of $\A$ and let $D\colon \D \to \B$ be a
    filtered diagram. Then we have the following chain of natural
    isomorphisms
    \begin{align*}
      \B(LX, \colim D) 
      & \cong \A(X, R(\colim D))\\
      & \cong \A(X, \colim RD) \\
      & \cong \colim(\A(X, RD(-)) \\
      & \cong \colim(\B(LX, D(-)).
    \end{align*}
    This shows that $LX$ is fp in $\B$. Now if $X$ is fg in $\A$ and
    $D$ is a directed diagram of monos, then $RD$ is also a directed diagram of
    monos (since the right-adjoint $R$ preserves monos). Thus, the
    same reasoning proves $LX$ to be fg in $\B$. 
  \item Suppose that $LX$ is fp in $\B$ and that $D\colon \D \to \A$ is a
    filtered diagram. Then we have the following chain of natural
    isomorphisms:
    \begin{align*}
      \A(X,\colim D) & \cong \B(LX, L(\colim \D)) \\
      & \cong \B(LX, \colim LD )\\
      & \cong \colim (\B(LX, LD(-)) \\
      & \cong \colim (\A(X,D(-))
    \end{align*}
    Indeed, the first and last step use that $L$ is fully faithful,
    the second step that $L$ is finitary and the remaining one that
    $X$ is fp in $\A$. 
  \item If $LX$ is fg in $\B$ and $D\colon \D \to \A$ a directed diagram of
    monos, then so is $LD$ since $L$ preserves monos by
    assumption. Thus the same reasoning as in point 2.~shows that $X$
    is fg in $\A$. \qedhere
  \end{enumerate}
\end{proof}
}% end takeout

\subsection{Coalgebras} 
\label{sec:coalgs}

Let $H\colon \C \to \C$ be an endofunctor. An \emph{$H$-coalgebra} is a
pair $(C, c)$, where $C$ is an object of $\C$ called the
\emph{carrier} and $c\colon C \to HC$ is a morphism called the
\emph{structure} of the coalgebra. A homomorphisms
$f\colon (C, c) \to (D, d)$ is a morphism $f\colon C \to D$ of $\C$ such that the following square commutes:
\[
    \begin{tikzcd}
      C \arrow{r}{c}
      \arrow{d}[swap]{f}
      & HC
      \arrow{d}{Hf}
      \\
      D \arrow{r}{d}
      & HD
    \end{tikzcd}
\]

Coalgebras and homomorphisms form a category,
which we denote by $\Coalg H$. 

If this category has a final object, then this final $H$-coalgebra is
denoted by 
\[
  \tau\colon \nu H \to H(\nu H).
\]
 The final coalgebra exists provided $H$ is
a finitary endofunctor on the lfp category $\C$ (see e.g.~\cite[Theorem~6.10]{AdamekEA17}). 

By the universal property, we have for every coalgebra $(C,c)$ a
unique homomorphism
$c^\FinalCoalgDagger\colon (C,c) \to (\nu H, \tau)$.\smnote{We should
  consistently use $\FinalCoalgDagger$ for the final coalgebra, $\LFFDagger$ for
  the LFF, and $\sharp$ for extensions to free $T$-algebras. But we
  also use $\EquationDagger$ for iterative algebras. Not sure this is a good
  idea.}  By Lambek's Lemma~\cite{lambek}, $\nu H$ is a fixpoint of
$H$ (i.e.~$\tau$ is an isomorphism). The final coalgebra represents a
canonical domain of behaviour of systems of type $H$, and the unique
homomorphism $c^\FinalCoalgDagger$ provides the semantics for a system
$(C,c)$. For a concrete category $\C$, i.e.~$\C$ is equipped with a
faithful functor $\mathopen|-\mathclose|\colon \C \to \Set$, we obtain
a notion of semantic equivalence called \emph{behavioural
  equivalence}: given two coalgebras $(C,c)$ and $(D,d)$, two states
$x \in |C|$ and $y \in |D|$ are called behavioural equivalent
(notation: $x \sim y$) whenever $|c^\FinalCoalgDagger(x)| = |d^\FinalCoalgDagger(y)|$.

Next, we recall a few categorical properties of $\Coalg H$. The
forgetful functor $\Coalg H \to \C$ creates colimits and reflects
monos and epis. A morphism $f$ in $\Coalg H$ is \emph{mono-carried}
(resp.~\emph{strong epi-carried}) if the underlying morphism in $\C$
is monic (resp.~a strong epi).  A \emph{directed union of coalgebras}
is a directed colimit of a diagram in $\Coalg H$ whose connecting
morphisms are mono-carried. Furthermore, by a \emph{subcoalgebra} we
mean a subobject in $\Coalg H$ represented by a mono-carried
homomorphism, and a \emph{quotient coalgebra} is represented by a
strong epi-carried homomorphism $(C,c) \twoheadrightarrow (D,d)$, and
we say that $(D,d)$ is a quotient of $(C,c)$.

\subsection{Non-empty Monos}

Recall that endofunctors on $\Set$ preserve all non-empty monomorphisms (because
they are split monos in $\Set$). We will assume a similar property for functors
on general lfp categories. Recall that an initial object $0$ is called
\emph{strict}, if every morphism $I\to 0$ is an isomorphism.
\begin{definition}
  A monomorphism $m\colon X\to Y$ is called \emph{empty} if its domain $X$ is a strict
  initial object.
\end{definition}
That means if $\C$ has no initial object or a non-strict one, all
monos are non-empty. Among the categories that have a strict initial
object are: sets, posets, graphs, topological spaces, nominal sets and
all Grothendieck toposes. In fact, every extensive category in the sense
of Carboni, Lack and Walters~\cite{clw93} has a strict initial object. Categories of
algebras (over \Set) have a strict initial object if the empty set
carries an algebra.

\begin{lemma}
  Let $T\colon \C \to \C$ be a monad and let $H\colon \C \to \C$ be an
  endofunctor. Then if $H$ preserves non-empty monos so does every
  lifting $H^T\colon \C^T\to\C^T$.
\end{lemma}
\begin{proof}
  The right-adjoint $U\colon \C^T\to \C$ preserves and reflects
  monos. Given a non-empty monomorphism $m\colon (A,a)\to (B,b)$ in
  $\C^T$, $Um\colon A\to B$ is monic as well.

  Furthermore, $Um\colon A \to B$ is non-empty; for if it were not,
  i.e.~if $A$ were a strict initial object in $\C$, then by
  $a\colon TA\to A$, $TA\cong A$ is initial in $\C$, and it follows
  that $(A,a)$ is a strict initial object in $\C^T$, a contradiction.

  Thus, $UH^Tm = HUm$ is monic, and therefore so is $H^Tm$.
\end{proof}

\begin{lemma} \label{factorizationsCoalgebra} If $H\colon\C\to \C$ preserves
  non-empty monos, then the (strong epi,mono)-factorizations lift from $\C$ to
  (strong epi-carried,mono-carried)-factorizations in $\Coalg H$.
\end{lemma}
\begin{proof}
  Given a coalgebra homomorphism $h\colon (C,c) \to (D,d)$ and its factorization in $\C$:
  \[
  \begin{tikzcd}
    C
    \arrow[->>]{r}{e}
    \arrow[shiftarr={yshift=6mm}]{rr}{h}
    & I
    \arrow[>->]{r}{m}
    & D
  \end{tikzcd}
  \]
  If $I$ is a strict initial object, then $e$ is an isomorphism, $h$ is monic and so $h =
  h\cdot \id_C$ is the factorization in $\Coalg H$. Otherwise, the mono $m$ is
  non-empty and thus preserved by $H$. By the diagonal fill-in property for the strong
  epi $e$ and the mono $Hm$, we obtain a unique coalgebra structure on $I$
  making $e$ and $m$ into homomorphisms:
  \[
    \begin{tikzcd}[baseline=(HC.base)]
      C
      \arrow[->>]{r}{e}
      \arrow[shiftarr={yshift=6mm}]{rr}{h}
      \arrow{d}[swap]{c}
      & I
      \arrow[>->]{r}{m}
      \arrow[dashed]{d}[left]{\exists!}
      & D
      \arrow{d}{d}
      \\
      |[alias=HC]|
      HC
      \arrow[->>]{r}{He}
      \arrow[shiftarr={yshift=-5mm}]{rr}[swap]{Hh}
      & HI
      \arrow[>->]{r}{Hm}
      & HD
    \end{tikzcd}
    \qedhere
  \]
\end{proof}

\subsection{The Rational Fixpoint} 
\label{sec:rationalfixpoint}

Let $H\colon \C \to \C$ finitary on the lfp category $\C$. We denote by
$\Coalg_\fp H$ the full subcategory of $\Coalg H$ of coalgebras
with fp carrier, and by $\Coalg_\lfp H$ the full subcategory of
$\Coalg H$ of coalgebras that arise as filtered colimits of coalgebras
with fp carrier~\cite[Corollary III.13]{streamcircuits}. The
coalgebras in $\Coalg_\lfp H$ are called \emph{lfp coalgebras}, and for
$\C = \Set$ those are precisely the locally finite coalgebras
(i.e.~those coalgebras where every element is contained in a finite
subcoalgebra). 

The final lfp coalgebra $(\varrho H, r)$ exists and is the colimit of the
inclusion $\Coalg_\fp H \hookrightarrow \Coalg H$. Moreover, it is a
fixpoint of $H$ (see~\cite[Lemma~3.4]{iterativealgebras}) called the
\emph{rational fixpoint} of $H$. Here are some examples: for the
functor $2 \times (-)^\Sigma$, where $\Sigma$ is some input alphabet,
the rational fixpoint is the set of regular languages over
$\Sigma$; the rational fixpoint of a polynomial set functor associated
to a finitary signature $\Sigma$ is the set of rational
$\Sigma$-trees~\cite{iterativealgebras}, i.e.~finite and infinite
$\Sigma$-trees having, up to isomorphism, finitely many subtrees only~\cite{ginali};
one obtains rational weighted languages for Noetherian semirings
$S$ for a functor on the category of $S$-semimodules~\cite{bms13}, and
rational $\lambda$-trees for a functor on the category of presheaves
on finite sets~\cite{highrecursion} or for a related functor on
nominal sets~\cite{MiliusWissmannRatlambda}.  

If the classes of fp and
fg objects coincide in $\C$, then the rational fixpoint is a
subcoalgebra of the final coalgebra~\cite[Theorem~3.12]{bms13}. This
is the case in the above examples, but not in general,
see~\cite[Example~3.15]{bms13} for a concrete example where the
rational fixpoint does not identify behaviourally equivalent
states. However, even if the classes of fp and fg objects differ, the
rational fixpoint can be a subcoalgebra, e.g.~for every constant
endofunctor.

\subsection{Iterative Algebras} 

One important property of the rational fixpoint $\varrho H$ is that,
besides being the final lfp coalgebra, it is also characterized by a
universal property as an algebra for $H$.

Let $H\colon \C \to \C$ be finitary on the lfp category $\C$ once again. An
$H$-algebra $(A, a\colon HA \to A)$ is called \emph{iterative} if every
\emph{flat equation morphism} $e\colon X \to HX + A$ where $X$ is an fp
object has a unique \emph{solution}, i.e.~there exists a unique
morphism $e^\EquationDagger\colon X \to A$ such that
\[
  \begin{tikzcd}
    X \arrow{rr}{e^\EquationDagger} \arrow{d}[left]{e} && A
    \\
    HX + A
    \arrow{rr}{He^\EquationDagger + A}
    &&
    HA + A
    \arrow{u}[right]{[a, \id_A]}
  \end{tikzcd}
\]
Morphisms of iterative algebras are the usual $H$-algebra
homomorphisms. 

The rational fixpoint $\varrho H$ is characterized as the initial iterative
algebra for $H$~\cite[Theorem~3.3]{iterativealgebras}.

Note that this result is the starting point of the coalgebraic
approach to Elgot's iterative theories~\cite{elgot} and to the
iteration theories of Bloom and
\'Esik~\cite{be,iterativealgebras,amv_em1,amv_em2}. For a
well-motivated and much more detailed introduction to iterative
algebras as well as examples see~\cite{iterativealgebras}.

\takeout{
\paragraph{The Rational Fixpoint.} Before we present our main results on the
locally finite fixpoint in the next section let us now first recall some facts
about the rational fixpoint from~\cite{iterativealgebras,streamcircuits}. Assume
that $F\colon \C\to\C$ is a finitary endofunctor on an lfp category. Then one
considers all $F$-coalgebras with a finitely presentable carrier; they are
supposed to capture all systems with a ``finite'' objects of states, and they
form the full subcategory $\Coalgfp F\hookrightarrow \Coalg F$. We further
consider \emph{locally finitely presentable} (\emph{lfp}, for
short) $F$-coalgebras. On $\Set$, they capture precisely
locally finite coalgebras$\mathrlap{\text{,}}$\footnote{leading to the general notion \emph{lfp}
coalgebra, not to be confused with lfp category.} i.e.~those coalgebras where every element of the
carrier lies in a finite subcoalgebra. We do not recall the general definition
but recall from~\cite[Corollary III.13]{streamcircuits} that lfp coalgebras are
characterized as precisely those $F$-coalgebras arising as filtered colimits of
coalgebras from $\Coalgfp F$. It then follows that a final
lfp coalgebra $r\colon \varrho F \to F(\varrho F)$ exists and that it is constructed
as the colimit of the filtered diagram $\Coalgfp F \hookrightarrow \Coalg F$ of
\emph{all} fp-carried $F$-coalgebras. Furthermore, one can prove that $\varrho
F$ is a fixpoint of $F$ (see~\cite[Lemma~3.4]{iterativealgebras}). The ensuing
$F$-algebra $r^{-1}\colon F(\varrho F) \to \varrho F$ also has a universal property,
too. This algebra is the \emph{initial iterative} $F$-algebra, where an
$F$-algebra $a\colon FA \to A$ is called \emph{iterative} if every \emph{flat
equation morphism} $e\colon X \to FX + A$ where $X$ is an fp object has a unique
\emph{solution}, i.e.~given $e$ there exists a unique morphism $e^\EquationDagger\colon X \to
A$ such that $e^\EquationDagger = [a,A] \o (Fe^\EquationDagger + A) \o e$. This universal
property of $\varrho F$ leads to more powerful finitary corecursion schemes,
solutions theorems and specification principles; this has been explored
e.g.~in~\cite{iterativealgebras,bmr12,mbmr13}. It has also been the starting
point of the coalgebraic approach to Elgot's iterative theories~\cite{elgot} and
also the iteration theories of Bloom and \'Esik~\cite{be}
(see~\cite{iterativealgebras,amv_em1,amv_em2}).

\begin{example}
Prominent examples of the rational fixpoint are:
\begin{enumerate}
\item For a signature functor $H_\Sigma$ on \Set, we get the rational
$\Sigma$-trees as $\varrho H_\Sigma$, i.e.~those $\Sigma$-trees with only
finitely many subtrees (up to isomorphism).
\item Rational weighted languages for Noetherian semirings.
\item Rational $\lambda$-trees. two times: in Nom and $\Set^{\mathcal F}$
\end{enumerate}
\end{example}

In all the above examples the rational fixpoint appears as a subcoalgebra of the
final coalgebra, that means that it collects precisely all behaviours of the
coalgebras from $\Coalgfp F$ modulo behavioural equivalence. It has been proved
that this happens whenever the classes of fp and fg objects coincide in $\C$
(see~\cite[Theorem~3.12]{bms3}). However, in arbitrary lfp categories (or
finitary varieties) this is sometimes not true (e.g.~in the categories of
groups, monoids, Heyting algebras or finitary monads on $\Set$), and most often
this seems to be unknown, e.g.~for idempotent semirings. Moreover, in those
cases where it is known the proofs tend to be rather non-trivial making use of
deep results in algebra or somewhat involved combinatorial arguments.\smnote{add
citations here?} And if the two classes of objects do not coincide it may happen
that the rational fixpoint is \emph{not} a subcoalgebra of the final coalgebra,
i.e.~there exist behaviourally equivalent states that are not identified in
$\varrho F$ (see~\cite[Example~3.15]{bms13} for a concrete example). 
} % takeout

\section{The Locally Finite Fixpoint}
\label{sec:locfp}

The locally finite fixpoint can be characterized similarly to the
rational fixpoint, but with respect to coalgebras with finitely
generated (not finitely presentable) carrier. We show that the
locally finite fixpoint always exists, and is a subcoalgebra of the
final coalgebra, i.e.~identifies all behaviourally equivalent
states. As a consequence, the locally finite fixpoint provides a
fully abstract domain of finitely generated behaviour.

\label{sec:lff}
\begin{assumption} \label{basicassumption}
    Throughout the rest of the paper we assume that $\C$ is an lfp category and
    that $H\colon \C→\C$ is finitary and preserves non-empty monos.
\end{assumption}
\noindent
Recall that the last assumption is met by every lifted functor $H^T\colon \Set^T\to\Set^T$ on a
finitary variety $\Set^T$.

As for the rational fixpoint, we denote the full subcategory of
$\Coalg H$ comprising all coalgebras with finitely generated
carrier by $\Coalgfg H$ and have the following notion of locally finitely
generated coalgebra.

%\subsection{Locally Finitely Generated Coalgebras}

\begin{definition}
    \label{lfgcoalgebra}
    A coalgebra $X\xrightarrow{x} HX$ is called \emph{locally finitely generated
      (lfg)} if for all $f\colon S\rightarrow X$ with $S$ finitely generated, there
    exist a coalgebra $p\colon P \to HP$ in $\Coalgfg H$, a homomorphism $h\colon
    (P,p) \rightarrow (X,x)$ and some $f'\colon S\rightarrow P$ such that $h \o f' =
    f$:
    \[
        \begin{tikzcd}
            S \arrow[->]{r}{f} \arrow[->]{dr}[below left]{f'} & X \arrow{r}{x}& HX \\
            {} & P\arrow[->]{u}[right]{h}
            %\descto{u}{\circlearrowleft}
            \arrow{r}{p}
            & HP \arrow{u}[swap]{Hp}
        \end{tikzcd}
    \]
    $\Coalglfg H\subseteq \Coalg H$ denotes the full subcategory of lfg coalgebras.
\end{definition}
\noindent
Equivalently, one can characterize lfg coalgebras in terms of subobjects and subcoalgebras,
making it a generalization of \emph{local finiteness} in \Set, i.e.~the property of a coalgebra that every element is contained in a finite subcoalgebra. 
\begin{lemma}
    $X\xrightarrow{x} HX$ is an lfg coalgebra iff for all 
    fg subobjects $S\smash{\,\xrightarrowtail{f}\,} X$, there exist
    a subcoalgebra $h\colon (P,p) \rightarrowtail (X,x)$ and a mono $f'\colon
    S\rightarrowtail P$ with~$h\cdot f' = f$,
    i.e.~$S$ is a subobject of $P$. 
\end{lemma}
\begin{proof}
  ($\Rightarrow$) Given a mono $f\colon S\rightarrowtail X$, consider
  the induced factor $f'\colon S \to P$ and factorize the induced
  homomorphism $h\colon (P,p) \to (X,x)$ into a strong epi-carried
  homomorphism $e$ followed by a mono-carried one $m$. Then $\Im(h)$
  is fg since fg objects are closed under strong quotients, and $e \cdot
  f'\colon S \to \Im(h)$ is the desired factor, which is monic since $f =
  h \cdot f' = m\cdot (e \cdot f')$ is so. 

  ($\Leftarrow$) Factor
  $f\colon S\to X$ into a strong epi $e\colon S \to\Im(f)$ and a mono
  $g\colon \Im(f) \rightarrowtail X$. By hypothesis we obtain a
  subcoalgebra $h\colon (P,p) \rightarrowtail (X,x)$ and
  $g'\colon \Im(f) \rightarrowtail P$ with $h \cdot g' = g$. Then
  $f' = e \cdot g'$ is the desired factor of~$f$.
\end{proof}

Evidently, every coalgebra with a finitely generated carrier is lfg.
Moreover, we will prove that the lfg coalgebras are precisely the filtered
colimits of coalgebras from $\Coalgfg H$.

\begin{proposition}\label{prop:lfgcolim}
    Every filtered colimit of coalgebras from $\Coalgfg H$ is lfg.
\end{proposition}
\begin{proof}
  \begin{enumerate}
  \item Observe first that the statement of \autoref{unionsofimages}
    holds for $\Coalg H$ in lieu of $\C$ (despite the fact that
    $\Coalg H$ is not lfp in general). Indeed, this follows since in
    $\Coalg H$ one works with the lifted (strong-epi carried,
    mono-carried) factorizations and using that the forgetful functor
    $\Coalg H \to \C$ creates colimits.

  \item\label{prop:lfgcolim:2} We prove that every directed union of
    coalgebras from $\Coalgfg H$ is an lfg coalgebra.\smnote{This used
      to be a lemma used only once, namely in what is now~(2) of this
      proof. I strongly vote for not having such a lemma separately.}
  
    Let $D\colon (I,\le) \to \Coalg H$ be a diagram of coalgebras from
    $\Coalgfg H$ and of mono-carried coalgebra homomorphisms, where
    $(I,\le)$ is a directed poset. For each $i \in I$ denote
    $Di = (D_i, d_i)$ and let $c_i\colon (D_i,d_i) \to (A,a)$ denote
    the colimit cocone in $\Coalg H$. In order to verify the condition
    in \autoref{lfgcoalgebra}, let $f\colon S \to A$ be a morphism in
    $\C$ where $S$ is fg. Recall that colimits in $\Coalg H$ are
    created by the forgetful functor to $\C$. Hence, the object $A$ is
    a directed colimit of the objects $D_i$ in $\C$, and since $S$ is
    an fg object in $\C$ we obtain the desired factorization:
    \[
      \begin{tikzcd}[ampersand replacement=\&,anchor=base]
        S \arrow{rr}{f}
        \arrow{dr}[below left]{f'}
        \& \&  A
        \\
        \& D_i \arrow{ur}[below right]{c_i}
      \end{tikzcd}
    \]
    
  \item Now let $c_i\colon (X_i,x_i) \to (X,x)$ be a colimit cocone of
    a filtered diagram with $(X_i,x_i)$ in $\Coalgfg H$.  Take the
    (strong epi,mono)-factorizations
    \[
        c_i = (\!\!
        \begin{tikzcd}
        X_i \arrow[->>]{r}{e_i} &
        T_i \arrow[>->]{r}{m_i} &
        X
        \end{tikzcd}
        \!\!)
    \]
    to obtain the subcoalgebras $(T_i,t_i)$ of $(X,x)$. By
    \autoref{unionsofimages}, $f=\id_X\colon X\to X$ is the directed
    union of the $m_i$, and therefore $\Im(f) = X$ is the directed
    colimit of the diagram formed by the $T_i$, both in $\C$ and in
    $\Coalg H$. The coalgebras $(T_i,t_i)$ are in $\Coalgfg H$ since
    strong quotients of finitely generated objects are finitely
    generated. Hence, according
    to~\ref{prop:lfgcolim:2}, $(X,x)$ is lfg.\qedhere
  \end{enumerate}
\end{proof}
\begin{proposition} \label{lfgdirectedunion}
    Every lfg coalgebra $(X,x)$ is a directed colimit of its subcoalgebras from
    $\Coalgfg H$.
\end{proposition}
\begin{proof}
  Recall~\cite[Proof I of Theorem 1.70]{adamek1994locally} that $X$ is
  the directed colimit of the diagram of all its finitely generated
  subobjects. Now the subdiagram given by all subcoalgebras of $X$ is
  cofinal. Indeed, this follows directly from the fact that $(X,x)$ is
  an lfg coalgebra: for every subobject $S\rightarrowtail X$, $S$ fg,
  we have a subcoalgebra of $(X,x)$ in $\Coalgfg H$ containing $S$.
\end{proof}

\begin{corollary}\label{lfgcharacterization}
The lfg coalgebras are precisely the filtered colimits, or
equivalently directed unions, of coalgebras with fg~carrier.
\end{corollary}

\noindent
As a consequence, a coalgebra is final in $\Coalglfg H$ if there is
a unique morphism from every coalgebra with finitely generated carrier:

\begin{proposition}\label{finalforfg}
  An lfg coalgebra $(L,\ell)$ is final in $\Coalglfg H$ if and only if
  for every coalgebra $(X,x)$ in $\Coalgfg H$ there exists a unique
  coalgebra morphism from $(X,x)$ to $(L,\ell)$.
\end{proposition}
\noindent
The proof is analogous to Milius' proof~\cite[Theorem 3.14]{streamcircuits}:
\begin{proof}
  The direction from left to right is clear, as $\Coalgfg H$ is a full
  subcategory of $\Coalglfg H$. For the converse, let $(S,s)$ be some
  lfg coalgebra. By \autoref{lfgdirectedunion}, it is the directed
  union of all its subcoalgebras with finitely generated carrier. For
  every subcoalgebra ${\inj_p}\colon (P,p) \monoto (S,s)$, there exists
  a unique coalgebra homomorphism $p^\LFFDagger\colon (P,p)\to (L,\ell)$. By
  the uniqueness of $p^\LFFDagger$ it follows that $L$ together with the
  $p^\LFFDagger$ form a cocone. Hence there exists a unique coalgebra
  homomorphism $u\colon (S,s) \to (L,\ell)$ such that
  $u\cdot \inj_p = p^\LFFDagger$ for every subcoalgebra $(P,p)$ of
  $(S,s)$. Moreover, for every coalgebra homomorphism
  $\bar u\colon (S,s)\to (L,\ell)$ the equation
  $\bar u\cdot \inj_p = p^\LFFDagger$ must hold as well, due to the
  uniqueness of $p^\LFFDagger$. Since the colimit injections $\inj_p$ are
  jointly epic, one obtains $\bar u = u$ so that $u$ is the unique
  homomorphism from $(S,s)$ to $(C,c)$.
\end{proof}
Cocompleteness of $\C$ ensures that the final lfg coalgebra always exists:
\begin{theorem}\label{thm:final}
    The category $\Coalglfg H$ has a final object, and the final
    lfg coalgebra is the colimit of the inclusion $\Coalgfg H \hookrightarrow \Coalglfg H$.
\end{theorem}
\begin{proof}
  By \autoref{lfgcharacterization}, the colimit of the inclusion $\Coalgfg H
  \hookrightarrow \Coalglfg H$ is the same as the (large) colimit of the entire
  category $\Coalglfg H$, and the latter is clearly the final object of $\Coalglfg H$.
\end{proof}
\begin{notation}\label{N:dagger}
  We denote the final lfg coalgebra by $\LFF H \xrightarrow{\ell}
  H(\LFF H)$, and for every lfg coalgebra $(X,x)$ we write
  \[
    x^\LFFDagger\colon (X,x) \to (\LFF H, \ell)
  \]
  for the unique coalgebra homomorphism.
\end{notation}
\begin{corollary} \label{cor:rat=lff}
  If in $\C$ the classes of fp- and fg-objects coincide, then the final lfg
  coalgebra coincides with the rational fixpoint, i.e.~we have $\LFF H \cong \rho
  H$. 
\end{corollary}
\noindent
Indeed, the colimit constructions of both coalgebras are the same
(cf.~\autoref{sec:rationalfixpoint}).

\autoref{thm:final} provides a construction of the final lfg coalgebra
collecting precisely the behaviours of the coalgebras with fg carriers. In the
following we shall show that this construction does indeed identify precisely
behaviourally equivalent states. In other words, the final lfg coalgebra is
always a subcoalgebra of the final coalgebra. First we show that since fg
objects are closed under strong quotients -- in contrast to fp objects -- we have
a similar property of lfg coalgebras:
\begin{lemma}\label{lfgquotients}
    Lfg coalgebras are closed under quotients, i.e.~for every strong epi carried
    coalgebra homomorphisms $X \twoheadrightarrow Y$, if $X$ is lfg then so is $Y$.
\end{lemma}
\begin{proof}
  Consider some quotient $q\colon (X,x) \twoheadrightarrow (Y,y)$
  where $(X,x)$ is lfg. As $(X,x)$ is the directed colimit of its
  subcoalgebras with fg~carrier, we have that $(Y,y)$ -- the codomain
  of the strong epi-carried $q$ -- is the directed union of the images
  of these subcoalgebras by \autoref{unionsofimages} applied to
  $f = q$. These images are coalgebras with a finitely generated
  carrier since fg object are closed under strong quotients, whence
  $(Y,y)$ is lfg as desired.
\end{proof}
  
\noindent
The failure of the corresponding closure property for lfp coalgebras
is the reason why the rational fixpoint is not necessarily a
subcoalgebra of the final coalgebra. In particular the rational
fixpoint given in \cite[Example 3.15]{bms13} is an lfp coalgebra
having a quotient which is not lfp. However, for the final lfg
coalgebra we have the following result.
\begin{theorem} \label{lffSubcoalg}
  The final lfg $H$-coalgebra is a subcoalgebra of the final $H$-coalgebra.
\end{theorem}
\begin{proof}
  Let $(L,\ell)$ be the final lfg coalgebra. Consider the unique
  coalgebra morphism $L \to \nu H$ and take its (strong epi, mono) factorization:
    \[
        \begin{tikzcd}
            (L,\ell) \arrow[->>,yshift=1mm]{r}{e}
                \arrow[loop left,>->,dashed,looseness=8]{}{\ell^\LFFDagger = \id}
            &(I,i) \arrow[>->]{r}{m}
                \arrow[dashed,yshift=-1mm]{l}{i^\LFFDagger}
            &(\nu H,\tau)
        \end{tikzcd},
        \quad\text{with $e$ strong epi in $\C$}.
    \]
    By \autoref{lfgquotients}, $I$ is an lfg coalgebra and so by
    finality of $L$ we have the coalgebra morphism $i^\LFFDagger$ such
    that $\id_L = i^\LFFDagger\cdot e$. It follows that $e$ is monic and
    therefore an iso.
\end{proof}

\noindent
In other words, the final lfg $H$-coalgebra collects precisely the finitely generated
behaviours from the final $H$-coalgebra. We now show that the final
lfg coalgebra is a fixpoint of $H$ which hinges on the following:
\begin{lemma}\label{Hlfg}
    For every lfg coalgebra $C\xrightarrow{c} HC$, the coalgebra
    $HC\xrightarrow{Hc}{HHC}$ is lfg.
\end{lemma}
\begin{proof}
  Let $(C,c)$ be an lfg coalgebra and consider any morphism
  $f\colon S\to HC$ with $S$ finitely generated.
  By case distinction, one can show that $(C,c)$ is a directed union
  of subcoalgebras $\inj_p\colon (P,p) \rightarrowtail (C,c)$ with
  $(P,p)$ in $\Coalgfg H$ and such that this colimit is preserved by~$H$:
  \begin{itemize}
  \item If $C$ is a strict initial object, then it is an fg object,
    and $(C,c)$ is the directed colimit of the diagram consisting of
    itself only. Hence, the directed colimit is preserved by $H$.

  \item If $C$ is not a strict initial object, then by
    \autoref{lfgdirectedunion} $(C,c)$ is the directed union of all its
    subcoalgebras from $\Coalgfg H$. Since $C$ is not a strict
    initial object, it is also the directed union of all its non-empty
    subcoalgebras (i.e.~the inclusions and connecting morphisms are
    carried by non-empty monos). The latter directed union is
    preserved by $H$.
  \end{itemize}
  Hence,
  $HC \xrightarrow{Hc} HHC$ is a directed colimit with colimit injections
  \[
    H\inj_p\colon (HP,Hp) \rightarrowtail (HC, Hc).
  \]
  Now, since $S$ is an fg object, the morphism
  $f\colon S\to HC$ factorizes through one of these colimit
  injections, i.e.~we have $HP\xrightarrow{Hp}HHP$ with
  $(P,p)\in \Coalgfg H$ and $f'\colon S \to HP$ with $H\inj_p\cdot f'=
  f$. This allows us to construct a coalgebra with fg carrier:
    \[
        S + P
        \xrightarrow{[f',p]}
        HP
        \xrightarrow{H\inr}
        H(S+P)
    \]
    and a coalgebra homomorphism $H\inj_p\cdot[f',p]\colon S+P \to HC$; in
    fact, in the following diagram every part trivially commutes:
    \[
    \begin{tikzcd}
    S
    \arrow{rr}{f}
    \arrow[bend right=5]{drr}[swap]{f'}
    \arrow[bend right=15]{ddrr}[swap]{\inl}
    &&
        HC \arrow{rr}{Hc}
        &&
        HHC
        \\
        &&
        HP \arrow{rr}{Hp}
        \arrow[oldequal]{dr}
        \arrow{u}[right]{H\inj_p}
        &&
        HHP
        \arrow{u}[right]{HH\inj_p}
        \\
        &&
        S+P
        \arrow{r}[pos=0.4]{[f',p]}
        \arrow{u}[left]{[f',p]}
        &
        HP
        \arrow{r}{H\inr}
        &
        H(S+P)
        \arrow{u}[right]{H[f',p]}
        \arrow[shiftarr={xshift=1.5cm}]{uu}[right]{H(H\inj_p\cdot[f',p])}
    \end{tikzcd}
    \]
    This provides the desired factorization of~$f$. 
\end{proof}
\takeout{% not needed anymore, but I leave it here just in case
\begin{remark}
  Note that the previous result also holds in the case where $H$ is
  only assumed to preserve non-empty monomorphisms. In fact, if $C$ is
  non-initial it can be written as the directed union of its
  non-initial subcoalgebras. So $HC \xrightarrow{Hc} HHC$ is a
  directed union since $H$ is finitary and preserves non-empty monos,
  and the rest of the argument is as before.
  
  In the case where $C \cong 0$ is initial, we see that $HC
  \xrightarrow{Hc} HHC$ is an lfg coalgebra using the following
  diagram (here $u_X\colon 0 \to X$ denotes the unique morphism):
  \[
  \begin{tikzcd}
      S \arrow{r}{f} 
      \arrow[oldequal]{rd}
      & 
      H0 \arrow{rr}{Hu_{H0}} 
      &&
      HH0
      \\
      &
      S \arrow{r}{f}
      \arrow{u}{f}
      &
      H0
      \arrow{r}{Hu_S}
      \arrow[oldequal]{ul}
      &
      HS
      \arrow{u}[right]{Hf}
  \end{tikzcd}
  \]
  Indeed, given any morphism $f\colon S \to H0$ where $S$ is an fg object,
  the lower edge provides a coalgebra with the fg carrier $S$ such that
  $f\colon S \to H0$ is a coalgebra homomorphism; in fact, to see that the
  right-hand part commutes remove $H$ and use initiality of $0$. 
\end{remark}
}% end takeout
\noindent
Hence with a proof in virtue to Lambek's Lemma \cite[Lemma 2.2]{lambek},
we obtain the desired fixpoint:
\begin{theorem}\label{T:lfglambek}
    The final lfg $H$-coalgebra is a fixpoint of $H$. 
\end{theorem}
\begin{proof}
  Let $(C,c)$ be a final lfg $H$-coalgebra. Then $(HC, Hc)$ is an lfg
  coalgebra by~\autoref{Hlfg}. Denote by $d\colon (HC, Hc) \to (C,c)$ the
  unique coalgebra homomorphism. Then $d \cdot c$ is a coalgebra
  homomorphism:
  \[
    \begin{tikzcd}
      C 
      \arrow{r}{c}
      \arrow{d}[left]{c}
      &
      HC 
      \arrow{d}{Hc}
      \\
      HC
      \arrow{r}{Hc}
      \arrow{d}[left]{d}
      &
      HHC
      \arrow{d}{Hd}
      \\
      C 
      \arrow{r}{c}
      &
      HC
    \end{tikzcd}
  \]
  Thus, $d \cdot c = \id$ by the finality of $(C,c)$. Therefore
  $c \cdot d = Hd \cdot Hc = H(d \cdot c) = H\id = \id$ using the
  commutativity of the upper square.
\end{proof}
\noindent
In the light of \autoref{T:lfglambek} we will call the final lfg
coalgebra the \emph{locally finite fixpoint} (LFF) of $H$. In
particular, the LFF always exists under \autoref{basicassumption}, and
its finality provides a finitary corecursion/coinduction principle.

\section{Iterative Algebras}
\label{sec:iterative}
Recall from~\cite{iterativealgebras,streamcircuits} that the rational fixpoint of a functor $H$ has a universal property both as a coalgebra and as an algebra for $H$. This situation is completely analogous for the LFF. We already established its universal property as a coalgebra in~\autoref{thm:final}. Now we turn to study the LFF as an algebra for $H$. 
%
%One can extend this to a more powerful corecursion principle, which represents the
%basics of Elgot's iterative theories \cite{iterativealgebras}. A corecursive definition is formalized
%by an equation morphism:
\begin{definition}
  A (flat fg-) \emph{equation morphism} $e$ in an object $A$ is a morphism
  \(
  X \to HX + A,
  \)
  where $X$ is a finitely generated object. If $A$ is the carrier of an
  algebra $\alpha\colon HA\to A$, we call the morphism $e^\EquationDagger\colon X \to A$ a
  \emph{solution} of $e$ if the diagram below commutes:
  \begin{equation}\label{eq:sol}
    \begin{tikzcd}[column sep=15mm]
      X \ar[d,"e"']
      \ar[r,"e^\EquationDagger"]
      &
      A
      \\
      HX+A \ar[r,"He^\EquationDagger+A"]
      & HA+A \ar{u}[swap]{[\alpha,A]}
    \end{tikzcd}
  \end{equation}
  An $H$-algebra $A$ is called \emph{\fgiterative} if every equation morphism in $A$
  has a unique solution.
\end{definition}
\noindent
Note that we are overloading the $\LFFDagger$-notation from
\autoref{N:dagger}. This is justified by the fact, established in
\autoref{LFFisIterative} below, that $\LFF H$ is an fg-iterative
algebra whose operation of taking a unique solution of a flat equation
morphism extends the final semantics operation $\LFFDagger$ from
\autoref{N:dagger}, as explained in \autoref{R:extdag}.

\begin{example}[{Milius~\cite[Example 2.5 (iii)]{m_cia}}]
    The final $H$-coalgebra (considered as an algebra for $H$) is \fgiterative. In fact, in this algebra even morphisms $X \to HX + \nu H$ where $X$ is not necessarily an fg object have a unique solution. 
\end{example}

\begin{definition}
  For \fgiterative algebras $A$ and $B$, an equation morphism
  $e\colon X \to HX + A$ and a morphism $h\colon A\to B$ of $\C$
  define an equation morphism $h\bullet e$ in $B$ as
  \[
    \begin{tikzcd}
      X \arrow{r}{e} &
      HX + A \arrow{r}{HX+h} & HX + B.
    \end{tikzcd}
  \]
  We say that $h$ preserves the solution $e^\EquationDagger$ of $e$ if
  \[
    \begin{tikzcd}
      & X \ar[dl,"e^\EquationDagger"']
      \ar[dr,"(h\bullet e)^\EquationDagger"]
      \\
      A \ar[rr,"h"]
      && B
    \end{tikzcd}
  \]
  The morphism $h$ is called \emph{solution preserving} if it
  preserves the solution of every equation morphism $e$.
\end{definition}
\begin{proposition}
  Let $(A,\alpha)$ and $(B,\beta)$ be \fgiterative algebras. Then a
  morphism $h\colon A \to B$ is solution preserving iff it is an
  algebra homomorphism. 
\end{proposition}
\noindent
The proof is identical to the one for ordinary iterative
algebras~\cite[Proposition~2.18]{iterativealgebras}; we leave it as an
easy exercise for the reader. It follows that \fgiterative algebras
form a full subcategory of the category of all $H$-algebras.

\begin{proposition}\label{LFFisIterative}
    The locally finite fixpoint is \fgiterative.
\end{proposition}
\takeout{% old proof
\begin{proof}  
  Let $e\colon X \to HX + \LFF H$ be an equation morphism. In the
  following we prove that $e$ has a unique solution in $\LFF H$. If
  $X$ is an initial object, we are done because the unique morphism
  $X \to C$ is the desired unique solution. 

  So suppose that $X$ is non-initial. We first prove that $HX + \LFF H$ is
  the directed colimit of the following diagram of monos:
  \begin{itemize}
  \item The diagram scheme $\mathcal{D}$ is the product category
    containing pairs
    $(T \overset{t}\rightarrowtail HX, v^\LFFDagger\colon (V,v) \monoto (\LFF H, \lff))$
    consisting of an fg subobject of $HX$ and a subcoalgebra of $(\LFF H, \lff)$
    where $V$ is fg. $\mathcal{D}$ is (essentially) a directed poset,
    because both of its product components are essentially small and
    directed posets.
%    \vspace{1mm} % the fg subscript was too close to the \mathcal{D} of the
                   % next line

  \item The diagram $D\colon \mathcal{D}\to \C$ is defined by
    \[
      D(t,v) = \Im(t+ v^\LFFDagger\colon T+V \to HX+\LFF H )
    \]
    on objects and by diagonal fill-in on morphisms. This implies that
    $D(t,v)$ is fg, since fg objects are closed under coproducts and
    strong quotients, and that all connecting morphisms are monic.
  \end{itemize}
  That $HX+\LFF H$ is indeed the colimit of $D$ is seen as follows.
  The object $HX$ is the directed colimit of all its fg subobjects
  $t\colon T \monoto HX$ (see~\cite[Proof I of Theorem
  1.70]{adamek1994locally}), and $(\LFF H, \lff)$ is the directed
  colimit of its subcoalgebras in $\Coalgfg H$ by
  \autoref{lfgdirectedunion}. Since colimits commute with coproducts,
  $HX + \LFF H$ is thus a directed colimit with the injection
  $t + v^\LFFDagger\colon T + V \to HX + \LFF H$ in $\C$. By
  \autoref{unionsofimages} applied to $f$ being the identity morphism
  on $HX +\LFF H$, we see that this object is the directed colimit of
  the diagram $D$ of monos.

  Because $X$ is fg, the morphism $e$ factors through one of the
  colimit injections, i.e.~we obtain a mono
  $m\colon W\rightarrowtail HX+\LFF H$, $W$ fg, and a morphism
  $e'\colon W \to X$ such that $m\cdot e' = e$.  We may assume that
  $W$ is not a strict initial object; for otherwise $X$ would be
  initial. Furthermore, choose some $t\colon T\rightarrowtail HX$ and
  $v\colon V\rightarrow HV$ from $\mathcal{D}$ such that $W = D(t,v)$
  as shown in the diagram below:
  \begin{equation}\label{diag:etv}
    \begin{tikzcd}
      & T + V
      \arrow[->>]{d}{[e_T,e_V]}
      \arrow[shiftarr={xshift=14mm}]{dd}{t+v^\LFFDagger}
      \\
      & W
      \arrow[>->]{d}{m}
      \\
      X
      \arrow{r}{e}
      \arrow[dashed]{ur}{e'}
      & HX + \LFF H
    \end{tikzcd}
  \end{equation}
  The intermediate object $W$ carries a coalgebra structure by the
  diagonal fill-in property (using that the mono $m$ is non-empty and
  therefore $Hm$ is monic):
  \[
    \begin{tikzcd}[column sep=0mm]
      & T+V
      \arrow[->>]{dl}[swap]{[e_T,e_V]}
      \arrow{r}{t+v}
      \arrow{dd}{t+v^\LFFDagger}
      &[18mm] HX+HV
      \arrow{dd}[swap]{[He,H\inr\cdot Hv^\LFFDagger]}
      \arrow{dr}{[He',He_V]}
      \\
      W
      \arrow[>->]{dr}[swap]{m}
      &&& HW
      \arrow[>->]{dl}{Hm}
      \\
      &HX + \LFF H
      \arrow{r}{[He,H\inr\cdot \lff]}
      & H(HX+\LFF H)
    \end{tikzcd}
  \]
  Indeed, the left-hand component of the inner square above commutes trivially, and
  its right-hand component commutes because $v^\LFFDagger$ is a
  $H$-coalgebra homomorphism. The two triangles commute by the
  previous diagram~\eqref{diag:etv}. Therefore we obtain a diagonal fill-in
  $w\colon W\to HW$ making $m$ and $e_V$ coalgebra homomorphisms. 
  Since $Hm$ is monic and since $m$ is independent of the choice of $t$ and
  $v$, so is $w$.  We have the following commutative diagram:
  \[
    \begin{tikzcd}[column sep=13mm, row sep=9mm]
      W
      \arrow[>->]{r}{m}
      & HX+\LFF H
      \arrow{r}{He'+\ell}
      \descto{d}{\begin{array}{c}v^\LFFDagger~\text{coalgebra}\\[-1mm]\text{homomorphism}\end{array}}
      \descto[pos=0.7,xshift=5mm]{dr}{\begin{array}{c}
                             \text{finality}\\[-1mm]
                             \text{of~}\LFF H\end{array}}
      & |[xshift=10mm]| HW+H\LFF H
      \arrow[to path={
        ([xshift=-4mm]\tikztostart.south east) |- (\tikztotarget) \tikztonodes
      }]{dd}[]{[Hw^\LFFDagger\!,\,H\LFF H]}
      \\[4mm]
      T+V
      \arrow[->>]{u}{[e_T,e_V]}
      \arrow[->>]{d}[swap]{[e_T,e_V]}
      \arrow{ur}[sloped,above]{t+v^\LFFDagger}
      \arrow{r}{t+v}
      \descto{dr}{\text{Definition~of~}w}
      & HX+HV
      \arrow[bend left=8]{ur}[sloped,above]{He'+Hv^\LFFDagger}
      \arrow{r}{He'+He_V}
      \arrow{d}[]{[He',He_V]}
      & HW+ HW
      \arrow{u}[sloped,above]{HW+Hw^\LFFDagger}
      \arrow{d}[]{[Hw^\LFFDagger,Hw^\LFFDagger]}
      \\
      W
      \arrow{r}{w}
      \arrow{dr}[swap]{w^\LFFDagger}
      & |[yshift=8mm]| HW
      \arrow{r}{H w^\LFFDagger}
      \descto{d}{\begin{array}{c}w^\LFFDagger~\text{coalgebra}\\[-1mm]\text{homomorphism}\end{array}}
      & |[yshift=8mm]| H\LFF H
      \arrow[bend left=10]{dl}{\ell^{-1}}
      \\[-6mm]
      & \LFF H
    \end{tikzcd}
  \]
  Since $[e_T,e_V]$ is an epimorphism, we therefore have
  \[
    w^\LFFDagger
    = \ell^{-1}\cdot [Hw^\LFFDagger \cdot He', \ell ]\cdot m
    = [\ell^{-1}\cdot Hw^\LFFDagger \cdot He', \id_{\LFF H}]\cdot m.
  \]
  It follows that $w^\LFFDagger \cdot e'$ is a solution of $e$ in $(\LFF H, \ell^{-1})$:
  \[
    \begin{tikzcd}[column sep = 14mm,row sep=10mm]
      X
      \arrow{r}{e'}
      \arrow{d}[swap]{e}
      & W
      \arrow{dl}[swap]{m}
      \arrow{r}{w^\LFFDagger}
      & \LFF H
      \\
      HX + \LFF H
      \arrow{r}{He' + \LFF H}
      & HW + \LFF H
      \arrow{r}{Hw^\LFFDagger + \LFF H}
      & H\LFF H + \LFF H
      \arrow{u}[swap]{[\ell^{-1},\LFF H]}
    \end{tikzcd}
  \]
  To verify that this solution is unique, let $s\colon X\to \LFF H$ be any solution
  of $e$, i.e.~we have
  \begin{align}
    s = [\ell^{-1}\cdot Hs, \id_{\LFF H}]\cdot e.
    \label{anyOtherEquationSolution}
  \end{align}
  This defines a coalgebra homomorphism from $(W,w)$ to $\LFF H$:
  \[
    \begin{tikzcd}[column sep = 15mm,row sep= 14mm]
      W
      \arrow{r}{m}
      \arrow{d}[swap]{w}
      \descto[xshift=-5mm]{dr}{\begin{array}{c}
                    \text{Definition} \\[-1mm]
                    \text{of }w
                    \end{array}}
      & HX + \LFF H
      \arrow{r}{[\ell^{-1}\cdot Hs,\LFF H]}
      \arrow{d}[swap]{[He,H\inr\cdot \ell]}
      \arrow[xshift=3mm]{dr}[sloped,above]{[Hs,\ell]}
      \descto[xshift=-13mm]{dr}{\text{\eqref{anyOtherEquationSolution}}}
      &[14mm] \LFF H
      \arrow{d}{\ell}
      \\
      HW
      \arrow{r}{Hm}
      & H(HX+\LFF H)
      \arrow{r}{H[\ell^{-1}\cdot Hs,\LFF H]}
      & H \LFF H
    \end{tikzcd}
  \]
  Hence $[\ell^{-1}\cdot Hs,\id_{\LFF H}]\cdot m = w^\LFFDagger$ and so
  \[
    w^\LFFDagger \cdot e'
    = [\ell^{-1}\cdot Hs,\id_{\LFF H}]\cdot m \cdot e'
    = [\ell^{-1}\cdot Hs,\id_{\LFF H}]\cdot e
    = s,
  \]
  which completes the proof.
\end{proof}}% end takeout
\begin{proof}%[Alternative Proof of \autoref{LFFisIterative}]
  Let $e\colon X \to HX + \LFF H$ be an equation morphism.
  If $X$ is an initial object, we are done because the unique morphism
  $X \to \LFF H$ is the desired unique solution of $e$. 

  So suppose that $X$ is non-initial. In the following, we first define an lfg
  coalgebra structure $\bar e$ on $HX+\LFF H$, then take the unique coalgebra
  homomorphism $\bar e^\LFFDagger\colon HX+\LFF H \to \LFF H$ into the
  final lfg coalgebra, and obtain the unique solution of $e$ as $\bar e^\LFFDagger\cdot e$.
  \begin{enumerate}
  \item We show that the following coalgebra is lfg:
    \[
      \bar e = (HX + \LFF H \xrightarrow{[He, H\inr \cdot \ell]} H(HX + \LFF H)).
    \]
    Consider an fg object $S$ and a monomorphism
    $f\colon S\rightarrowtail HX+\LFF H$. The carrier $HX + \LFF H$ is the
    directed colimit of the following diagram of monos:
    \begin{itemize}
    \item The diagram scheme $\mathcal{D}$ is the product category
      containing pairs
      $(T \overset{t}\rightarrowtail HX, v^\LFFDagger\colon (V,v) \monoto (\LFF H, \lff))$
      consisting of an fg subobject of $HX$ and a subcoalgebra of $(\LFF H, \lff)$
      where $V$ is fg. $\mathcal{D}$ is (essentially) a directed poset,
      because both of its product components are essentially small and
      directed posets.

    \item The diagram $D\colon \mathcal{D}\to \C$ is defined by
      \[
        D(t,v) = \Im(t+ v^\LFFDagger\colon T+V \to HX+\LFF H )
      \]
      on objects and by diagonal fill-in on morphisms. This implies that
    $D(t,v)$ is fg, since fg objects are closed under coproducts and
    strong quotients, and that all connecting morphisms are monic.
  \end{itemize}
    That $HX+\LFF H$ is indeed the colimit of $D$ is seen as follows.
  The object $HX$ is the directed colimit of all its fg subobjects
  $t\colon T \monoto HX$ (see~\cite[Proof I of Theorem
  1.70]{adamek1994locally}), and $(\LFF H, \lff)$ is the directed
  colimit of its subcoalgebras in $\Coalgfg H$ by
  \autoref{lfgdirectedunion}. Since colimits commute with coproducts,
  $HX + \LFF H$ is thus a directed colimit with the injection
  $t + v^\LFFDagger\colon T + V \to HX + \LFF H$ in $\C$. By
  \autoref{unionsofimages} applied to $f$ being the identity morphism
  on $HX +\LFF H$, we see that this object is the directed colimit of
  the diagram $D$ of monos.
  
  Because $X+S$ is fg, the morphism $[e,f]\colon X+S\to HX+\LFF H$
  factors through one of the colimit injections, i.e.~we obtain a mono
  $m\colon W\rightarrowtail HX+\LFF H$, $W$ fg, and a morphism
  $[e',f']\colon X+S \to W$ such that $m\cdot [e', f'] = [e,f]$. We
  know that $W$ is not a strict initial object; for otherwise
  $e'\colon X\to W$ would imply that $X$ is a strict initial
  object. Furthermore, choose some $t\colon T\rightarrowtail HX$ and
  $v\colon V\rightarrow HV$ from $\mathcal{D}$ such that $W = D(t,v)$
  as shown in the diagram below:
  \begin{equation}\label{diag:estv}
    \begin{tikzcd}
      & T + V \arrow[->>]{d}{[e_T,e_V]}
      \arrow[shiftarr={xshift=14mm}]{dd}{t+v^\LFFDagger}
      \\
      & W \arrow[>->]{d}{m}
      \\
      X+S \arrow{r}{[e,f]} \arrow[dashed]{ur}{[e',f']} & HX + \LFF H
    \end{tikzcd}
  \end{equation}
  Since $T+V$ is fg, so is its strong quotient $W$. The intermediate object $W$
  carries a coalgebra structure by the diagonal fill-in property (using that the
  mono $m$ is non-empty and therefore $Hm$ is monic):
  \[
    \begin{tikzcd}[column sep=0mm]
      & T+V
      \arrow[->>]{dl}[swap]{[e_T,e_V]}
      \arrow{r}{t+v}
      \arrow{dd}{t+v^\LFFDagger}
      &[18mm] HX+HV
      \arrow{dd}[swap]{[He,H\inr\cdot Hv^\LFFDagger]}
      \arrow{dr}{[He',He_V]}
      \\
      W
      \arrow[>->]{dr}[swap]{m}
      &&& HW
      \arrow[>->]{dl}{Hm}
      \\
      &HX + \LFF H
      \arrow{r}{[He,H\inr\cdot \lff]}
      & H(HX+\LFF H)
    \end{tikzcd}
  \]
  Indeed, the left-hand component of the inner square above commutes trivially,
  and its right-hand component commutes because $v^\LFFDagger$ is a $H$-coalgebra
  homomorphism. The two triangles commute by the previous
  diagram~\eqref{diag:estv}. Therefore we obtain a morphism $w\colon W\to HW$
  making $m$ a coalgebra homomorphism from $(W,w)$ to $(HX + \LFF H,
  \bar e)$. Thus $m$ is the desired subcoalgebra through which $f$
  factorizes (see~\eqref{diag:estv}).

\item We take the unique coalgebra homomorphism
  \[
    \bar e^\LFFDagger\colon (HX+\LFF H,\bar e)\longrightarrow (\LFF H,\ell)
  \]
  and put
  \[
    s = (X \xrightarrow{e} HX + \LFF H \xrightarrow{\bar e^\LFFDagger} \LFF H).
  \]
  
  Clearly, $\inr\colon \LFF H \to HX + \LFF H$ is a coalgebra
  homomorphism from $(\LFF H, \lff)$ to $(HX + \LFF H, \bar
  e)$. Therefore, we see that $\bar e^\LFFDagger \cdot \inr = 
  \id_{\LFF H}$.

  We proceed to prove that $s$ is a solution of the equation
  morphism $e$, i.e.~diagram~\eqref{eq:sol} commutes:
  \[
    \begin{tikzcd}[column sep = 0mm]
      X
      \arrow{r}{e}
      \arrow{dd}[swap]{e}
      \arrow[shiftarr={yshift=8mm}]{rrr}{s}
      &[10mm] HX + \LFF H
      \arrow{dr}[outer sep = -.5mm]{\bar e}
      \arrow{dd}[swap]{He+\LFF H}
      \arrow{rr}{\bar e^\LFFDagger}
      &[-12mm]
      {} \descto{dr}{\hspace{2em}\begin{array}{c}\text{coalgebra}\\[-2mm]\text{homomorphism}\end{array}}
      &[8mm]
      \LFF H
      \\[2mm]
      &       &
      H(HX+\LFF H)
      \arrow{r}{H\bar e^\LFFDagger}
      \descto{dr}{\hspace{2em}(*)}
      & H(\LFF H)
      \arrow{u}[swap]{\ell^{-1}}
      \\
      HX + \LFF H
      \arrow[shiftarr={yshift=-8mm}]{rrr}[below]{Hs+\LFF H}
      \arrow{r}[below,yshift={-1mm}]{He+\LFF H}
      &
      H(HX+\LFF H)+\LFF H
      \arrow{ru}[swap,outer sep = -.5mm]{[\id, H\inr\cdot \ell]}
      \arrow{rr}[sloped,below]{H\bar e^\LFFDagger+\LFF H}
      &
      &
      H(\LFF H)+\LFF H
      \arrow{u}[swap]{[\id,\ell]}
      \arrow[shiftarr={xshift=16mm}]{uu}[description,xshift={2mm}]{[\ell^{-1},\LFF H]}
    \end{tikzcd}
  \]
  For the commutativity of the part $(*)$, we consider the coproduct
  components separately. The left-hand component trivially commutes,
  and for the right-hand one use $\bar e^\LFFDagger \cdot \inr =
  \id_{\LFF H}$. Since all other parts clearly commute, so does the desired
  outside of the diagram.

  To verify uniqueness of solutions, suppose that $s'\colon
  X\to \LFF H$ is any solution of $e$, i.e.~we have
  \begin{align}\label{otherEquationSolution}
    s' = [\ell^{-1}\cdot Hs', \id_{\LFF H}]\cdot e.
  \end{align}
  Then $[\ell^{-1}\cdot Hs', \id_{\LFF H}]$ is a coalgebra homomorphism
  from $(HX + \LFF, \bar e)$ to $(\LFF H, \lff)$:
  \[
    \begin{tikzcd}[column sep = 15mm,row sep= 14mm]
      HX + \LFF H
      \arrow{r}{[\ell^{-1}\cdot Hs',\LFF H]}
      \arrow{d}[swap]{[He,H\inr\cdot \ell]}
      \arrow[xshift=3mm]{dr}[sloped,above]{[Hs',\ell]}
      \descto[xshift=-13mm]{dr}{\text{\eqref{otherEquationSolution}}}
      &[14mm] \LFF H
      \arrow{d}{\ell}
      \\
      H(HX+\LFF H)
      \arrow{r}{H[\ell^{-1}\cdot Hs',\LFF H]}
      & H \LFF H
    \end{tikzcd}
  \]
  Hence $[\ell^{-1}\cdot Hs',\id_{\LFF H}] = \bar e^\LFFDagger$ so that we obtain
  \[
    s'
    = [\ell^{-1}\cdot Hs',\id_{\LFF H}]\cdot e
    = \bar e^\LFFDagger \cdot e = s.\tag*{\qedhere}
  \]
\end{enumerate}
\end{proof}
\begin{remark}\label{R:extdag}
  Every coalgebra $e\colon X\to HX$ in $\Coalgfg H$ canonically defines an
  equation morphism $\inl\cdot e\colon X\to HX+ \LFF H$, and its solution in
  $\LFF H$ is just the unique coalgebra homomorphism from $(X,e)$ to
  $(\LFF H, \lff)$. To see this consider the diagram below:
  \[
    \begin{tikzcd}[column sep=15mm]
      X \arrow{d}[swap]{e} \arrow{r}{s}
      &
      \LFF H
      \\
      HX
      \arrow{d}[swap]{\inl}
      \arrow{r}{Hs}
      &
      H(\LFF H)
      \arrow{u}[swap]{\lff}
      \arrow{d}{\inl}
      \\
      HX + \LFF H
      \arrow{r}{Hs + \LFF H}
      &
      H(\LFF H) + \LFF H
      \arrow[shiftarr={xshift=15mm}]{uu}[swap]{[\lff, \id]}
    \end{tikzcd}
  \]
  Since the lower square and the right-hand part trivially commute, we
  see that the upper square commutes iff so does the outside of the
  diagram. This shows that the operation $\EquationDagger$ of the
  fg-iterative algebra $\LFF H$ extends the final semantics operation
  of \autoref{N:dagger} and so justifies our overloading of this notation. 
\end{remark}

\begin{lemma}\label{fgiterativeunique}
  Let $\alpha\colon HA\to A$ be an \fgiterative algebra and
  $e\colon X\to HX$ a coalgebra from $\Coalgfg H$. Then there exists a unique
  coalgebra-to-algebra morphism from $X$ to $A$, i.e.~a unique
  morphism $u_e\colon X\to A$ such that $u_e = \alpha \cdot Hu_e \cdot e$.
  \[
    \begin{tikzcd}
      X \arrow[dashed]{r}{\exists! u_e} \arrow{d}[left]{e}
      & A \\
      HX \arrow{r}{Hu_e}
      \descto{ur}{\circlearrowleft}
      & HA \arrow{u}{\alpha}
    \end{tikzcd}
  \]
\end{lemma}
\begin{proof}
  Consider the equation morphism $\inl\cdot e\colon X\to HX+A$. Let
  $s\colon X \to A$ be any morphism and consider the diagram below:
  \[
    \begin{tikzcd}[column sep=1.6cm]
      X \arrow{rr}{s} \arrow{d}{e}
      &
      & A
      \\
      HX \arrow{r}{\inl}
      \arrow[ to path = |- (\tikztotarget) \tikztonodes
      ]{drr}[pos=0.75,below]{Hs}
      & HX+A \arrow{r}{Hs + A}
      &
      HA + A \arrow{u}[right]{[\alpha,A]}
      \\
      &{} \descto{u}{\circlearrowleft} &
      HA \arrow{u}[right]{\inl}
      \arrow[shiftarr={xshift=9ex}]{uu}[right]{\alpha}[left]{\circlearrowleft\ }
    \end{tikzcd}
  \]
  Its lower and right-hand parts always commute. The upper square
  expresses that $s$ is a solution of $\inl \cdot e$, and we see that
  this square commutes if and only if the outside of the diagram
  commutes. Thus, solutions of $\inl \cdot e$ are equivalently,
  coalgebra-to-algebra morphisms from $X$ to $A$. Hence, since the
  former exists uniquely so does the latter.
\end{proof}
\begin{theorem}\label{alllfginitial}
  Let $\alpha\colon HA \to A$ be an \fgiterative algebra and
  $e\colon X\to HX$ an lfg coalgebra. Then there exists a unique
  coalgebra-to-algebra morphism from $X$ to $A$.
\end{theorem}
\begin{proof}
  By \autoref{lfgdirectedunion}, $e\colon X\to HX$ is the union of the
  diagram $D$ of its subcoalgebras $s\colon S\to HS$ with $S$ finitely
  generated. Denote the corresponding colimit injections by
  ${\inj_s\colon (S,s)\to (X,e)}$. By \autoref{fgiterativeunique},
  each such $s$ induces a unique morphism $u_s\colon S\to A$ with
    \begin{equation}
    u_s = \alpha\cdot Hu_s\cdot s.
    \label{usproperty}
    \end{equation}
    For every coalgebra homomorphism $h\colon (R,r) \to (S,s)$ in $\Coalgfg H$ the
    diagram below commutes:
    \[
        \begin{tikzcd}
            R \arrow{r}{h} \arrow{d}[left]{r}
            & S \arrow{r}{u_s} \arrow{d}[left]{s} & A
            \\
            HR \arrow{r}[below]{Hh}
            & HS \arrow{r}[below]{Hu_s}
            & HA \arrow{u}[right]{\alpha}
        \end{tikzcd}
    \]
    Hence $u_r = u_s\cdot h$. In other words, $A$ together with the
    morphisms $u_s\colon S\to A$ form a cocone on $D$. Thus, we obtain a
    unique morphism $u_e\colon X\to A$ such that $u_e \cdot \inj_s = u_s$
    holds for every $(S,s)$ in $\Coalgfg H$.

    We now prove that $u_e$ is an coalgebra-to-algebra
    morphism. For this we consider the following diagram:
    \[
      \begin{tikzcd}
        S \arrow{r}{\inj_s} \arrow{d}[left]{s}
        \arrow[shiftarr={yshift=4ex}]{rr}{u_s}[below]{\circlearrowleft}
        \descto{dr}{\text{(i)}}
        & X \arrow{d}[left]{e} \arrow{r}{u_e}
        \descto{dr}{\text{(ii)}}
        & A
        \\
        HS \arrow{r}[below]{H\inj_s}
        \arrow[shiftarr={yshift=-4ex}]{rr}[below]{Hu_s}[above]{\circlearrowleft}
        & HX \arrow{r}[below]{Hu_e}
        & HA \arrow{u}[right]{\alpha}
      \end{tikzcd}
    \]
    Indeed, the outside and all inner parts except, perhaps, part~(ii)
    commute. This shows that part~(ii) commutes when precomposed by
    every colimit injection $\inj_s$. Since these colimit injections
    are jointly epic, we have that~(ii) commutes as desired. 

    To see the desired uniqueness assume that $u_e$ is any
    coalgebra-to-algebra morphism, i.e.~part~(ii) of the above diagram
    commutes. Since part~(i) also commutes, we see that $u_e \cdot
    \inj_s$ is a coalgebra-to-algebra morphism from $(S,s)$ to
    $(A,\alpha)$. Thus $u_e \cdot \inj_s = u_s$ by the uniqueness of
    the latter (see \autoref{fgiterativeunique}). 
\end{proof}

\begin{corollary} \label{lffInitialIter}
    The locally finite fixpoint is the initial fg-iterative algebra.
\end{corollary}

\takeout{ % I vote for taking this out. There is no explanation why we say this.
In total, we have seen three kinds of systems:
\[
\begin{tikzcd}
X \mathrlap{\text{ fg}}
\ar[d,"e"]
%\descto{dr}{\subseteq}
& X
%\descto{dr}{\hspace{4mm}\subseteq}
\mathrlap{\text{ fg}}
\ar[d,"e"]
& X
\ar[d,"e\mathrlap{\text{\normalsize\hspace{4mm}lfg}}"]
\\
HX
& HX+\LFF H
& HX
\end{tikzcd}
\]
}% end takeout

\section{Relation to the Rational Fixpoint}
\label{sec:relrat}
There are examples, where the rational fixpoint is not a subcoalgebra
of the final coalgebra (e.g.~\cite[Example~3.15]{bms13}
and~\cite[Example~2.18]{milius18}). In categories, where the classes of fp~and fg~objects
coincide, the rational fixpoint and the LFF are isomorphic
(see~\autoref{cor:rat=lff}). In this section we will see, under
slightly stronger assumptions, that fg-carried coalgebras are
quotients of fp-carried coalgebras, and in particular the locally
finite fixpoint is the image of the rational fixpoint in the final
coalgebra, i.e.~we have the following picture:
\[
  \varrho F \twoheadrightarrow \LFF F \rightarrowtail \nu F. 
\]

Recall that an object $X$ of $\C$ is called \emph{projective} if for
every strong epi $e\colon A \twoheadrightarrow B$ and every morphism
$f\colon X \to A$ there exists a morphism $f'\colon X \to A$ such that
$e \cdot f' = f$:
\[
  \begin{tikzcd}
    X \arrow[dashed]{r}{\exists f'}
    \arrow{rd}[below left]{\forall f}
    & A
    \arrow[->>]{d}{e}
    \\
    &
    B
  \end{tikzcd}
\]

\begin{assumption}\label{projectiveassumption}
  In addition to our standing \autoref{basicassumption}, we assume
  that in the base category $\C$, every finitely
  presentable object is a strong quotient of a
  finitely presentable projective object and that the
  endofunctor $H$ preserves strong epis.
\end{assumption}
Note that the related condition for arbitrary objects, i.e.~that every
object is the strong quotient of a projective object is phrased as
\emph{having enough
  projectives}~\cite[Definition~4.6.5]{borceux1994handbook}.
\autoref{projectiveassumption}
is relatively strong but still is met in many situations:
\begin{example}\label{E:proj}
\begin{enumerate}
\item In categories in which all (strong) epis are split, every object is
  projective and every endofunctor preserves epis, e.g.~in \Set or the
  category of vector spaces over a fixed field. In such categories fp and fg objects coincide.

\item\label{E:proj:2} In the category $\Funf(\Set)$ of finitary
  endofunctors on sets, every polynomial functors is projective. This
  is easy to see for the polynomial functor $PX = X^n$ associated to
  the signature $\Sigma$ with a single $n$-ary operation
  symbol. Indeed, this follows from the Yoneda Lemma, since
  $P \cong \Set(n,-)$: given a natural transformation
  $q\colon K \twoheadrightarrow L$ with surjective components, a
  natural transformation $f\colon P \to L$ corresponds to an element
  of $Ln$, and we find its inverse image (under $q_n$) in $Kn$. This
  gives us $f'\colon P \to K$ such that $q \cdot f' = f$. If $\Sigma$
  has more symbols, apply Yoneda Lemma to each of them separately and
  use that $P$ is the coproduct of the corresponding hom-functors.

  Furthermore, note that the finitely presentable
  functors are precisely the quotients of polynomial functors
  $H_\Sigma$, where $\Sigma$ is a finite signature~\cite[Corollary 3.31]{amsw19functor}.
\item In the Eilenberg-Moore category $\Set^T$ for a finitary monad
  $T$, strong epis are surjective $T$-algebra homomorphisms, and thus
  preserved by every endofunctor $H^T$ lifting the endofunctor $H$ on
  $\Set$. Moreover, in $\Set^T$, every free algebra $TX$ is projective; this is
  easy to see using the projectivity of $X$ in \Set. Every finitely
  generated object of $\Set^T$ is a strong quotient of some free
  algebra $TX$ with $X$ finite. Eilenberg-Moore algebras for set
  monads are the setting of the generalized powerset construction
  (see~\autoref{sec:powerset}).
  \end{enumerate}
\end{example}

\begin{proposition}\label{lfpquotient}
    Every coalgebra in $\Coalgfg H$ is a strong quotient of a coalgebra with
    finitely presentable carrier.
\end{proposition}
\begin{proof}
  Take a coalgebra $(X,x)$ with finitely generated carrier. Recall
  that in an lfp category an object is fg if and only if it is a
  strong quotient of some fp
  object~\cite[Proposition~1.69]{adamek1994locally}. Hence $X$ is the
  strong quotient of some fp~object $X'$ via
  $q\colon X' \twoheadrightarrow X$.  By assumption, $X'$ is the
  strong quotient of a projective fp~object $X''$ via
  $q'\colon X'' \to X'$. Since $H$ preserves strong epis, the
  projectivity of $X''$ induces a coalgebra structure $x''$ such that
  $q \cdot q'$ is a homomorphism:
    \[
        \begin{tikzcd}[baseline = (X.base)]
            X'' \arrow[dashed]{r}{x''} \arrow[->>]{d}[left]{q'} &
            HX''\arrow[->>]{d}[right]{Hq'}
            \\
            X' \arrow[->>]{d}[left]{q} &
            HX' \arrow[->>]{d}[right]{Hq}
            \\
            |[alias=X]|
            X \arrow{r}{x} & HX
        \end{tikzcd}
        \qedhere
    \]
\end{proof}

\begin{theorem}
  The locally finite fixpoint $\LFF H$ is the image of the rational
  fixpoint $\varrho H$ in the final coalgebra.
\end{theorem}
\begin{proof}
  Consider the factorization
  $(\varrho H,r) \overset{e}{\twoheadrightarrow} (B,b)
  \overset{m}{\rightarrowtail} (\nu H,\tau)$. Since $\varrho H$ is the
  colimit of all fp carried $H$-coalgebras it is an lfg coalgebra
  by~\autoref{prop:lfgcolim} using that fp objects are also fg. Hence,
  by~\autoref{lfgquotients} the coalgebra $B$ is lfg, too. By
  \autoref{finalforfg} it now suffices to show that from every
  $(X,x) \in \Coalgfg H$ there exists a unique coalgebra morphism into
  $(B,b)$. Given $(X,x)$ in $\Coalgfg H$, it is the quotient
  $q\colon (P,p) \twoheadrightarrow (X,x)$ of an fp-carried coalgebra
  by \autoref{lfpquotient}. Hence, we obtain a unique coalgebra
  morphism $p^\LFFDagger\colon (P,p) \to (\varrho H,r)$. By finality of
  $\nu H$, we have $m\cdot e\cdot p^\LFFDagger = x^\FinalCoalgDagger \cdot q$,
  with $x^\FinalCoalgDagger \colon (X,x)\to (\nu H,\tau)$ the unique
  homomorphism. So the diagonal fill-in property induces a
  homomorphism $(X,x) \to (B,b)$. By the finality of $\nu H$ and
  because $m$ is monic, this is the unique homomorphism
  $(X,x) \to (B,b)$.
\end{proof}

\section{Instances of the Locally Finite Fixpoint}
\label{sec:app}
We will now present a number of instances of the LFF. First note, that
all the instances of the rational fixpoint mentioned in previous work
(see e.g.~\cite{iterativealgebras,bms13,streamcircuits}) are also
instances of the locally finite fixpoint, because in all those cases
the classes of fp and fg objects coincide. For example, the class of
regular languages is the rational fixpoint of $2×(-)^\Sigma$ on
\Set. In this section, we will study further instances of the LFF that
are not known to be instances of the rational fixpoint and which -- to
the best of our knowledge -- have not been characterized by a
universal property yet:
\begin{enumerate}
\item Behaviours of finite-state machines with side-effects as considered by
the generalized powerset construction (cf.~Section~\ref{sec:powerset}),
in particular the following:
\begin{enumerate}
\item Deterministic and ordinary context-free languages obtained as the behaviours of deterministic and non-deterministic stack-machines, respectively.
\item Constructively $S$-algebraic formal power series, i.e.~the
\textqt{context-free} subclass of weighted languages with weights from a
semiring $S$, obtained from weighted context-free grammars.
\end{enumerate}
\item The monad of Courcelle's algebraic trees~\cite{courcelle}.
\end{enumerate}

\subsection{Generalized Powerset Construction}
\label{sec:powerset}
The determinization of a non-deterministic automaton using the powerset
construction is an instance of a more general construction, described by
\citet*{sbbr13} based on an observation by \citet{bartelsphd} (see also \citet{Jacobs05abialgebraic}). 
In that \emph{generalized powerset
construction}, an automaton with side-effects is turned into an ordinary
automaton by internalizing the side-effects in the states. The LFF
interacts well
with this construction, because it precisely captures the behaviours of
finite-state automata with side effects. The notion of side-effect is formalized
by a monad, which induces the category, in which the LFF is considered.

\begin{notation}
  Given a monad $(T, \eta^T, \mu^T)$ on $\C$ and an Eilenberg-Moore
  algebra $a\colon TA \to A$ we denote for any morphism
  $f\colon X \to A$ by $f^\sharp\colon TX \to A$ the unique
  $T$-algebra morphism from the free Eilenberg-Moore algebra
  $(TX,\mu_X)$ to $(A,a)$ extending $f$, i.e.~such that
  $f^\sharp \cdot \eta_X = f$.
\end{notation}

\begin{example} \label{exceptionmonad}
  In Sections~\ref{sec:cfl}
  and~\ref{sec:alg} we are going to make use of Moggi's exception
  monad transformer (see e.g.~\cite{cm93}). Let us recall that for a
  fixed object $E$, the finitary functor $(-)+E$ together with the
  unit $\eta_X = \inl\colon X\to X+E$ and multiplication
  $\mu_X = \id_X+[\id_E,\id_E]: X+E+E \to X+E$ forms a finitary monad,
  the \emph{exception monad}. Its algebras are $E$-pointed objects,
  i.e.~objects $X$, together with a morphism $E\to X$, and
  homomorphisms are morphisms preserving the pointing. So the induced
  Eilenberg-Moore category is just the slice category
  $E/\C \cong \C^{(-)+E}$.
    
    Now, given any monad $T$ we obtain a new monad $T(- + E)$ with obvious unit
    and multiplication. An Eilenberg-Moore algebra for $T(-+E)$ consists of an
    Eilenberg-Moore algebra for $T$ and an $E$-pointing, and homomorphisms are
    $T$-algebra homomorphisms preserving the pointing~\cite{combiningeffects}.
    \twnote{equivalently, it's the slice category $TE / \C^T$}
\end{example}

An automaton with side-effects is modelled as an $HT$-coalgebra, where
$T$ is a finitary monad on $\C$ providing the type of side-effect. For
example, for $HX = 2 \times X^\Sigma$, where $\Sigma$ is an input
alphabet, $2 = \{0,1\}$ and $T$ the finite powerset monad on \Set,
$HT$-coalgebras are non-deterministic automata. However, the
coalgebraic semantics using the final $HT$-coalgebra does not yield
the usual language semantics of non-deterministic automata. This is
obtained by turning the $HT$-coalgebra into a coalgebra for a lifting
of $H$ on $\C^T$ via the generalized powerset construction that we now
recall. We work under the following 

\begin{assumption}
  We assume that $\C$ is an lfp category, $T$ a finitary monad on $\C$ and $H$ a
  finitary endofunctor on $\C$ preserving non-empty monos and $H^T\colon \C^T \to
  \C^T$ is a lifting of $H$, i.e.~$H\cdot U = U \cdot H^T$, where $U\colon \C^T \to
  \C$ is the canonical forgetful functor.
\end{assumption}

The generalized powerset construction transforms an $HT$-coalgebra into an $H^T$-coalgebra on $\C^T$:
For a coalgebra $x\colon X\to HTX$, $HTX$ carries an Eilenberg-Moore algebra, and one
uses freeness of the Eilenberg-Moore algebra $TX$ to obtain a canonical
$T$-algebra homomorphism $x^\sharp\colon (TX,\mu_X^T) \to H^T(TX,\mu_X^T)$. The \emph{coalgebraic language semantics} of $(X,x)$ is then given by 
composing the unique coalgebra morphism induced by $x^\sharp$ with $\eta_X$:  
\[
    \begin{tikzcd}
        X
        \arrow{r}{\eta_X}
        \arrow{d}[left]{x}
        &
        TX
        \arrow{r}{x^{\sharp\FinalCoalgDagger}}
        \arrow[dashed]{dl}{x^\sharp}
        &
        \nu H
        \arrow{d}{\tau}
        \\
        HTX
        \arrow{rr}{Hx^{\sharp\FinalCoalgDagger}}
        &&
        H \nu H
    \end{tikzcd}
\]
This construction yields a functor
\[
     T'\colon \Coalg (HT) \to \Coalg H^T
\]
mapping coalgebras $X\xrightarrow{x} HTX$ to $x^\sharp$ and homomorphisms $f$ to
$Tf$ (see e.g.~\cite[Proof of Lemma 3.27]{bms13} for a proof).

Note that since the right adjoint $U$ preserves monos and is faithful, we know
that $H^T$ preserves monos, and since $T$ is finitary, filtered colimits in
$\C^T$ are created by the forgetful functor to $\C$, and we therefore see that
$H^T$ is finitary. Thus, by Theorem~\ref{thm:final}, $\LFF H^T$ exists and is a
subcoalgebra of $\nu H^T$. Furthermore, recall from \cite{plotkinturi97} and
\cite[Corollary~3.4.19]{bartelsphd} that $\nu H^T$ is carried by $\nu H$
equipped with a canonical $T$-algebra structure, see e.g.~\cite[Notation
3.22]{bms13}.\smnote{Shouldn't we recall this explicitly? But then we need to
  recall that liftings are in bijective correspondence with distributive laws.
  TW: We can do it in disguise: The lifting $H^T$ maps the free $(T\nu H,
  \mu_{\nu H})$ to an algebra on $HT\nu H$. So the morphism $H\nu H
  {\to} HT\nu H$ induces a unique $T$-algebra morphism $\lambda\colon TH\nu H\to
  HT\nu H$. In composition with $T\tau$ we have the coalgebra inducing the
  algebra structure on $\nu H$ by finality.}

In the remainder of this section we will assume that $\C = \Set$. It
is our aim to show that the LFF of $H^T$ characterizes precisely the
coalgebraic language semantics of all finite $HT$-coalgebras.  Formally,
the coalgebraic language semantics of all finite $HT$-coalgebras
is collected by forming the colimit 
\[
    (K,k) = \colim \big(
    \Coalgfg HT \xrightarrow{T'} \Coalg H^T \xrightarrow{U} \Coalg H
    \big).
\]
%Denote its colimit by $k: K\to HK$ with injections $\inj_X: (TX,x^\sharp)\to (K,k)$. % injections not used below
%However, not all identical behaviours are necessarily identified in $K$.
%But identifying those gives precisely the locally finite fixpoint of
%$H^T$ 
Note that this is a filtered colimit because the category $\Coalgfg H$ is closed
under finite colimits and therefore filtered. 

The coalgebra $K$ is not yet a subcoalgebra of $\nu H$ (that
means, not all behaviourally equivalent states are identified in $K$), but taking its image in $\nu H$ we obtain the LFF:
\begin{proposition} \label{firstLFFImage} The image of the unique
  coalgebra morphism $k^\FinalCoalgDagger\colon K \to \nu H^T$ is precisely the
  locally finite fixpoint of the lifting
  $H^T$.% , i.e.~$(I,i) \cong (U\LFF H^T, U\ell)$. -- würde ich weglassen, die U's machen die Sache eher komplizierter...
\end{proposition}
\begin{proof}
  Let us denote by $\inj_x\colon (TX, x^\sharp) \to (K,k)$ the colimit
  injection of the above colimit. For every finite $X$, $(TX,\mu_X)$
  is finitely generated in $\Set^T$, and hence $(TX,x^\sharp)$ is in
  $\Coalgfg H^T$. Therefore we have the unique coalgebra homomorphism
  $x^{\sharp\LFFDagger}\colon (TX,x^\sharp)\to (\LFF H^T,\ell)$. By
  finality of $(\nu H, \tau)$, we see that the outside of the square
  below commutes:
    \[
        \begin{tikzcd}[ampersand replacement=\&]
        \displaystyle\coprod_{\mathclap{(X,x)\in \Coalgfg(HT)}}~(TX,x^\sharp)
        \arrow[->>]{r}[above]{[\inj_x]}
        \arrow{d}[swap]{[x^{\sharp\LFFDagger}]}
        \&[15mm]
        (K,k)
        \arrow{d}{k^\FinalCoalgDagger}
        \arrow[dashed]{ld}[swap]{w}
        \\
        (\LFF H^T, \ell)
        \arrow[>->]{r}{n}
        \&
        (\nu H, \tau)
        \end{tikzcd}
        \text{ in }\Coalg H
    \]
    Recall from \autoref{colimitStrongEpi} that $[\inj_x]_{(X,x)}$ is
    a strong epi in $\Coalg H$. Since $n$ is a mono in $\Coalg H$, we obtain a diagonal
    $w\colon (K,k)\to (\LFF H^T,\ell)$. To prove that
    $(\LFF H^T,\ell)$ is indeed the image of $k^\FinalCoalgDagger$, it remains
    to show that $w$ is a strong epi in $\Set$
    (cf.~\autoref{factorizationsCoalgebra}), i.e.~a surjective map.

    To see that $w$ is surjective we first establish that every
    coalgebra $(Y,f)$ in $\Coalgfg H^T$ is the quotient of some
    $(TX,x^\sharp)$ with $X$ finite. Indeed, given $f\colon Y \to HY$ where
    $Y$ is a finitely generated $T$-algebra, we know that it is the
    quotient of some free $T$-algebra $TX$, $X$ finite, via
    $q\colon TX \twoheadrightarrow Y$, say. We know that $H^T$ preserves
    surjective $T$-algebra morphism since it is a lifting and every
    set functor $H$ preserves surjections. Thus, we can use
    projectivity of the free algebra $TX$ to obtain a coalgebra
    structure $e\colon TX \to HTX$ such that $q$ is a coalgebra
    homomorphism:
    \[
      \begin{tikzcd}
        TX \arrow{r}{e} \arrow[->>]{d}{q}
        & 
        HTX
        \arrow[->>]{d}{Hq}
        \\
        Y
        \arrow{r}{f} 
        & 
        HY
      \end{tikzcd}
    \]
    Note that $e$ is of the desired form $x^\sharp$ for $x = e\cdot \eta_X$. 
    Now since the $f^\LFFDagger \colon Y \to \LFF H$, $(Y,f)$ in $\Coalgfg
    H^T$ are jointly surjective, it follows that so are the
    $x^{\sharp\LFFDagger}$, whence $x^{\sharp\LFFDagger}$ is a jointly
    surjective family. Thus, $w$ is surjective as desired. 
\end{proof}

One can also directly take the union of all desired behaviours:
\begin{theorem} \label{prop:LFFunion}
    The locally finite fixpoint of the lifting $H^T$ comprises precisely the
    images of determinized $HT$-coalgebras:
\begin{equation}
\label{LFFunion}
    \LFF H^T = \ \bigcup_{\mathclap{\substack{x\colon X\to HTX\\
                           X\textrm{ \upshape finite}}}}
        \ x^{\sharp\FinalCoalgDagger}[TX]
    = \ \bigcup_{\mathclap{\substack{x\colon X\to HTX\\
                           X\textrm{ \upshape finite}}}}
        \ x^{\sharp\FinalCoalgDagger}\cdot\eta_X^T[X]
        \subseteq \nu H^T.
\end{equation}
\end{theorem}
\begin{proof}
    Combining the previous \autoref{firstLFFImage} together with 
    \autoref{unionsofimages} proves the first equality. For the second equality,
    consider any element $t \in TX$ and define a new coalgebra on $X+1$ by
    \[
    (Y,y) =\big(\!\!
    \begin{tikzcd}[ampersand replacement = \&, column sep = 1.5cm]
        X+1 \rar{[x, x^\sharp(t)]}
        \& HTX \rar{HT\inl}
        \& HT(X+1)
    \end{tikzcd}
    \!\!\big).
    \]
    It is not difficult to see that  $[\eta_X, t]^\sharp\colon TY\to TX$
    is a $H^T$-coalgebra homomorphism; indeed, to see that the
    following square of $T$-algebra morphisms commutes
    \[
      \begin{tikzcd}
        T(X+1) \arrow{r}{y^\sharp}
        \arrow{d}[left]{[\eta_{X},t]^\sharp}
        &
        HT(X+1)
        \arrow{d}{H[\eta_{X},t]^\sharp}
        \\
        TX \arrow{r}{x^\sharp}
        &
        HTX
      \end{tikzcd}
    \]
    one uses the universal property of the free $T$-algebra $(TY,\mu_Y)$,
    i.e.~it suffices to see that the square commutes when precomposed
    with $\eta_Y\colon Y \to TY$. This is easily done by considering the coproduct
    components of $Y= X+1$ separately. \smnote{Perhaps, we want to do
      these calculations explicity; but then we should somewhere recall the
      laws of Kleisli star.}
    
    Furthermore, we clearly have $t \in
    y^{\sharp\FinalCoalgDagger}\cdot\eta_Y^T[Y]$, and we are done. 
\end{proof}
This result shows that the locally finite fixpoint $\LFF H^T$ captures precisely
the behaviour of finite $HT$-coalgebras,  i.e.~it is a fully abstract
domain for finite state behaviour w.r.t.~the coalgebraic language
semantics. 

In the following subsections, we instantiate the general theory with
examples from the literature to characterize several well-known
notions as LFF.

\subsection{The Languages of Non-deterministic Automata}
Let us start with a simple standard example. We already mentioned that non-de\-ter\-mi\-nis\-tic automata are coalgebras for the functor $X \mapsto 2 \times \Potf(X)^\Sigma$. Hence they are $HT$-coalgebras for $H = 2 \times (-)^\Sigma$ and $T = \Potf$ the finite powerset monad on $\Set$. The above generalized powerset construction then instantiates as the usual powerset construction that assigns to a given non-deterministic automaton its determinization. 

Now note that the final coalgebra for $H$ is carried by the set
$\mathcal L = \Pot(\Sigma^*)$ of all formal languages over $\Sigma$
with the coalgebra structure given by $o\colon \mathcal L \to 2$ with
$o(L) = 1$ iff $L$ contains the empty word and
$t\colon \mathcal L \to \mathcal{L}^\Sigma$ with
$t(L)(s) = \{ w \mid sw \in L\}$ the left language derivative. The
functor $H$ has a canonical lifting $H^T$ on the Eilenberg-Moore
category of $\Potf$, viz.~the category of join semi-lattices. The
final coalgebra $\nu H^T$ is carried by all formal languages with the
join semi-lattice structure given by union and $\emptyset$ and with
the above coalgebra structure. Furthermore, the coalgebraic language
semantics of $x\colon X \to HTX$ assigns to every state of the
non-deterministic automaton $X$ the language it accepts. Observe that
join semi-lattices form a so-called \emph{locally finite variety},
i.e.~the finitely presentable algebras are precisely the finite
ones. Hence, Theorem~\ref{prop:LFFunion} states that the LFF
$\LFF H^T$ is precisely the subcoalgebra of $\nu H^T$ formed by all
languages accepted by finite non-deterministic automata, i.e.~regular languages.

Note that in this example the LFF and the rational fixpoint coincide
since both fp and fg join semi-lattices are simply the finite
ones. Similar characterizations of the coalgebraic language semantics
of finite coalgebras follow from Theorem~\ref{prop:LFFunion} in other
instances of the generalized powerset construction (cf.~e.g.~the
treatment of the behaviour of finite weighted automata
in~\cite{bms13}).

We now turn to examples that, to the best of our knowledge, cannot be
treated using the rational fixpoint.
 
%\subsection{The Behaviour of Stack Machines}
\subsection{The Behaviour of Stack Machines}
Push-down automata are finite state machines with infinitely many
configurations.  It is well-known that deterministic and
non-deterministic pushdown automata recognize different classes of
context-free languages. We will characterize both as instances of the
locally finite fixpoint, using results on stack
machines~\cite{coalgchomsky}; they are finite state machines which can
push or read \emph{multiple} elements to or from their stack in a single
transition, respectively.

That is, a transition of a stack machine in a certain state consists of reading
an input character, going to a successor state based on the stack's topmost
elements and of modifying the topmost elements of the stack. These stack
operations are captured by the stack monad.

\begin{definition}[{Stack monad, \cite[Proposition 5]{goncharov13}}]
    For a finite set of stack symbols $\Gamma$, the \emph{stack monad} is the
    submonad $T$ of the store monad $(-×\Gamma^*)^{\Gamma^*}$ for which the
    elements $\fpair{r,t}$ of $TX \subseteq (X×\Gamma^*)^{\Gamma^*}
    \cong X^{\Gamma^*} \times (\Gamma^*)^{\Gamma^*}\!$ satisfy the
    following restriction: there exists $k$ depending on $r,t$ such that for
    every $w\in \Gamma^k$ and $u\in \Gamma^*$, $r(wu) =r(w)$ and $t(wu) =
    t(w)u$.
\end{definition}

\noindent
Note that the parameter $k$ gives a bound on how many of the topmost stack cells the machine can
access in one step. 

Using the stack monad, stack machines are $HT$-coalgebras, where
$H=B×(-)^\Sigma$ is the Moore automaton functor for the finite input
alphabet $\Sigma$ and the set $B$ of all predicates on (initial) stack
configurations which depend only on the topmost $k$ elements on the stack:
\[
    B = \{
        p \in 2^{\Gamma^*}
        \mid
        \exists k\in \N_0\colon 
        \forall w,u\in \Gamma^*, |w|\ge k\colon
        p(wu) = p(w)
    \} \subseteq 2^{\Gamma^*}.
\]
The final coalgebra $\nu H$ is carried by $B^{\Sigma^*}$ which is
(isomorphic to) a set of functions $\Gamma^* \to 2^{\Sigma^*}$,
mapping stack configurations to
formal languages. \citet{coalgchomsky} show that $H$ lifts to $\Set^T$ and
that finite-state $HT$-coalgebras can be understood as a
coalgebraic version of \emph{deterministic}
pushdown automata without spontaneous transitions. The languages
accepted by those automata are precisely the \emph{real-time deterministic context-free
languages}; this notion goes back to \citet{HarrisonHavel72}.
We obtain the following, with $\gamma_0$ playing
the role of an initial symbol on the stack:

\begin{theorem} \label{stackRealTime}
  The locally finite fixpoint $\LFF H^T$ is carried by the set of all
  maps $f \in B^{\Sigma^*}$ such that for every fixed
  $\gamma_0 \in \Gamma$,
  \( \{ w\in \Sigma^* \mid f(w)(\gamma_0) = 1\} \) is a real-time
  deterministic context-free language.
\end{theorem}
\begin{proof}
    By \cite[Theorem 5.5]{coalgchomsky}, a language $L$ is a real-time
    deterministic context-free language iff there exists some $x\colon X\to
    HTX$, $X$ finite, with its determinization $x^\sharp\colon TX \to HTX$ and there
    exist $s\in X$ and $\gamma_0 \in \Gamma$ such that $f =
    x^{\sharp\FinalCoalgDagger}\cdot\eta_X^T(s) \in B^{\Sigma^*}$ and $f(w)(\gamma_0) = 1$ for
    all $w\in \Sigma^*$. The rest follows by~\eqref{LFFunion}.
\end{proof}
Just as for pushdown automata, the expressiveness of stack machines
increases when equipping them with non-determinism. Technically, this
is done by considering the \emph{non-deterministic stack monad} $T'$,
i.e.~$T'$ denotes a submonad of the non-deterministic store monad
$\Potf(-×\Gamma^*)^{\Gamma^*}$~\cite[Section~6]{coalgchomsky}. In the
non-deterministic setting, a similar property holds, namely that the
determinized $HT'$-coalgebras with finite carrier describe precisely
the context-free languages \cite[Theorem~6.5]{coalgchomsky}.
Combining this with~\eqref{LFFunion} we obtain:
\begin{theorem} \label{stackContextFree}
  The locally finite fixpoint $\LFF H^{T'}$ is carried by the set of
  all maps $f \in B^{\Sigma^*}$ such that for every fixed
  $\gamma_0 \in \Gamma$,
  \( \{ w\in \Sigma^* \mid f(w)(\gamma_0) = 1\} \) is a context-free
  language.
\end{theorem}

\subsection{Context-Free Languages and Constructively $S$-Algebraic Power Series}
\label{sec:cfl}
One generalizes from formal (resp.~context-free) languages to weighted formal
(resp.\ context-free) languages by assigning to each word a weight from a fixed
semiring. More formally, a weighted language -- a.k.a.~\emph{formal power
  series} -- over an input alphabet $X$ is defined as a map $X^* \to S$, where
$S$ is a semiring. The set of all formal power series is denoted by $\FPS{X}$.
Ordinary formal languages are formal power series over the boolean semiring
$\mathbb{B}= \{0,1\}$, i.e.~maps $X^*\to \{0,1\}$.

An important class of formal power series is that of
\emph{constructively $S$-algebraic} formal power series. We show that
this class arises precisely as the LFF of the standard functor
$H = S×(-)^\Sigma$ for deterministic Moore automata on a finitary
variety, i.e.~an Eilenberg-Moore category of a finitary set monad. As
a special case, constructively $\mathbb{B}$-algebraic formal power
series are precisely the context-free languages and they form the
LFF of the functor $\mathbb{B} \times (-)^\Sigma$ on the category of
idempotent semirings.

The original definition of constructively $S$-algebraic formal power series goes back to
\citet{Fliess1974}, see also~\cite{handbookWeightedAutomata}.
Here, we use the equivalent coalgebraic characterization by \citet{jcssContextFree}.

Let $\Poly{X}\subseteq \FPS{X}$ be the subset of those maps $X^* \to S$
having finite support, i.e.~which map all but finitely many $w\in X^*$ to
$0$. If $S$ is commutative, then $\Poly{-}$ yields a finitary
monad and therefore we also have the monad $T=\Poly{-+\Sigma}$
by~\autoref{exceptionmonad}. Note that the monad $\Poly{-}$ is a
composition of two monads: we have $\Poly{X} = S_\omega^{(X^*)}$ where $X
\mapsto X^*$ is the free monoid monad and $X \mapsto S_\omega^{(X)}$ maps a
set $X$ to the free $S$-semimodule on $X$, which is carried by the set
of finite support functions $X \to S$. 

Recall that the algebras for the monad $\Poly{-}$ are the associative
$S$-algebras (over the commutative semiring $S$), i.e.~(left)
$S$-modules $A$ together with a monoid structure $(A, *, 1_A)$ that is
\emph{bilinear}, i.e.~an $S$-module morphism in both of its arguments
separately. We write $(A,+, 0_A)$ for the commutative monoid and
$(s,x) \mapsto s.x$ for every $s\in S$ and $x \in A$ for the action of
the semiring $S$ on $A$ which together form the module structure on
$A$. Note that $S$ itself is an $S$-algebra where the scalar and
monoid multiplication are just the semiring multiplication of
$S$. Moreover, for every $S$-algebra $A$ there is the $S$-algebra
morphism
\begin{equation}\label{eq:i}
  i\colon S \to A \qquad\text{with}\qquad i(s) = s.1_A.
\end{equation}

The algebras for $T$ are $\Sigma$-pointed $S$-algebras. %\cite{combiningeffects}. 
The following notions are special instances of $S$-algebras:
\begin{example}
  \begin{enumerate}
  \item Idempotent semirings for $S=\mathbb{B} =
    \{0,1\}$ the Boolean semiring.
  \item Semirings for $S=\mathbb{N}$ the semiring of natural numbers
    (with the usual addition and multiplication).
  \item Rings for $S=\mathbb{Z}$ the semiring of integers (again with the
    usual addition and multiplication).
  \end{enumerate}
\end{example}

\citet[Proposition 4]{jcssContextFree} show that the final
$H$-coalgebra is carried by $\FPS{\Sigma}$ and that constructively
$S$-algebraic series are precisely those elements of $\FPS{\Sigma}$
that arise as the behaviours of finite coalgebras
$c\colon X\to H\Poly{X}$, after determinizing them to some
$\hat c\colon \Poly{X} \to H\Poly{X}$
(see~\cite[Theorem~23]{jcssContextFree}).

However, this determinization is not directly an instance of the
generalized powerset construction. We shall show that the same
behaviours can be obtained by using the standard generalized powerset
construction with an appropriate lifting of $H$ to the category of
$T$-algebras.

\subsection{A Lifting of $S×(-)^\Sigma$ to $S$-algebras}

Let $\Sigma$ be a fixed input alphabet. Given an $S$-algebra structure
on $A$ and a $\Sigma$-pointing $j\colon \Sigma \to A$, we will define
an $S$-algebra structure and $\Sigma$-pointing on $HA =
S×A^\Sigma$. While the $S$-module structure is given by the usual
componentwise operations on the product, a bit of care is needed for
the monoid multiplication on $HA$. To this end we first define the
operation $\fuse{-}\colon S \times A^\Sigma \to A$ by
\[
    \fuse{o,\delta} :=
        i(o) + \displaystyle\sum_{\tau\in \Sigma}\big(j(\tau) * \delta(\tau)\big),
\]
where $i\colon S \to A$ is the morphism from~\eqref{eq:i}. The idea
is that $\fuse{o,\delta}$ acts like a state with output $o$ and
`next states' $\delta$. 

\autoref{tab:lifting} shows the definition of the $S$-algebra
operations and $\Sigma$-pointing on $HA$ (given separately on the
product components $S$ and $A^\Sigma$). 
\begin{table}
\[
\begin{nicearray}{l@{\hskip 3mm}l@{\hskip 3mm}l@{\hskip 3mm}l}
        \text{Structure} & \text{Operation} & \text{in }S &\text{in }A^\Sigma\\
        \midrule
        \text{$S$-Module}
        & 0 & 0_S& \sigma\mapsto 0_A \\
        & (o_1,\delta_1) + (o_2, \delta_2)
        & o_1 + o_2
        & \sigma\mapsto \delta_1(\sigma) + \delta_2(\sigma)
        \\
        & s.(o_1,\delta_1)
        & s \cdot o_1
        & \sigma\mapsto s.\delta_1(\sigma)
        \\
        \midrule
        \text{Monoid}
        & 1 & 1_S& \sigma\mapsto 0_A \\
        & (o_1,\delta_1) * (o_2, \delta_2)
        & o_1 \cdot o_2
        & \sigma\mapsto
            \delta_1(\sigma)* \fuse{o_2,\delta_2}
            + i(o_1)* \delta_2(\sigma) \\
        \midrule
        \text{$\Sigma$-pointing} & \sigma\in \Sigma & 0_S& \chi_\sigma\colon \sigma\mapsto 1_A,\quad \tau\mapsto 0_A, \tau\neq \sigma
\end{nicearray}
\]
\caption{$\Poly{-+\Sigma}$-algebra structure on $S×A^\Sigma$}
\label{tab:lifting}
\end{table}
Since these operations only make use of the operations from $S$ (seen as an
$S$-algebra), the $S$-algebra $A$ and its $\Sigma$-pointing $j$, we
immediately see that for every $\Poly{-+\Sigma}$-algebra morphism
$h\colon A\to B$, the morphism $Hh = \id_S \times h^\Sigma:
S×A^\Sigma\to S×B^\Sigma$ is an $\Poly{-+\Sigma}$-algebra morphism
again. 

Thus, in order to see that we have defined a lifting $H^T$ of
$H$ it suffices to prove that $HA$ with the operations defined in
\autoref{tab:lifting} is an $\Poly{-+\Sigma}$-algebra. To this end we
first prove that $[-]\colon S×A^\Sigma \to A$ is an
$\Poly{-+\Sigma}$-algebra morphism.
\begin{lemma}
    The map $\fuse{-}\colon S×A^\Sigma \to A$ preserves the operations
    defined in \autoref{tab:lifting}.
\end{lemma}
\begin{proof}
First, we show that $\fuse{-}$ preserves the $S$-module operations: 
\begin{enumerate}
\item Zero: \(
\begin{aligned}[t]
    \fuse{0_S, \sigma\mapsto 0_A}
    = i(0_S) + \sum_{\tau\in\Sigma} j(\tau)\cdot 0_A
    = 0_A + \sum_{\tau\in\Sigma} 0_A 
    = 0_A.
\end{aligned}\)
\item Addition: \(
\begin{aligned}[t]
    &
    [o_1+o_2, \sigma\mapsto \delta_1(\sigma) +\delta_2(\sigma)]
    \\ =\ &
    i(o_1+o_2) + \sum_{\tau\in \Sigma} \big(j(\tau) * (\delta_1(\sigma) + \delta_2(\sigma))\big)
    \\ =\ &
    i(o_1)+i(o_2) + \sum_{\tau\in \Sigma} \big(j(\tau) * \delta_1(\sigma)\big)
    + \sum_{\tau\in \Sigma}\big(j(\tau) * \delta_2(\sigma))\big)
    \\ =\ &
    [o_1, \delta_1] + [o_2, \delta_2].
\end{aligned} \)
\item Scalar multiplication:
\begin{align*}
     \fuse{s\cdot o,\sigma\mapsto s.\delta(\sigma)}
    & =
    i(s\cdot o) + \sum_{\tau\in \Sigma}\big(j(\tau) * (s.\delta(\tau))\big)
    \\
    =
    s.i(o) + \sum_{\tau\in \Sigma}s.\big(j(\tau)* \delta(\tau)\big)
    &=
    s.\!\left(i(o) + \sum_{\tau\in \Sigma}\big(j(\tau)* \delta(\tau)\big)\right)
    =
    s. [o,\delta].
\end{align*}
\end{enumerate}
Now note that for every $s \in S$ and $x \in A$ we have
\begin{equation}\label{eq:isx}
  i(s) * x = x * i(s), 
\end{equation}
because we can compute as follows:
\[
  i(s) * x = (s.1_A) * x = s.(1_A * x) = s.x = s.(x
  * 1_A) = x * (s.1_A) = x * i(s). 
\]
Using~\eqref{eq:isx}, we see that the monoid operations are preserved by $\fuse{-}$:
\begin{enumerate}
\item One:
\(\quad
    \fuse{1_S, \sigma\mapsto 0_A}
    = i(1_S) +\sum_{\tau\in\Sigma}j(\tau) * 0_A = i(1_S) = 1_A.
\)
\item Multiplication:
\allowdisplaybreaks
\begin{align*}
    &\phantom{=}\ \  [o_1,\delta_1] *  [o_2,\delta_2]
    \\
    & = \left(i(o_1) + \sum_{\tau\in \Sigma} j(\tau) * \delta_1(\tau)\right) * [o_2,\delta_2]
    \\ &
    = i(o_1) * [o_2,\delta_2] + \sum_{\tau\in\Sigma}
    j(\tau)*\delta_1(\tau) * [o_2,\delta_2]
    \\ &
    = i(o_1) * \left(i(o_2)+\sum_{\tau\in\Sigma}j(\tau) * \delta_2(\tau)\right)
    + \sum_{\tau\in\Sigma} j(\tau)*\delta_1(\tau)* [o_2,\delta_2]
    \\ &
    \overset{\eqref{eq:isx}}{=} i(o_1\cdot o_2)+\sum_{\tau\in\Sigma}j(\tau)* i(o_1)* \delta_2(\tau)
    + \sum_{\tau\in\Sigma} j(\tau)*\delta_1(\tau)* [o_2,\delta_2]
    \\ &
    = i(o_1\cdot o_2)+\sum_{\tau\in\Sigma}j(\tau)* \big(i(o_1)* \delta_2(\tau)
    + \delta_1(\tau)* [o_2,\delta_2]\big)
    \\ &
    = \big[
        o_1\cdot o_2, \sigma\mapsto i(o_1)*\delta_2(\sigma)+\delta_1(\sigma)*[o_2,\delta_2]
    \big]
    = \big[(o_1,\delta_1)* (o_2,\delta_2)\big].
\end{align*}
\end{enumerate}
Finally, the $\Sigma$-pointing is also preserved:
\[
    \fuse{0_S,\chi_\sigma}
    = i(0_S) + \displaystyle\sum_{\tau\in \Sigma}\big(j(\tau) * \chi_\sigma(\tau)\big)
    = j(\sigma) * 1_A = j(\sigma).
    \qedhere
\]
\end{proof}
\begin{lemma}
  For every $\Poly{-+\Sigma}$-algebra $A$, $HA = S×A^\Sigma$ equipped
  with the operations in \autoref{tab:lifting} is an
  $\Poly{-+\Sigma}$-algebra.
\end{lemma}
\begin{proof}
  Recall from~\autoref{exceptionmonad} that an
  $\Poly{-+\Sigma}$-algebra is an $\Poly{-}$-algebra together with a
  $\Sigma$-pointing. Thus, it suffices to prove that $HA$ is an
  associative $S$-algebra. It is clear that $HA$ satisfies the axioms
  of an $S$-module because the $S$-module operations in
  \autoref{tab:lifting} are just the usual coordinatewise operations
  on the product.
  
  Next we prove that $(HA,1,*)$ is a monoid.

  \noindent
  Unit axioms: 
\begin{align*}
    (1_S,\sigma\mapsto 0_A)* (o,\delta)
    &= \big(1_S\cdot o, \sigma\mapsto 0_A* [o,\delta] + i(1_S)* \delta(\sigma)\big)
    = \big(o, \sigma\mapsto \delta(\sigma)\big),
\\
    (o,\delta)* (1_S,\sigma\mapsto 0_A)
    &= \big(o\cdot 1_S, \sigma\mapsto \delta(\sigma)* [1_S,\sigma\mapsto 0_A] + i(o)* 0_A\big).
    \\
    &=  \big(o,  \sigma\mapsto \delta(\sigma)* 1_A + 0_A\big)
    =  (o,\delta)
\end{align*}
\noindent
Associativity:
\begin{align*}
    &\big((o_1,\delta_1)* (o_2,\delta_2)\big)*(o_3,\delta_3)
    \\ =\ &
    \big(o_1\cdot o_2,
        \sigma\mapsto \delta_1(\sigma)* [o_2,\delta_2]+ i(o_1)* \delta_2(\sigma)\big)
            * (o_3,\delta_3)
    \\ =\ &
    \big(o_1\cdot o_2 \cdot o_3, \sigma\mapsto \\ & \quad
        \big(\delta_1(\sigma)* [o_2,\delta_2]
            + i(o_1)* \delta_2(\sigma)\big)
             * [o_3,\delta_3]
          + i(o_1\cdot o_2)* \delta_3(\sigma)\big)
    \\ =\ &
    \big(o_1\cdot o_2 \cdot o_3, \sigma\mapsto \\ & \quad
        \delta_1(\sigma)* [o_2,\delta_2]* [o_3,\delta_3]
            + i(o_1)* \delta_2(\sigma)* [o_3,\delta_3]
          + i(o_1\cdot o_2)* \delta_3(\sigma)\big)
    \\ =\ &
    \big(o_1\cdot o_2 \cdot o_3, \sigma\mapsto \\ & \quad
        \delta_1(\sigma)* \big[(o_2,\delta_2)*(o_3,\delta_3)\big]
            + i(o_1)* \big(\delta_2(\sigma)* [o_3,\delta_3]
          + i(o_2)* \delta_3(\sigma)\big)\big)
    %\\ =\ &
    %(o_1,\delta_1)* \big(o_2\cdot o_3, \sigma\mapsto
    %\delta_2(\sigma)* [o_3,\delta_3] + i(o_2)* \delta_3(\sigma)\big)
    \\ =\ &
    (o_1,\delta_1)* \big((o_2,\delta_2)*(o_3,\delta_3)\big).
\end{align*}
So $S×A^\Sigma$ is both a monoid and a $S$-module. We still need to
establish the bilinearity of $*$ with respect to the $S$-Module
structure. For linearity of $*$ in the first argument, we use the 
same property in $A$:
\begin{align*}
  &
  \big((o_1,\delta_1) + (o_2,\delta_2)\big) * (o_3,\delta_3)
  \\ =\ &
  (o_1+o_2, \sigma\mapsto \delta_1(\sigma) + \delta_2(\sigma)) * (o_3,\delta_3)
  \\ =\ &
  \big((o_1+o_2)\cdot o_3, \sigma\mapsto
  (\delta_1(\sigma)+\delta_2(\sigma))* [o_3,\delta_3] + i(o_1+o_2)*
  \delta_3(\sigma)\big)
  \\ =\ &
  \big(o_1\cdot o_3 +o_2\cdot o_3, \sigma\mapsto \\ & \quad
  \delta_1(\sigma)* [o_3,\delta_3]+\delta_2(\sigma)* [o_3,\delta_3] +
  (i(o_1)+i(o_2))* \delta_3(\sigma)\big)
  \\ =\ &
  (o_1,\delta_1) * (o_3,\delta_3) + (o_2,\delta_2) * (o_3,\delta_3),
\end{align*}
\begin{align*}
  (0_S, \sigma\mapsto 0_A)*(o,\delta)
  &= \big(0_S\cdot o, \sigma\mapsto 0_A* [o,\delta] + i(0_S)* \delta(\sigma)\big)
  \\
  &= \big(0_S\cdot o, \sigma\mapsto 0_A* [o,\delta] + 0_A* \delta(\sigma)\big)
  = \big(0_S, \sigma\mapsto 0_A\big),
\end{align*}
\begin{align*}
  \big(s. (o_1,\delta_1)\big) * (o_2,\delta_2)
  &= \big(s\cdot o_1 \cdot o_2,
  \sigma\mapsto s. \delta_1(\sigma)* \fuse{o_2,\delta_2}
  + i(s\cdot o_1)* \delta_2(\sigma)\big)
  \\
  &= \big(s\cdot o_1 \cdot o_2,
  \sigma\mapsto s. (\delta_1(\sigma)* \fuse{o_2,\delta_2}
  + i(o_1)* \delta_2(\sigma))\big)
  \\
  &= s. \big((o_1,\delta_1)* (o_2,\delta_2)\big).
\end{align*}
Finally, linearity in the second argument of $*$ is proved using the
identities for $[-]$:
\begin{align*}
  &
  (o_1,\delta_1) * \big((o_2,\delta_2) + (o_3,\delta_3)\big)
  \\ =\ &
  (o_1\cdot (o_2+o_3), \sigma\mapsto \delta_1(\sigma)* [o_2+o_3, \sigma\mapsto 
  \delta_2(\sigma)+\delta_3(\sigma)] \\ &\quad + i(o_1)* (\delta_2(\sigma) + \delta_3(\sigma))
  \\ =\ &
  (o_1\cdot o_2+ o_1\cdot o_3, \sigma\mapsto 
  \delta_1(\sigma)* [o_2,\delta_2]
  + \delta_1(\sigma)* [o_3,\delta_3] \\ &\quad
  + i(o_1)* \delta_2(\sigma)
  + i(o_1)*  \delta_3(\sigma)
  \\ =\ &
  (o_1,\delta_1) * (o_2,\delta_2) +(o_1,\delta_1) * (o_3,\delta_3),
\end{align*}
\begin{align*}
    (o,\delta) * (0_S, \sigma\mapsto 0_A)
    &= \big(o\cdot 0_S,\sigma\mapsto  \delta(\sigma) * [0_S, \sigma\mapsto 0_A] + i(o)* 0_A\big)
    \\ &=
    \big(o\cdot 0_S, \sigma\mapsto \delta(\sigma) * 0_A + 0_A\big)
    = (0_S,\sigma\mapsto 0_A),
\end{align*}
\begin{align*}
  &
  (o_1,\delta_1) * \big(s. (o_2,\delta_2)\big)
  = (o_1,\delta_1) * (s\cdot o_2,\sigma\mapsto s. \delta_2(\sigma))\big)
  \\ =\ &
  \big(o_1\cdot (s\cdot o_2), \delta_1(\sigma)* [s\cdot o_2,\sigma\mapsto s.
  \delta_2(\sigma)] + i(o_1)* (s. \delta_2(\sigma))\big)
  \\ =\ &
  \big(o_1\cdot (s\cdot o_2), \delta_1(\sigma)* (s. [o_2,\delta_2]) +
  i(o_1)* (s. \delta_2(\sigma))\big)
  \\ =\ &
  \big(s\cdot (o_1\cdot o_2), s. (\delta_1(\sigma)* [o_2,\delta_2]) +
  s. (i(o_1)* \delta_2(\sigma))\big)
  = s. \big((o_1,\delta_1) * (o_2,\delta_2)\big).
\end{align*}
This completes the proof.
\end{proof}
We now prove that applying $\fuse{-}$ does not
change the behaviour of states after unfolding them using the
coalgebra structure. 
\begin{lemma}
  \label{bisimilarAlgebra}
  Let $c\colon A \to H^T A$ be a coalgebra in $\Set^T$, and let $w \in
  A$. 
  Then $w$ and $\fuse{c(w)}$ are behaviourally equivalent w.r.t.~$H$
  on \Set.
\end{lemma}
\begin{proof}
  We show that $h= \fuse{c(-)}\colon A \to A$ is an $H$-coalgebra
  homomorphism. This implies that $c^\FinalCoalgDagger = c^\FinalCoalgDagger \cdot h$, for
  which we obtain the desired result: $c^\FinalCoalgDagger (w) = c^\FinalCoalgDagger \cdot
  h(w) = c^\FinalCoalgDagger \big(\fuse{c(w)}\big)$. 

  First we use that $c = \langle o, \delta\rangle\colon A \to S \times
  A^\Sigma$ is an $\Poly{-+\Sigma}$-algebra morphism to
  see that for every $s \in S$ we have that
  \begin{equation}\label{eq:aux1}
    c(i(s)) = c(s.1_A) = s.(c(1_A)) = s.(1_S, \sigma \mapsto 0_A) = (s, \sigma
    \mapsto 0_A).
  \end{equation}
  Furthermore, for every $\tau\in \Sigma$ and $v \in A$ we prove that 
  \begin{equation}\label{eq:aux2}
    c(j(\tau)) * c(v) = \big(0_S, \sigma\mapsto \chi_\tau(\sigma) * \fuse{c(v)}\big),
  \end{equation}
  where recall that $\chi_\tau\colon \Sigma \to A$ with
  $\chi_\tau(\tau) = 1$ and $\chi_\tau(\sigma) =0$ for
  $\sigma \neq \tau$. Indeed, we compute
    \begin{align*}
        c(j(\tau)) * c(v)
        &= (0_S,\chi_\tau) * c(v)
        = (0_S,\chi_\tau) * \big(o(v), \delta(v)\big)
        \\ &
        = \big(0_S\cdot o(v),
            \sigma\mapsto \chi_\tau(\sigma) * \fuse{o(v), \delta(v)}
                     + i(0_S) * \delta(v)(\sigma) \big)
        \\ &
        = \big(0_S, \sigma\mapsto \chi_\tau(\sigma) * \fuse{o(v),\delta(v)}\big)
        \\ & 
        = \big(0_S, \sigma\mapsto \chi_\tau(\sigma) * \fuse{c(v)}\big).  
    \end{align*}
    We now prove that $h$ is a coalgebra homomorphism. Let $w \in A$
    and compute: 
    \begin{align*}
        c \cdot h(w) 
        &= c(\fuse{c(w)})
        = c(\fuse{o(w), \delta(w)})
        \\ &
        = c\left(i(o(w)) + \sum_{\tau\in \Sigma} \big(j(\tau) * \delta(w)(\tau)\big)\right)
        \\ &
        = c\big(i(o(w))\big) + \sum_{\tau\in \Sigma}
        c\big(j(\tau)\big) * c\big(\delta(w)(\tau)\big)
        \\ &
        \overset{\mathclap{\eqref{eq:aux1},\eqref{eq:aux2}}}{=}~~~ (o(w), \sigma\mapsto 0_A) +
          \sum_{\tau\in \Sigma} \big(0_S, \sigma\mapsto \chi_\tau(\sigma) * \fuse{c(\delta(w)(\tau))}\big)
        \\ &
        = (o(w), \sigma\mapsto 0_A) +
          \left(0_S, \sigma\mapsto \sum_{\tau\in \Sigma} \chi_\tau(\sigma) *
          \fuse{c(\delta(w)(\tau))}\right)
        \\ &
        = (o(w), \sigma\mapsto 0_A) +
          \left(0_S, \sigma\mapsto \fuse{c(\delta(w)(\sigma))}\right)
        \\ &
        = (o(w), \sigma\mapsto \fuse{c(\delta(w)(\sigma))})
        \\ &
        = (\id_S \times h^\Sigma)(o(w),\delta(w))
        \\ &
        = Hh \cdot c(w).
    \end{align*}
    This completes the proof. 
\end{proof}

Given a coalgebra $c\colon X\to H\Poly{X}$,
\citet[Section~4]{jcssContextFree} determinize $c$ to the
coalgebra $\hat c = \fpair{\hat o,\hat\delta}\colon \Poly{X}\to H\Poly{X}$
defined as follows: first, one extends $\fpair{o,\delta}$ to
$\fpair{\bar o, \bar\delta}\colon X^* \to H\Poly{X}$ by the following
inductive definition (in the following we will often write $\delta$
and its relatives in uncurried form):
\[
  \begin{array}{r@{\ }c@{\ }l@{\qquad}r@{\ }c@{\ }l}
    \bar o (\epsilon) & = & 1 
    &
    \bar\delta(\epsilon, \sigma) & = & 0 \\
    \bar o(xu) & = & o(x)\cdot \bar o(u) 
    &
    \bar\delta(xu, \sigma) & = & \delta(x,\sigma) * u + i(o(x)) *
                            \bar\delta(u,\sigma),
  \end{array} 
\]
where $x \in X$, $u \in X^*$, $\sigma \in \Sigma$, and
$i\colon S \to \Poly{X}$. Second, one uses that
$H\Poly{X} = S \times \Poly{X}^A$ is an $S$-module with the usual
coordinatewise structure on the product and freely extends
$\fpair{\bar o,\bar\delta}$ to the free $S$-module
$S_\omega^{(X^*)} = \Poly{X}$ on $X^*$ to obtain
$\hat c = \fpair{\hat o, \hat\delta}$. It follows that $\hat c$ is an
$S$-module homomorphism, and moreover it is shown in \emph{loc.~cit.}~that
for every $v,w\in \Poly{X}$:
\begin{align}%
  \label{WEAproperty}%
  \!\hat o(v*w) = \hat o(v) \cdot \hat o(w)
  ~\text{and}~
  \hat\delta(v*w,\sigma) = \hat\delta(v,\sigma)*w+i(\hat o(v))*\hat\delta(w,\sigma).
\end{align}
However, for the given coalgebra $(X,c)$ we may also form the coalgebra
\[
X \xrightarrow{c} H \Poly{X} \xrightarrow{\Poly{\inl}} H\Poly{X+\Sigma}
\]
and obtain (abusing notation slightly) a coalgebra $c^\sharp\colon
\Poly{X+\Sigma} \to H\Poly{X + \Sigma}$ by performing the generalized
powerset construction w.r.t.~$T$. 

In \autoref{CFSameCoalgebra}, we show that the
property~\eqref{WEAproperty} together with \autoref{bisimilarAlgebra}
and the definition of $*$ imply that $\hat c$ and $c^\sharp$ are
essentially the same coalgebra structures.
\begin{remark}
  \label{rem:bisim}
  Recall from \autoref{sec:coalgs} the notion of behavioural
  equivalence. One way to establish behavioural equivalence of two
  states is via a bisimulation. For a set functor $H$, a bisimulation
  between two $H$-coalgebras $(C,c)$ and $(D,d)$ is a relation
  $R \subseteq C \times D$ such that $R$ carries a coalgebra structure
  $r\colon R \to HR$ such that the two projections maps $\pi_0\colon R \to C$ and
  $\pi_1\colon R \to D$ are coalgebra homomorphisms. The greatest bisimulation
  on a given coalgebra is called \emph{bisimilarity}.

  Whenever two states $x \in C$ and $y \in D$ are \emph{bisimilar},
  i.e.~contained in any bisimulation $R \subseteq C \times D$, then they are
  behaviourally equivalent.  The converse holds for every functor $H$
  preserving weak pullbacks.

  In \autoref{CFSameCoalgebra}, we will actually use a more refined
  bisimulation proof method, namely, bisimulation up
  to behavioural equivalence. Up-to-techniques such as this one have
  been studied by Rot et al.~\cite{RotEA13}. Here one considers a
  function $f\colon \Pot(C \times D) \to \Pot(C\times D)$, and a
  \emph{bisimulation up to $f$} is a relation $R \subseteq C \times D$
  such that there is a map $r\colon R \to H(f(R))$ making the
  following diagram commutative:
  \[
    \begin{tikzcd}
      C \arrow{d}{c} & R \arrow{l}[above]{\pi_0}\arrow{r}{\pi_1} \arrow{d}{r} & D
      \arrow{d}{d} \\
      HC & H(f(R))\arrow{l}[swap]{H\pi_0} \arrow{r}{H\pi_1} & HD
    \end{tikzcd}
  \]
  We are interested in the function $f$ defined by
  $f(R) = \Rbeheq$, where $\sim$ denotes the behavioural equivalences on $C$ and
  $D$, respectively. Let us spell out the meaning of the above
  diagram for the case of $HX = S \times X^\Sigma$ on $\Set$. Given
  two coalgebras $\fpair{o,\delta}\colon X \to S \times X^\Sigma$ and
  $\fpair{o',\delta'}\colon X' \to S \times (X')^\Sigma$ a bisimulation up
  to behavioural equivalence is a relation $R \subseteq X \times X'$
  such that for all $x \mathbin{R} x'$ we have
  \begin{equation}\label{eq:bisim}
    o(x) = o'(x')
    \qquad\text{and}\qquad
    \text{$\delta(x, \sigma) \Rbeheq \delta'(x',\sigma)$ for every $\sigma \in
      \Sigma$,}
  \end{equation}
  with $\delta$ and $\delta'$ written in their uncurried form.

  It follows from the results of Rot et
  al.~\cite{RotEA13} that whenever two states $x \in C$ and $y \in D$
  are contained in a bisimulation up to behavioural equivalence then
  they are behaviourally equivalent.   
\end{remark}
\begin{lemma}\label{CFSameCoalgebra}
  For every coalgebra $c\colon X \to H\Poly{X}$ in $\Set$, $u$ in
  $(\Poly{X},\hat c)$ and $\Poly{\inl}(u)$ in
  $(\Poly{X+\Sigma}, c^\sharp)$ are behaviourally equivalent.
\end{lemma}
\begin{proof}
  We use that the free $S$-algebra $\Poly{X}$ is a
  quotient of the algebra $MX$ of terms for the signature of
  $S$-algebras via the surjective map $q_X\colon MX \twoheadrightarrow
  \Poly{X}$. Similarly, we have $q_{X+\Sigma}\colon M(X+\Sigma)
  \twoheadrightarrow \Poly{X+\Sigma}$. In fact, $q\colon M \to
  \Poly{-}$ is a natural transformation (even a monad morphism), and
  therefore we have the following commutative square:
  \[
    \begin{tikzcd}
      MX \arrow{r}{q_X} \arrow{d}{M\inl}
      &
      \Poly{X}
      \arrow{d}{\Poly{\inl}}
      \\
      M(X+\Sigma) \arrow{r}{q_{X+\Sigma}}
      &
      \Poly{X+\Sigma}
    \end{tikzcd}
  \]
  Note that $M\inl$ is the embedding of terms over $X$ into the terms
  over $X + \Sigma$. We now prove that for every term $t \in MX$, its
  equivalence classes in $\Poly{X}$ and $\Poly{X+\Sigma}$ are
  behavioural equivalent. This is done by showing that the relation
  \[
    R = \{ (q_X(t),\Poly{\inl}\cdot q_X(t)) \mid t \in MX\} \subseteq \Poly{X} \times \Poly{X +
      \Sigma}
  \]
  is a bisimulation up to behavioural equivalence (see
  \autoref{rem:bisim}). We will abuse notation and often denote the
  equivalence class of a term $t \in MX$ in $\Poly{X}$ or
  $\Poly{X+\Sigma}$ by $t$ again.  Put $c = \fpair{o,\delta}$ and
  $c^\sharp = \fpair{o^\sharp,\delta^\sharp}$.

  Note first that $R$ is nothing but the map $\Poly{\inl}\colon \Poly{X}
  \to \Poly{X+\Sigma}$ considered as a relation. Hence, since this map
  is a morphism of $S$-algebras we have that $R$ is a congruence
  (w.r.t.~$S$-algebra operations). Now we proceed by induction over
  the terms $t \in MX$: 
  
  \begin{enumerate}
    \item \emph{Base Case:} For every $x\in X$, we have that
      $\hat o (x) = o(x) = o^\sharp(x)$ and
      $\hat \delta(x,\sigma) = \delta(x,\sigma)$ whereas
      $\delta^\sharp(x,\sigma) = \Poly{\inl}
      (\delta(x,\sigma))$. Thus,~\eqref{eq:bisim} holds for $t = x$. 

    \item \emph{Induction step for the $S$-module structure:} The
      definition of $\hat c = \fpair{\hat o,\hat\delta}$ on $S$-Module
      operations is coordinatewise \cite[Sect.~3+4]{jcssContextFree}
      and thus identical to the definition of
      $c^\sharp = \fpair{o^\sharp,
        \delta^\sharp}$. Hence,~\eqref{eq:bisim} holds for $t =
      t_1+t_2$, $t = 0$, and $t = s.t'$ for every $s \in S$.

  \item \emph{Induction step for the monoid structure:} The neutral
    element is mapped by $\hat c$ to $(1,\sigma\mapsto 0)$
    \cite[Sect.~4]{jcssContextFree}, and this is identical to the
    definition $c^\sharp$. Thus,~\eqref{eq:bisim} holds for $t = 1$.

    Now suppose that $v,w \in \Poly{X}$ and
    $v',w'\in \Poly{{X+\Sigma}}$, and assume that $v \mathbin{R} v'$,
    $w \mathbin{R} w'$. By induction hypothesis, we have for every $\sigma \in
    \Sigma$, 
    \[
      \begin{array}{r@{\ }c@{\ }l@{\qquad}r@{\ }c@{\ }l}
        \hat o(v) & = & o^\sharp(v'), & 
        \hat\delta(v,\sigma) & \Rbeheq & \delta^\sharp(v',\sigma),
        \\
       \hat o(w) & = & o^\sharp(w'), & 
       \hat\delta(w,\sigma) & \Rbeheq & \delta^\sharp(w',\sigma).
      \end{array}
    \]
    Then we clearly have, using the induction hypothesis in the second step and the definition of $*$ in the last one, that 
    \[
        \hat o(v*w)
        \overset{\eqref{WEAproperty}}{=} \hat o(v)\cdot \hat o(w)
        \overset{\text{IH}}{=} o^\sharp(v') \cdot o^\sharp(w')
        = o^\sharp(v'* w').
    \]
    % Note that final homomorphism
    % $\hat c^\FinalCoalgDagger: \Poly{X} \to \nu H$ preserves multiplication by
    % \cite[Prop 15]{jcssContextFree} and the final
    % $c^{\sharp\FinalCoalgDagger}$ as well, because it lives in $\Set^T$. So
    % for any $x \sim x'$ and $y \sim y'$, $x,y,\in \Poly{X}$,
    % $x',y'\in TX$, we have:
    % \[
    %   \hat c^\dagger(x* y) = \hat c^\dagger(x) * \hat c^\dagger(y)
    %   \overset{x\sim x'}{\underset{y\sim y'}{=}}
    %   c^{\sharp\dagger}(x') * c^{\sharp\dagger}(y') =
    %   c^{\sharp\dagger}(x' * y'),
    % \]
    % i.e.~$\sim$ is a congruence for $*$ (and also for $+$).
    Moreover, for every $\sigma \in \Sigma$ we have
    \[
      \begin{array}{r@{\ }c@{\ }l}
        \hat\delta(v*w, \sigma)
        & \overset{\eqref{WEAproperty}}{=} 
        & \hat\delta(v,\sigma)*w+i(\hat o(v))*\hat\delta(w,\sigma)
        \\ 
        & \mathclap{\Rbeheq} 
        & \ \ \delta^\sharp(v',\sigma)*w'+i(o^\sharp(v')) * \delta^\sharp(w',\sigma)
        \\ 
        & \overset{\mathclap{\text{\autoref{bisimilarAlgebra}}}}{\sim}
        & \ \quad \delta^\sharp(v',\sigma)*\fuse{o^\sharp(w'),\delta^\sharp(w')}
          +i(o^\sharp(v'))*\delta^\sharp(w',\sigma)
        \\ 
        & = 
        & \delta^\sharp(v'*w', \sigma),
      \end{array}
    \]
    where the second step uses the induction hypothesis as well as the
    fact that $R$ and $\sim$ are congruences of $S$-algebras, and the
    last equation uses the definition of $*$
    again. Thus,~\eqref{eq:bisim} holds for $t = t_1 * t_2$, which
    completes the proof.\qedhere
  \end{enumerate}
\end{proof}
\begin{corollary} \label{lffPowerSeries}
  The locally finite fixpoint $\vartheta H^T$ is carried by the set of
  all constructively $S$-algebraic power-series.
\end{corollary}
\begin{proof}
  From \autoref{CFSameCoalgebra} we conclude that $\hat c^\FinalCoalgDagger = c^{\sharp\FinalCoalgDagger}\cdot \Poly{\inl}$ and thus their
  images in $\nu H$ are identical.
  
  By \cite[Theorem~23]{jcssContextFree}, a formal power series is
  constructively $S$-algebraic if and only if it is in the image of
  some
  \[
    \hat c^\FinalCoalgDagger \cdot \eta_X = c^{\sharp\FinalCoalgDagger}\cdot \Poly{\inl} \cdot
  \eta_X = c^{\sharp\FinalCoalgDagger}\cdot \eta^T_X,
  \]
  where $X$ is finite and $\eta_X\colon X \to \Poly{X}$ is the unit of
  the monad $\Poly{-}$.

  The desired result now follows by \eqref{LFFunion}.
\end{proof}

\subsection{Courcelle's Algebraic Trees}
\label{sec:alg}
For a fixed signature $\Sigma$ of so called \emph{givens}, a \emph{recursive
program scheme} (or \emph{rps}, for short) contains mutually recursive definitions of
new operations $\varphi_1,\ldots,\varphi_k$ (with respective arities
$n_1,\ldots,n_k$). The recursive definition of $\varphi_i$ may involve
symbols from $\Sigma$, operations $\varphi_1,\ldots,\varphi_k$ and
$n_i$ variables $x_1,\ldots,x_{n_i}$. The (uninterpreted) solution of an rps is
obtained by unfolding these recursive definitions, producing a possibly
infinite $\Sigma$-tree over $x_1,\ldots,x_{n_i}$ for each operation
$\varphi_i$. The following example shows an rps over the signature $\Sigma
=\{\nicefrac{\star}{0}, \nicefrac{×}{2},
\nicefrac{+}{2}\}$ and its solution:
\[
  \varphi(z) = z + \varphi(\star \times z)
  \qquad\qquad
  \begin{tikzpicture}[lambdatree,mathnodes, level distance=5mm]
    \node (z) {+ }
    child { node {\mathllap{\phantom{X}}z} }
    child { node[yshift=0mm] {+}
            child { node[xshift=-1mm] {×}
                child { node {\star} }
                child { node {z} } }
            child { node[xshift=1mm] {+}
                child { node {×}
                    child { node {\star} }
                    child { node {×}
                        child { node {\star} }
                        child { node {z} } } }
                child { node (ddots) {\ } }
            }
            }
    ;
            \node[anchor=north west] at (ddots.north) {\vdots};
\end{tikzpicture}
\]

In general, an \emph{algebraic $\Sigma$-tree} is a $\Sigma$-tree which
is definable by an rps over $\Sigma$ (see
Courcelle~\cite{courcelle}). Generalizing from a signature to a
finitary endofunctor $H\colon\C\to \C$ on an lfp category,
\citet{secondordermonad} describe an rps as a coalgebra for a functor
$\Hf$ on the category $H/\Mndf(\C)$ whose objects are finitary $H$-pointed
monads on $\C$, i.e.~finitary monads $M$ together with a natural
transformation $H\to M$. They introduce the \emph{context-free} monad
$C^H$ of $H$, which is an $H$-pointed monad that is a subcoalgebra of
the final coalgebra for $\Hf$ and which is the monad of Courcelle's
algebraic $\Sigma$-trees in the special case where $\C = \Set$ and $H$
is the polynomial functor associated to the signature $\Sigma$. We will
now prove that this monad is the LFF of $\Hf$, and thereby we characterize
it by a universal property; this solves an open problem
in~\cite{secondordermonad}.

The setting is again an instance of the generalized powerset
construction, but this time with the category of finitary endofunctors
on $\C$ as the base category in lieu of $\Set$. 

\begin{assumption}
  We assume that $\C$ is an lfp category in
  which the coproduct injections are monic and a coproduct of two monos is also monic.
  Moreover, $H\colon \C\to \C$ is a finitary mono-preserving endofunctor. 
\end{assumption}

Denote by $\Funf(\C)$ the category of finitary endofunctors on $\C$,
which is an lfp category (see~\cite{adamek1994locally}).

Then $H$ induces an endofunctor $H\cdot(-)+\Id$
on $\Funf(\C)$, denoted $\dot H$ and mapping an endofunctor $V$ to the functor
$X\mapsto HVX+X$.
This functor $\dot H$ gets precomposed with a monad on $\Funf(\C)$ as
we now explain.  
\begin{proposition}[Free monad,
  \cite{freealgebras,freetriples}]\label{prop:mon}
  For every object $X$ of $\C$ there exists a free $H$-algebra $F^HX$
  on $X$. Moreover, the object assignment $X \mapsto F^HX$ gives rise
  to a finitary monad on $\C$, and this monad is the \emph{free monad}
  on $H$.
\end{proposition}
\noindent
For example, if $H$ is the polynomial functor associated to a signature
$\Sigma$, then $F^{H}X$ is the usual term algebra that contains all
finite $\Sigma$-trees over the set of generators
$X$. Proposition~\ref{prop:mon} implies that $H \mapsto F^H$ is the
object assignment of a monad on $\Funf(\C)$. Moreover, it is not
difficult to show, using Beck's theorem (see
e.g.~\cite{lane1998categories}), that the Eilenberg-Moore category of
this monad is $\Mndf(\C)$, the category of finitary monads on $\C$.
In addition, we have the following 
\begin{lemma}\label{L:finitary}
  The monad $H \mapsto F^H$ is finitary. 
\end{lemma}
\begin{proof}
Note that for every finitary functor $H\colon \C \to \C$, $- \cdot H$
on $\Funf(\C)$ preserves all colimits, and $H \cdot -$ preserves
filtered colimits. It follows from Kelly's
result~\cite[Theorem~23.3]{kelly80} that the free monad $F^H$ on the
finitary functor $H$ is the initial algebra for $\dot H = H \cdot (-) + \Id$
on $\Fun_f(\C)$. This initial algebra can be constructed as the
colimit of the chain of the functors $H_i$, $i < \omega$, where
$H_0 = \Id$ and $H^{i+1} = H\cdot H^i + \Id$. It follows that the
monad $H \mapsto F^H$ is finitary.  
\end{proof}
\noindent
As a consequence, we see that $\Mndf(\C)$ is an lfp category
(see~\cite[Remark~2.78]{adamek1994locally}). 

\begin{remark}\label{R:funfp=fg}
  As shown by Ad\'amek et al.~\cite[Theorem~2.16]{amv_horps_full} (see
  also~\cite[Corollary~3.31]{amsw19functor}) that in $\Fun(\Set)$ fp and fg
  objects coincide. Moreover, every fp endofunctor on $\Set$ is the
  quotient of the polynomial endofunctor associated to a finite
  signature~\cite[Lemma~3.27]{amsw19functor}.
\end{remark}
\noindent
However, in $\Mndf(\Set)$ the classes of fp and fg objects
differ~\cite[Corollary~4.13]{amsw19algebra}. This means that the rational
fixpoint of a finitary functor on $\Mndf(\Set)$ may not be fully
abstract,\twnote{is there a formal definition of being fully
  abstract? SM: Yes, there is.} and therefore its LFF is needed.
\takeout{% this has been moved to the paper with Jirka and Lurdes
We show that fp and fg objects differ in $\Mndf(\Set)$ by relating
monads to monoids via an adjunction. For every monoid $(M, *, 1_M)$ we
have the free $M$-set monad $LM$ with the following object assignment,
unit and multiplication:
\[
  LM(X) = M \times X, 
  \quad
  \eta_X \colon x \mapsto (1_M,x),
  \quad
  \mu_X \colon (n, (m,x)) \mapsto (n * m, x).
\]
This extends to a functor $L\colon \Mon \to \Mndf(\Set)$. For the reverse
direction, recall first that every monad on $\Set$ is \emph{strong};
this was observed by Moggi~\cite[Proposition~3.4]{Moggi1991}. This
means that for every monad $(T, \eta, \mu)$ on $\Set$ we have a
canonical strength, i.e.~a family of morphisms
\[
  s_{X,Y}\colon TX \times Y \to T(X \times Y)
\]
natural in $X$ and $Y$ and such that the following axioms hold
\begin{eqnarray*}
  \begin{tikzcd}[column sep=0mm,row sep=4mm,baseline=(firstrow.base)]
    |[alias=firstrow]|
    TX \times 1
    \arrow{rr}{s_{1,X}}
    & & T(X\times 1)
    \\
    & TX
    \arrow[draw=none]{ul}[sloped,description]{\cong}
    \arrow[draw=none]{ur}[sloped,description]{\cong}
  \end{tikzcd}
  \quad
  \begin{tikzcd}[column sep=5mm,row sep=4mm,baseline=(firstrow.base)]
    |[alias=firstrow]|
    TX \times Y \times Z 
    \arrow{rrd}[below left]{s_{X, Y \times Z}}
    \arrow{rr}{s_{X,Y} \times Z}
    &&
    T(X \times Y) \times Z
    \arrow{d}{s_{X\times Y,Z}}
    \\
    &&
    T(X\times Y\times Z)
  \end{tikzcd}
       \\
  \begin{tikzcd}[column sep=7mm,row sep=4mm,
      every label/.append style={
        inner sep=2mm,
      },baseline=(firstrow.base)
]
    |[alias=firstrow]|
    TX \times Y 
    \arrow{r}{s_{X,Y}}
    &
    T(X \times Y)
    \\
    X \times Y
    \arrow{u}{\mathllap{\eta_X \times Y}}
    \arrow{ru}[below right]{\eta_{X\times Y}}
  \end{tikzcd}
  \begin{tikzcd}[column sep=7mm,row sep=4mm,
      every label/.append style={
        inner sep=2mm,
      },
      baseline=(firstrow.base)
    ]
    |[alias=firstrow]|
    TTX \times Y
    \arrow{d}[swap]{\mu_X \times Y}
    \arrow{r}{s_{TX,Y}}
    &
    T(TX \times Y)
    \arrow{r}{Ts_{X,Y}}
    &
    TT(X\times Y)
    \arrow{d}{\mu_{X\times Y}}
    \\
    TX \times Y
    \arrow{rr}{s_{X,Y}}
    &&
    T(X \times Y)
  \end{tikzcd}
\end{eqnarray*}
Now we have a functor $R\colon \Mndf(\Set) \to \Mon$ sending the set
monad $(T,\eta,\mu)$ with strength $s$ to the monoid $T1$ with unit
$\eta_1\colon 1\to T1$ and multiplication
\[
  m\colon T1\times T1 \xrightarrow{s_{1,T1}} T(1\times T1) \xrightarrow{\cong} TT1
  \xrightarrow{\mu_1} T1.
\]
\begin{proposition}[label = monoidAdjunction]
    We have an adjoint situation $L\dashv R$ with the following unit $\nu$ and
    counit $\epsilon$:
    \begin{align*}
      \nu_M\colon& M\xrightarrow{\ \cong\ } M\times 1 = RLM
                   \\
      \epsilon_T\colon & LRT = T1\times (-)
         \xrightarrow{s_{1,-}}
         T(1\times (-)) \xrightarrow{\ T\cong\ } T,
    \end{align*}
    where $s$ is the strength of $T$.
  \end{proposition}
  \begin{proof}
    It is not hard to see that $\nu_M$ is a monoid morphism, because the
    monoid structure in $M\times 1 = RLM$ boils down to the monoid structure of
    $M$. Furthermore, $\nu_M$ is clearly natural in $M$. 

    For every monad $T$, $\epsilon_T$
    is a natural  transformation $T1\times (-) \to T$ because the strength $s$
    is. The axioms of $s$ being the strength of the monad $T$ prove that
    $s_{1,-}\colon T1\times (-) \to T(1\times (-))$ is a monad morphism by
    straightforward diagram chasing. To see $\nu$ and $\epsilon$ establish an
    adjunction, it remains to check the triangle identities:
    \begin{itemize}
      \item The first identity $\epsilon_{LM}\cdot L\nu_M = \id_{LM}$ is just
        the associativity of the product:
        \[
        \begin{tikzcd}[row sep = 1 mm]
          LM
          \arrow{r}{L\nu_M}
          & LRLM \arrow{rr}{\epsilon_{LM}}
          &
          & LM
          \\
          M\times (-)
          \arrow{r}{\nu_M\times (-)}
          & (M\times 1) \times (-)
          \arrow{r}{\cong}
          & M\times (1\times (-))
          \arrow{r}{M\times \cong}
          & M \times (-)
        \end{tikzcd}
        \]
        The composition is obviously just the identity on $M\times (-)$.

      \item The second identity $R\epsilon_T\cdot \nu_{RT} = \id_{RT}$
        comes directly from the first axiom of the strength $s$ of $T$:
        \[
        \begin{tikzcd}[row sep = 4 mm, baseline=(T1.base)]
          RT
          \arrow{r}{\nu_{RT}}
          & RLRT \arrow{rr}{R\epsilon_{T}}
          &
          & RT
          \\
          |[alias=T1]|
          T1
          \arrow[shiftarr={yshift={-7mm}}]{rr}{T\nu_1}
          \arrow{r}{\nu_{T1}}
          & T1 \times 1
          \arrow{r}{s_{1,1}}
          & T (1\times 1)
          \arrow{r}{\cong}
          & T 1
        \end{tikzcd}
        \qedhere
        \]

      \end{itemize}
  \end{proof}

  Note that from the fact that the unit of the adjunction $L \dashv R$ is
  an isomorphism we see that $L$ is fully
  faithful. Thus, we may regard $\Mon$ as a full coreflective
  subcategory of $\Mndf(\Set)$. Furthermore, the right-adjoint $R$
  preserves filtered colimits; this follows from the fact that
  filtered colimits in $\Mndf(\Set)$ are created by the forgetful
  functor $U\colon \Mndf(\Set) \to \Funf(\Set)$ and are formed objectwise
  in $\Funf(\Set)$. In addition, $L$ preserves monos; in fact, for an
  injective monoid morphism $m\colon M \monoto M'$ the monad morphism
  $Lm\colon LM \to LM'$ is monomorphic since all its components
  $m \times \id_X\colon M \times X \to M' \times X$ are clearly
  injective. By \autoref{lem:fppres}, we therefore have that a monoid
  is fp (resp.~fg) if and only if
  the monad $LM$ is fp (resp.~fg). 
  
  Now it is well-known that in the category $\Mon$ of monoids fp and fg
  objects do not coincide; see Campbell et
  al.~\cite[Example~4.5]{fpsemigroups} for an example of a fg monoid
  which is not fp. Thus, we conclude with the desired result: 
\begin{corollary}
  In the category of finitary monads on $\Set$ the classes of fp and
  fg objects do not coincide.
\end{corollary}
}% end takeout

Let us now proceed to presenting the category and endofunctor whose
LFF will turn out to be the monad of Courcelle's algebraic trees. 
Similarly as in the case of context-free languages, we will work with
the monad $E^{(-)} = F^{H+(-)}$
(cf.~Example~\ref{exceptionmonad}). Its category of Eilenberg-Moore
algebras is isomorphic to the category $H/\Mndf(\C)$ of 
$H$-pointed finitary monads on $\C$. 

Notice that this category is equivalent to a slice category: the
universal property of the monad $F^H$ states, that for every finitary
monad $B$ the natural transformations $H\to B$ are in one-to-one
correspondence with monad morphisms $F^H \to B$.  Hence, the category
$H/\Mndf(\C)$ of finitary $H$-pointed monads on $\C$ is isomorphic to
the slice category $F^H/\Mndf(\C)$.  This finishes the description of
the base category and we now lift the functor $\dot H$ to this
category.

Consider an $H$-pointed monad
$(B,\beta\colon H\to B) \in H/\Mndf(\C)$. As shown by Ghani et
al.~\cite{monadsofcoalgebras}, the endofunctor $H\cdot B+\Id$ carries
a canonical monad structure with the unit
$\inr\colon \Id \to H \cdot B + \Id$ and the multiplication
\[
  \begin{tikzcd}
    (HB+\Id)(HB+\Id)
    \arrow[oldequal]{d}
    \\
    HB(HB+\Id) + HB + \Id
    \arrow{d}{HB[\mu \cdot \beta B, \eta] + HB + \Id}
    \\
    HBB + HB + \Id
    \arrow{d}{[H\mu, HB] + \Id}
    \\
    HB + \Id,
   \end{tikzcd}
\]
where $\eta\colon \Id \to B$ and $\mu\colon B \cdot B \to B$ are the unit and
multiplication of the monad $B$. Furthermore, we have an obvious
$H$-pointing 
\[
  H\xrightarrow{\inl\cdot H\eta} H\cdot B +\Id.
\]
Milius and Moss~\cite{mmcatsolrps} proved that this defines an
endofunctor on the category of $H$-pointed monads,
\( \Hf\colon H/\Mndf(\C)\to H/\Mndf(\C), \) which is a lifting of
$\dot H$. In order to verify that $\Hf$ is finitary, we first need to
know how filtered colimits are formed in $H/\Mndf(\C)$.

\takeout{% This is an exercise for a graduate course! We should spare
         % our readers these computations.
\begin{lemma}\label{mndfilteredcolimits}
    The forgetful $U\colon \Mndf(\C) \to \Funf(\C)$ creates filtered colimits.
\end{lemma}
\begin{proof}
    Let $D\colon \D\to \Mndf(\C)$, $Di = (M_i, \eta^i, \mu^i)$ be a filtered diagram.
    Take its colimit $M = \colim D$ with injections $\inj_i\colon M_i \to M$ in
    $\Funf(\C)$ and define a monad unit by
    \[
        \eta \equiv \big(
            \Id \xrightarrow{\eta^i} M_i
                \xrightarrow{\inj_i} M
        \big),
        \quad\text{ for any }i\in \D.
    \]
    Similarly, define the monad multiplication $\mu\colon MM\to M$ as the unique
    natural transformation with
    \[
        \begin{tikzcd}[ampersand replacement=\&]
            M_iM_i
                \arrow{r}{\mu^i}
                \arrow{d}[left]{\inj_i* \inj_i}
            \& M_i
                \arrow{d}[right]{\inj_i}
            \\
            MM
                \arrow[dashed]{r}{\mu}
            \& M
        \end{tikzcd}
        \quad\text{for any }i\in \D.
    \]
    The filteredness of $D$ proves the independence of the choice of $i$: for
    any other candidate $j\in \D$ choose an upper bound $m_{i,k}\colon M_i
    \rightarrow M_k \leftarrow M_j\colon m_{j,k}$ of $M_i$ and
    $M_j$. Then we have a commutative diagram
    \[
    \begin{tikzcd}[ampersand replacement=\&]
        \&
        M_i
            \arrow{dr}[above right]{\inj_i}
            \arrow{d}{m_{i,k}}
        \\
        \Id
        \arrow{ur}[above left]{\eta^i}
        \arrow{dr}[below left]{\eta^j}
        \arrow{r}[above]{\eta^k}
        \& M_k
            \arrow{r}[above]{\inj_k}
        \& M
        \\
        \&
        M_j
            \arrow{u}[right]{m_{j,k}}
            \arrow{ur}[below right]{\inj_j}
    \end{tikzcd}
    \]
    The left-hand triangles commute because $m_{i,k}, m_{j,k}$ are monad
    morphisms and the right-hand triangles because $m_{i,k}, m_{j,k}$ are
    connecting natural transformations of $D$ and the $\inj$ the colimit
    injections.

    Note that $(M_iM_i)_{i\in \D}$ is a filtered diagram with colimit $MM$ in
    $\Funf(\C)$. Let us check the monad laws:
    \begin{itemize}
    \item Unit laws: the diagrams
    \[
        \begin{tikzcd}[ampersand replacement=\&]
            M_i
            \arrow{dd}[left]{\inj_i}
            \arrow{r}{\eta^iM_i}
            \arrow[oldequal, shiftarr={yshift=7mm}]{rr}
        \descto[fill=none]{dr}{\text{\parbox{2cm}{\centering\scriptsize
                            Naturality\\of $\eta^i$ }}}
            \&
            M_iM_i
            \arrow{d}{M_i\inj_i}
            \arrow{r}{\mu^i}
        \descto[fill=none,xshift=2mm]{ddr}{\text{\parbox{2cm}{\centering\scriptsize
                            Definition\\of $\mu$ }}}
            \&M_i
            \arrow{dd}{\inj_i}
            \\
            {}
        \descto[fill=none,pos=0.8]{dr}{\text{\parbox{2cm}{\centering\scriptsize
                            Def.~$\eta$ }}}
            \&M_iM
                \arrow{d}{\inj_iM}
            \\
            M
                \arrow{ur}[above left]{\eta^iM}
                \arrow{r}[below]{\eta M}
            \& MM \arrow{r}[below]{\mu}
            \& M
        \end{tikzcd}
        \quad\text{ and }\quad
        \begin{tikzcd}[ampersand replacement=\&]
            M_i
                \arrow{r}{M_i\eta^i}
                \arrow[oldequal, shiftarr={yshift=7mm}]{rr}
                \arrow{dr}[below left]{M_i\eta}
                \arrow{dd}[left]{\inj_i}
            \& M_iM_i
                \arrow{r}{\mu^i}
                \arrow{d}{M_i\inj_i}
        \descto[fill=none,xshift=2mm]{ddr}{\text{\parbox{2cm}{\centering\scriptsize
                            Definition\\of $\mu$ }}}
            \& M_i
                \arrow{dd}{\inj_i}
            \\
            {}
        \descto[fill=none,pos=0.8]{ur}{\text{\parbox{2cm}{\centering\scriptsize
                            Def.~$\eta$ }}}
            \& M_iM
                \arrow{d}{\inj_iM}
            \\
            M
                \arrow{r}[below]{M\eta}
            \& MM
                \arrow{r}[below]{\mu}
            \& M
        \end{tikzcd}
    \]
    commute. As the $\inj_i$ are jointly epic, $(M,\eta,\mu)$ fulfills the unit
    laws.

    \item Associativity:
    \[
    \begin{tikzcd}[ampersand replacement=\&]
        M_iM_iM_i
            \arrow{rrr}{\mu^iM_i}
            \arrow{ddd}[left]{M_i\mu^i}
            \arrow{dr}[sloped,above]{\inj_i*\inj_i*\inj_i}
        \&\&\&
        M_iM_i
            \arrow{ddd}{\mu^i}
            \arrow{dl}[sloped,above]{\inj_i*\inj_i}
        \\[2mm]
        {}
        \& MMM \arrow{r}{\mu M}
              \arrow{d}{M \mu}
        \& MM \arrow{d}{\mu}
        \\
        {}
        \& MM \arrow{r}{\mu}
        \& M
        \\[2mm]
        M_i \arrow{rrr}[below]{\mu^i}
            \arrow{ur}[sloped,above]{\inj_i*\inj_i}
        \&\&\&
        M_i
            \arrow{ul}[sloped,above]{\inj_i}
    \end{tikzcd}
    \]
    The outside commutes, and by definition of $\mu$ also all inner parts
    (except possibly for the middle square). As the $\inj_i*\inj_i*\inj_i$ are
    jointly epic, the middle square commutes as well. 
    \end{itemize}
    By definition of $\eta$ and $\mu$, each $\inj_i\colon M_i\to M$ is a monad
    morphism. In fact, $\eta$ and $\mu$ are the unique natural transformations
    making the diagrams (in the definition) commute, i.e.~are the unique monad
    structure on $M$ such that $\inj_i$ is a monad morphism.

    To see that $(M,\eta,\mu)$ is a colimiting cocone, consider another cocone
    $n_i\colon M_i \to N$ in $\Mndf(\C)$. This induced a unique natural
    transformation $m\colon M\to N$ with $n_i = m\cdot \inj_i$. To see that $m$ is
    also a monad morphism, use the jointly epicness of the $\inj_i$:
    \[
        m \cdot \eta = m\cdot \inj_i \cdot \eta^i = n_i\cdot \eta^i = \eta^N,
    \]

    Consider the following diagram:
    \[
    \begin{tikzcd}[ampersand replacement=\&]
        M_iM_i
            \arrow{rrr}{\mu^i}
            \arrow{dr}[sloped,above]{\inj_i*\inj_i}
            \arrow[bend right]{ddr}[sloped,below]{n_i*n_i}
        \&\&\& M_i
            \arrow{dl}[sloped,above]{\inj_i}
            \arrow[bend left]{ddl}[sloped,below]{n_i}
        \\
        {}
        \& MM \arrow{r}{\mu}
             \arrow{d}[left]{m*m}
             \descto{dr}{?}
        \& M \arrow{d}{m}
        \\
        {}
        \& NN \arrow{r}{\mu^N}
        \& N
    \end{tikzcd}
    \]
    The outside commutes, because $n_i$ is a monad morphism. The outer triangles
    commute on the level of $\Funf(\C)$ and the upper part commutes because
    $\inj_i$ is a monad morphism. Again, as the $\inj_i*\inj_i$ are jointly
    epic, the inner square commutes as well, hence $m$ is a monad morphism.
\end{proof}
}% end takeout
It is a straightforward exercise to prove that the forgetful functor
$U\colon \Mndf(\C) \to \Funf(\C)$ is finitary. Since $U$ is also monadic,
i.e.~$\Mndf(\C)$ is isomorphic to the Eilenberg-Moore category for the
monad $H \mapsto F^H$ on $\Funf(\C)$, we see that $U$ creates filtered
colimits.

Clearly, the canonical projection functor $H/\Mndf(\C) \to \Mndf(\C)$
creates filtered colimits, too. Therefore, filtered colimits in the
slice category $H/\Mndf(\C)$ are formed on the level of $\Funf(\C)$,
i.e.~objectwise. The functor $\dot H$ is finitary on $\Funf(\C)$ and
thus also its lifting $\Hf$ is finitary (see~\autoref{sec:powerset}).
Hence, we see that all requirements from \autoref{basicassumption} are
met: we have a finitary endofunctor $\Hf$ on the lfp category
$H/\Mndf(\C)$, and by \cite[Proposition~2.23]{secondordermonad} the
monos in $H/\Mndf(\C)$ are those monad morphisms in that category
whose components are monic in $\C$. Hence $\Hf$ preserves monos: given
any monomorphism $m\colon B \to B'$ in $H/\Mndf(\C)$ we know that
$Hm_X$ is monic since $H$ preserves monos and then $Hm_X + \id_X$ is
monic since monos are assumed to be closed under coproduct in
$\C$. Thus, by \autoref{thm:final}, we obtain
\begin{corollary}
  The functor $\Hf\colon H/\Mndf(\C) \to H/\Mndf(\C)$ has a locally finite
  fixpoint.
\end{corollary}

\begin{remark}\label{rem:T}
  The final $\Hf$-coalgebra is not of interest to us, but that of a
  related functor is. $\Hf$ generalizes to a functor
  \( \H\colon H/\Mndc(\C) \to H/\Mndc(\C) \) on $H$-pointed countably
  accessible\footnote{A colimit is \emph{countably filtered} if its
    diagram has for every countable subcategory a cocone. A functor is
    \emph{countably accessible} if it preserves countably filtered
    colimits.} monads. For every object $X \in \C$, the finitary
  endofunctor given by $X \mapsto HX + X$ has a final coalgebra
  $TX$. Then $X \mapsto TX$ is the object assignment of a monad
  \cite{aamvcia}, the monad $T$ is countably accessible~\cite{secondordermonad},
  and it carries the final $\H$-coalgebra~\cite{mmcatsolrps}.
\end{remark}
\noindent
\citet{secondordermonad} characterize a (guarded) recursive program scheme as 
a natural transformation 
\[
  V \to H \cdot F^{H+V} + \Id
  \qquad 
  \text{with $V$ fp (in $\Funf(\C)$)}, 
\]
or equivalently, via the generalized powerset construction
w.r.t.~the monad $E^{(-)}=F^{H+(-)}$ as an $\Hf$-coalgebra 
\[
  E^V \to \Hf(E^V)
  \qquad
  \text{(in $H/\Mndf(\C)$).}
\]
These $\Hf$-coalgebras on carriers $E^{V}$ where $V\in \Funf(\C)$
is fp form the full subcategory $\EQ \subseteq \Coalg \Hf$.
\emph{Op.~cit.}~provides two equivalent ways of constructing the monad of
Courcelle's algebraic trees: one works with $\Hf$ for a polynomial endofunctor
$H_\Sigma$ on $\C = \Set$ and obtains the monad of algebraic
$\Sigma$-trees  
\begin{enumerate}
\item  as the image of $\colim \EQ$ in the final coalgebra $T$
  of~\autoref{rem:T}, and 
\item as the colimit of $\EQ[2]$, where $\EQ[2]$ is the closure of $\EQ$
under strong quotients.
\end{enumerate}
We now provide a third characterization, and show that the monad of
Courcelle's algebraic trees is the locally finite fixpoint of $\Hf$.

To this end it suffices to show that $\EQ[2]$ is precisely the diagram
of $\Hf$-coalgebras with an fg carrier. This is established with the
help of the following technical lemmas. We now assume that
$\C = \Set$.
\begin{lemma}
    \label{quasiEpiPreservation}
    $\Hf$ maps strong epis to morphisms carried by strong epi natural
    transformations.
\end{lemma}
\begin{proof}
    Strong epis in slice categories are carried by strong epis, so consider a
    strong epi $q\colon M\epito N$ in $\Mndf(\Set)$.
    Consider the (strong epi,mono)-factorizations of the components
    $q_X\colon MX \to NX$ in $\Set$. This yields a (strong epi,
    mono)-factorization of $q$ in $\Funf$:
    \[
    \begin{tikzcd}
    M \arrow[shiftarr={yshift=7mm}]{rr}{q}
       \arrow[->>]{r}{e}
    & I \arrow[>->]{r}{m}
    & N
    \end{tikzcd}
    \]
    The factorization lifts further to $\Mndf(\Set)$, i.e.~we have
    factorized the monad morphism $q$ into an epi $e$ and a mono $m$
    in $\Mndf(\Set)$. Since every strong epi is an extremal epi (see
    e.g.~\cite{ahs09}), we get that $m$ is an isomorphism. Hence $q$
    has epic components.  Every endofunctor on $\Set$ preserves (strong) epis, so
    $Hq_X+\Id$ is epic for every set $X$. Therefore, so is the natural
    transformation $Hq+\Id$.
\end{proof}
\begin{remark}\label{rem:quotmnd}
  We recall a few properties of finitary monads and endofunctors on
  sets that we shall use in the proof of the next lemma.
  \begin{enumerate}
  \item Every fg object $B$ in $\Mndf(\Set)$ is the strong quotient of
    a free monad $F^P$ where $P$ is the polynomial functor associated
    to a finite signature. To see this, recall that the category
    $\Mndf(\Set)$ is finitary monadic over $\Funf(\Set)$,
    i.e.~$\Mndf(\Set)$ is (isomorphic to) the category of
    Eilenberg-Moore algebras for the finitary monad $H \mapsto F^H$ on
    $\Funf(\Set)$
    (cf.~\autoref{L:finitary}). By~\cite[Theorem~3.5]{amsw19algebra}, we thus
    have that the fg object $B$ is a strong quotient of $F^V$, where
    $V$ is an fp object in $\Funf(\Set)$. In the latter category, the
    fp objects are precisely the quotients of the polynomial functors
    on a finite signature (see \autoref{R:funfp=fg}). Hence, we have
    some polynomial functor $P$ and strong quotient $P \epito V$ in
    $\Funf(\Set)$. Since the left-adjoint $F^{(-)}$ preserves strong
    epis we obtain the desired strong quotient in $\Mndf(\Set)$:
    \[
      F^P \epito F^V \epito B.
    \]

  \item We conclude that every $H$-pointed monad $(B, \beta)$ is the
    strong quotient in $H/\Mndf(\Set)$ of $E^P = (F^{H+P},
    \kappa \cdot \inl)$ where $\kappa\colon H + P \to F^{H+P}$ denotes the
    universal natural transformation of the free monad\twnote{... is the unit of
    the monad $F^{(-)}$ on $\Funf(\Set)$} and $P$ is a
    polynomial functor on a finite signature. Indeed, given $\beta\colon H
    \to B$ take a strong quotient $q\colon F^P \twoheadrightarrow B$ in
    $\Mndf(\Set)$ and the monad morphism $m\colon F^H \to B$ induced by
    $\beta$. Observing that $F^{H+P}$ is the coproduct of $F^H$ and
    $F^P$ in $\Mndf(\Set)$ and that copairing the strong epi $q$ with
    $m$ yields a strong epi again, we obtain the desired strong
    quotient $[m,q]\colon (F^{H+P}, \kappa \cdot \inl) \to (B,\beta)$.
  \item Recall from~\autoref{E:proj}\ref{E:proj:2} that the polynomial endofunctors on
    $\Set$ are projective.  
  \end{enumerate}
\end{remark}
We obtain the following variation of \autoref{lfpquotient}:
\begin{lemma}
    \label{EQquotient}
    Every $\Hf$-coalgebra $b\colon (B,\beta) \to \Hf(B,\beta)$, with $B$ fg, is the
    strong quotient of a coalgebra from $\EQ$.
\end{lemma}
\begin{proof}
  By \autoref{rem:quotmnd}, $(B,\beta)$ is the strong quotient of
  $(F^{H+P}, \kappa\cdot \inl)$, where $P$ a polynomial functor associated to a
  finite signature and therefore a projective object in
  $\Funf(\C)$.
  The following morphism in $H/\Mndf(\Set)$ 
    \[
        \begin{tikzcd}
            E^P =
            (F^{H+P}, \kappa\cdot \inl)
            \arrow[->>]{r}{q}
            &
            (B, \beta)
            \arrow[->]{r}{b}
            &
            \H (B, \beta)
        \end{tikzcd}
    \]
    corresponds to a natural transformation
    $\overline{b\cdot q}\colon P \to HB+\Id$ (using that $E^P$ is the
    free Eilenberg-Moore algebra on $P$). Since $P$ is projective and,
    by \autoref{quasiEpiPreservation}, $\Hf q$ is epic as a natural
    transformation, we obtain a natural transformation
    $p\colon P \to HF^{H+P}+\Id$ such that the triangle on the left below
    commutes; equivalently, the square on the right below commutes
    using again the universal property of $E^P$ as a free
    Eilenberg-Moore algebra:
    \[
        \begin{tikzcd}
            P \arrow{dr}[below left]{\overline{b\cdot q}}
            \arrow[dashed]{r}{p}
            &
            HF^{H+P}+\Id
            \arrow[->>]{d}{Hq+\Id}
            \\
            {}
            &
            HB+ \Id
        \end{tikzcd}
        \quad\Longleftrightarrow
        \begin{tikzcd}
            (F^{H+P},\kappa\cdot \inl)
            \arrow{r}{p^\sharp}
            \arrow[->>]{d}[left]{q}
            &
            \Hf(F^{H+P},\hat\kappa\cdot \inl)
            \arrow{d}{Hq+\Id}
            \\
            (B,\beta)
            \arrow{r}{b}
            &
            \Hf(B,\beta)
        \end{tikzcd}
    \]
    Thus we see that the coalgebra $b$ is the strong quotient of the
    coalgebra $p^\sharp$, which is a coalgebra in $\EQ$.
\end{proof}

It follows from~\autoref{EQquotient} that $\Coalgfg \Hf$ is the same category as $\EQ[2]$; thus their colimits in $\Coalg \Hf$ are
isomorphic and we conclude:
\begin{corollary} \label{lffAlgTree}
  Let $H_\Sigma\colon \Set \to \Set$ be polynomial endofunctor on
  $\Set$. Then the locally finite fixpoint of
  $\Hf\colon H_\Sigma/\Mndf(\Set)\to H_\Sigma/\Mndf(\Set)$ is the monad of
  Courcelle's algebraic trees, mapping a set to the algebraic
  $\Sigma$-trees over it.
\end{corollary}

\section{Conclusions and Future Work}
\label{sec:con}
We have introduced the locally finite fixpoint of a finitary
mono-preserving endofunctor on an lfp category. We proved that this
fixpoint is characterized by two universal properties: it is the final
lfg coalgebra and the initial \fgiterative algebra for the given
endofunctor. Moreover, we have seen many instances where the LFF is
the domain of behaviour of finite-state and finite-equation
systems. In particular, all previously known instances of the rational
fixpoint are also instances of the LFF, and we have obtained a number
of interesting further instances not captured by the rational
fixpoint.

On a more technical level, the LFF solves a problem that sometimes makes the
rational fixpoint hard to apply. The latter identifies behaviourally equivalent
states (i.e.~is a subcoalgebra of the final coalgebra) if the classes of fp and
fg objects coincide. This condition, however, may be false or unknown (and
sometimes non-trivial to establish) in a given lfp category. But the LFF always
identifies behaviourally equivalent states.  

%always exists for finitary
%mono-preserving endofunctors on lfp categories. And we have seen that the LFF is
%an applicable framework to talk about finite-state and finite-equation systems,
%because we inherit all known instances of the rational fixpoint and in addition
%obtain a lot more instances, in order to characterize well-established notions
%by universal properties.

There are a number of interesting topics for further work concerning
the LFF. First, it should be interesting to obtain further instances
of the LFF, e.g.~analyzing the behaviour of tape
machines~\cite{coalgchomsky} might lead to a description of the
recursively enumerable languages by the LFF. Second, syntactic
descriptions of the LFF are of interest.  In works such
as~\cite{brs_lmcs,bbrs_ic,bms13,myersphd} Kleene type theorems and
axiomatizations of the behaviour of finite systems are
studied. Completeness of an axiomatization is then established by
proving that expressions modulo axioms form the rational fixpoint. It
is an interesting question whether the theory of the LFF we presented
here may be of help as a tool for syntactic descriptions and
axiomatizations of further system types.

As we have mentioned already the rational fixpoint is the starting
point for the coalgebraic study of Elgot's iterative~\cite{elgot} and
Bloom and \'Esik's iteration theories~\cite{be}. A similar path could
be followed based on the LFF and this should lead to new coalgebraic
iteration/recursion principles, in particular in instances such as
context-free languages or constructively $S$-algebraic formal power
series.

Another approach to more powerful recursive definition principles are
abstract operational rules (see~\cite{Klin11} for an overview). It has
been shown that certain rule formats define operations on the rational
fixpoint~\cite{bmr12,mbmr13}, and it should be investigated whether a
similar theory can be developed based on the LFF.

Furthermore, in the special setting of Eilenberg-Moore categories, one can
base the study of finite systems on \emph{free} finitely generated
algebras (rather than all fp or all fg algebras). Urbat~\cite{Urbat17}
recently proved that this yields a third fixpoint $\varphi H^T$
besides the rational fixpoint and the LFF, and Milius~\cite{Milius17}
investigated the relationship of the three fixpoints obtaining the
following picture for a lifting $H^T$ on an Eilenberg-Moore category
$\Set^T$:
\[
  \varphi H^T \twoheadrightarrow \varrho H^T \twoheadrightarrow
  \vartheta H^T \rightarrowtail \nu H^T.
\]
In addition, \emph{op.~cit.}~establishes sufficient conditions when the
three fixpoints on the left are isomorphic, i.e.~$\varphi H^T \cong
\varrho H^T \cong \vartheta H^T$. This is the most desired situation
where the rational fixpoint is fully abstract and determined by the
$H^T$-coalgebras $TX$, where $X$ is a finite set, i.e.~precisely the
targets of the generalized  powerset construction. 

Finally, the parallelism in the technical development between rational
fixpoint and LFF indicates that there should be a general theory that
is parametric in a class of ``finite objects'' and that allows to
obtain results about the rational fixpoint, the LFF and other possible
``finite behaviour domains'' as instances. This has also been studied
by Urbat~\cite{Urbat17}, and he obtains a uniform theory which yields
some results about the four fixpoints above as special instances.

%However, there are still some notions from the Chomsky hierarchy left which
%have not been characterized by universal properties yet. \citet{coalgchomsky}
%characterize tape machines coalgebraically as an automata having two kinds of
%side effects. One is the modification of a tape and the other is a possible
%$\varepsilon$-transition to another state. While the former monad is finitary,
%its extension by $\varepsilon$-transitions is not. It is open how to 
%modify the setting in order to describe their behaviour -- the decidable
%languages -- by a universal property.

%\newpage

%Furthermore, the LFF solves the technical fg vs fp question a category in which
%we consider finite-state coalgebras -- however not by answering it, but by
%providing a uniform framework that does not rely on the answer and still
%collects precisely the finite behaviours.

%
% Bibliography
%
%\section{References}
\bibliographystyle{elsarticle-harv}
\bibliography{refs,all2,delta2}

\takeout{% moved to main text
\begin{appendix}
  \section{Adjunction between monoids and monads}
  \renewcommand\thesection{\Alph{section}}
  For every monoid $(M, *, 1_M)$ we have the free $M$-set
  monad $LM$, i.e.~the object assignment, unit and multiplication of $L$
  are as follows: 
  \[
    LM(X) = M \times X, 
    \quad
    \eta_X : x \mapsto (1_M,x),
    \quad
    \mu_X : (n, (m,x)) \mapsto (n * m, x).
  \]
  This extends to a functor $L: \Mon \to \Mndf(\Set)$. In the reverse
  direction, we have a functor $U: \Mndf(\Set) \to \Mon$ sending a
  monad $(T,\eta,\mu)$ with strength $s$ to the monoid $T1$ with unit
  $\eta_1: 1\to T1$ and multiplication
  \[
    m: T1\times T1 \xrightarrow{s_{1,T1}} T(1\times T1) \xrightarrow{\cong} TT1
    \xrightarrow{\mu_1} T1.
  \]
  \begin{proposition}[label = monoidAdjunction]
    We have an adjoint situation $L\dashv U$ with the unit $\nu$ and
    counit $\epsilon$ as follows:
    \begin{align*}
      \nu_M\colon& M\xrightarrow{\ \cong\ } M\times 1 = ULM
                   \\
      \epsilon_T\colon & LUT = T1\times (-)
         \xrightarrow{s_{1,-}}
         T(1\times (-)) \xrightarrow{\ T\cong\ } T
    \end{align*}
    where $s$ is the strength of $T$.
  \end{proposition}
  \begin{proof}
    It is not hard to see that $\nu_M$ is a monoid homomorphism, because the
    monoid structure in $M\times 1 = ULM$ boils down to the monoid structure of
    $M$. Furthermore, $\nu_M$ is natural in $M$. 

    For every monad $T$, $\epsilon_T$
    is a natural  transformation $T1\times (-) \to T$ because the strength $s$
    is. The axioms of $s$ being the strength of the monad $T$ prove that
    $s_{1,-}: T1\times (-) \to T(1\times (-))$ is a monad morphism by
    straightforward diagram chasing. To see $\nu$ and $\epsilon$ establish an
    adjuction, it remains to check the triangle identities:
    \begin{itemize}
      \item The first identity $\epsilon_{LM}\cdot L\nu_M = \id_{FM}$ is just
        the associativity of the product:
        \[
        \begin{tikzcd}[row sep = 1 mm]
          LM
          \arrow{r}{L\nu_M}
          & LULM \arrow{rr}{\epsilon_{LM}}
          &
          & LM
          \\
          M\times (-)
          \arrow{r}{\nu_M\times (-)}
          & (M\times 1) \times (-)
          \arrow{r}{\cong}
          & M\times (1\times (-))
          \arrow{r}{M\times \cong}
          & M \times (-)
        \end{tikzcd}
        \]
        The composition is obviously just the identity on $M\times (-)$.

      \item The second identity $U\epsilon_T\cdot \nu_{UT} = \id_{UT}$
        comes directly from one of the axioms of the tensorial strength $s$ of $T$:
        \[
        \begin{tikzcd}[row sep = 4 mm, baseline=(T1.base)]
          UT
          \arrow{r}{\nu_{UT}}
          & ULUT \arrow{rr}{U\epsilon_{T}}
          &
          & UT
          \\
          |[alias=T1]|
          T1
          \arrow[shiftarr={yshift={-7mm}}]{rr}{T\nu_1}
          \arrow{r}{\nu_{T1}}
          & T1 \times 1
          \arrow{r}{s_{1,1}}
          & T (1\times 1)
          \arrow{r}{\cong}
          & T 1
        \end{tikzcd}
        \qedhere
        \]

      \end{itemize}
  \end{proof}
\end{appendix}
}% end takeout
\end{document}